\documentclass[11pt]{article}

\usepackage[margin=1.125in]{geometry}
\usepackage[table]{xcolor}
\definecolor{DarkGreen}{rgb}{0.1,0.5,0.1}
\definecolor{DarkRed}{rgb}{0.5,0.1,0.1}
\definecolor{DarkBlue}{rgb}{0.1,0.1,0.5}

\usepackage[small]{caption}
\usepackage[pdftex]{hyperref}
\hypersetup{
    unicode=false,          
    pdftoolbar=true,        
    pdfmenubar=true,        
    pdffitwindow=false,      
    pdfnewwindow=true,      
    colorlinks=true,       
    linkcolor=DarkBlue,          
    citecolor=DarkGreen,        
    filecolor=DarkGreen,      
    urlcolor=DarkBlue,          
    pdftitle={},
    pdfauthor={},    
}
\usepackage{amssymb,amsmath,amsthm,amsfonts,enumitem}
\usepackage{fullpage,nicefrac,comment}
\usepackage{tikz}
\usetikzlibrary{arrows,decorations.pathmorphing,decorations.shapes,snakes,patterns,positioning, shapes.callouts, shapes.arrows}
\usepackage[noend]{algpseudocode}
\usepackage{algorithm}
\usepackage{thmtools}
\usepackage{thm-restate}
\usepackage{fontawesome5}

\makeatletter
\AddToHook{env/algorithmic/before}{\def\@currentcounter{ALG@line}}
\makeatother

\usepackage{cleveref}
\crefalias{ALG@line}{line}
\crefname{line}{line}{lines}
\Crefname{line}{Line}{Lines}



\makeatletter
\providecommand{\theHALG@line}{}
\renewcommand{\theHALG@line}{\theHalgorithm.\arabic{ALG@line}}
\makeatother

\newcommand{\cA}{\ensuremath{\mathcal{A}}}

\newcommand{\cE}{\ensuremath{\mathcal{E}}}
\newcommand{\cS}{\ensuremath{\mathcal{S}}}

\newcommand{\cD}{\ensuremath{\mathcal{D}}}

\newcommand{\F}{{\mathbb F}}

\newcommand{\cT}{\mathcal{T}}

\newcommand{\cL}{\mathcal{L}}
\newcommand{\cI}{\mathcal{I}}
\newcommand{\ctL}{\Tilde{\mathcal{L}}}

\newcommand{\EE}{\mathbb{E}}

\newcommand{\ADLLR}{\ensuremath{\mathrm{DALLRC}}}
\newcommand{\DLLR}{\ensuremath{\mathrm{DLLRC}}}
\newcommand{\DLCC}{\ensuremath{\mathrm{DLCC}}}

\newcommand{\inabs}[1]{\left|#1\right|}

\newcommand{\inset}[1]{\left\{#1\right\}}

\newcommand{\inparen}[1]{\left(#1\right)}
\newcommand{\inbrak}[1]{\left[#1\right]}

\newcommand{\argmin}{\mathrm{argmin}}
\newcommand{\polylog}{\mathrm{polylog}}
\newcommand{\poly}{\mathrm{poly}}

\newcommand{\dist}{\delta}

\newcommand{\eps}{\varepsilon}
\renewcommand{\epsilon}{\varepsilon}

\newtheorem{theorem}{Theorem}[section] 
\newtheorem{lemma}[theorem]{Lemma} 
\newtheorem{definition}[theorem]{Definition}

\newtheorem{corollary}[theorem]{Corollary} 

\newtheorem{remark}{Remark}

\newtheorem{claim}[theorem]{Claim}

\newtheorem{proposition}[theorem]{Proposition}

\newcommand{\rep}{\mathrm{Rep}}

\algrenewcommand\algorithmiccomment[1]{\hfill\textcolor{blue!50!black!50}{\small $\triangleright$\ #1}}

\newcommand{\Cot}{C^{\otimes t}}

\newcommand{\DALLR}[1]{\ensuremath{\mathcal{P}^{(#1)}}}
\newcommand{\cApre}{\mathcal{P}}
\newcommand{\cP}{\mathcal{P}}

\newcommand{\Cstar}{C^{\star}}

\newcommand{\Tone}{\ensuremath{\textsc{Test}_1}}
\newcommand{\Ttwo}{\ensuremath{\textsc{Test}_2}}
\newcommand{\Tthree}{\ensuremath{\textsc{Test}_3}}

\global\long\def\inn{in}
\global\long\def\out{out}
\newcommand{\Sp}{\mathrm{Sp}}
\newcommand{\Tm}{\mathrm{T}}

\newcommand{\Spop}{\mathrm{Len}^{(Out)}}
\newcommand{\Spcon}{\mathrm{Sp}^{(Pre)}}
\newcommand{\Tmcon}{\mathrm{T}^{(Pre)}}
\newcommand{\Sppre}{\mathrm{Sp}^{(Pre)}}
\newcommand{\Tmpre}{\mathrm{T}^{(Pre)}}
\newcommand{\Speval}{\mathrm{Sp}^{(Eval)}}
\newcommand{\Tmeval}{\mathrm{T}^{(Eval)}}
\newcommand{\Testt}{\ensuremath{\textsc{Test}^{(t)}}}

\newcommand{\TTestt}{\ensuremath{\Tm_{\Testt}}}
\newcommand{\Tmdectwo}{\ensuremath{\Tm_{\mathrm{Dec}_{C \otimes C}}}}
\newcommand{\Tmdetect}{\ensuremath{\Tm_{\textsc{Detect}_C}}}
\newcommand{\Spdectwo}{\ensuremath{\Sp_{\mathrm{Dec}_{C \otimes C}}}}
\newcommand{\Spdetect}{\ensuremath{\Sp_{\textsc{Detect}_C}}}

\title{Time- and Space-Efficient List Decoding up to Capacity}

\author{Dorsa Fathollahi\thanks{Stanford University, \ \texttt{dorsafth@stanford.edu}.  Work partially supported by NSF grants CNS-2321489 and CFF-2231157.}, Noga Ron-Zewi\thanks{University of Haifa, \ \texttt{nronzewi@ds.haifa.ac.il}. Work partially supported by  BSF grant 2021683, and by the European Union (ERC, ECCC, 101076663). Views and opinions expressed are however those of the author(s) only and do not necessarily reflect those of the European Union or the European Research Council. Neither the European Union nor the granting authority can be held responsible for them.}, and Mary Wootters\thanks{Stanford University, \texttt{marykw@stanford.edu}. Work partially supported by NSF grants CNS-2321489 and CFF-2231157.}. }

\begin{document}
    \maketitle

\begin{abstract}
    In the theory of error correcting codes, \emph{list-decoding} refers to the following problem.  Given a code $C \subseteq \Sigma^N$ and a received word $y \in \Sigma^N$, find all codewords $c \in C$ so that $\delta(c,y) \leq \rho$, where $\delta$ is relative Hamming distance and $\rho \in (0,1)$.  Codes that approach the optimal trade-off between the \emph{rate} $R := \frac{\log_{|\Sigma|}(|C|)} N$ and the \emph{list-decoding radius} $\rho$ are said to \emph{achieve capacity}.
    By now, there are constructions of capacity-achieving list-decodable codes with fast near-linear-time list-decoding algorithms, but most existing work has not considered \emph{space} complexity.  
    
    In a recent line of work, Cook and Moshkovitz~\cite{CM25_enc,CM25,CM26} initiated the study of \emph{low-space} deterministic algorithms for error correcting codes.  In particular, in~\cite{CM26}, they gave a construction of list-decodable codes with deterministic near-linear-time and sublinear space list-decoding algorithms.  However, these codes were far from achieving capacity.

    In this paper, we present list-decodable codes approaching capacity with deterministic time- and space-efficient list-decoding algorithms.  More precisely, for any $R \in (0,1)$ and any arbitrarily small constant $\tau > 0$, we present a family of codes $C\subseteq \Sigma^N$ with rate $R$ that are deterministically list-decodable up to radius $\rho = 1 - R - \tau$, in time $N^{1 + \tau}$ and space $N^{\tau}$ with constant output list size and constant alphabet size.  Our results can be extended to capacity-achieving list-\emph{recoverable} codes.

    Our main tool is a \emph{deterministic} version of a locally list-recoverable code.  It is common wisdom that a local decoding algorithm cannot be deterministic.  However, we observe that if the algorithm is allowed to do some non-local (but still small time  and space) pre-processing,  ``deterministic local list-recovery'' is possible.  Moreover, such codes behave nicely under black-box transformations that are common in coding theory, like composition, intersection, concatenation, the AEL distance-amplification procedure, and tensoring.  We take advantage of this, and the main ingredient in our construction is a \emph{high-rate} deterministic locally list-recoverable code (with sublinear-space, near-linear-time global pre-processing).    Similar notions of deterministic local decoding were implicit in prior work \cite{CM25,CM26};  we hope that making this notion explicit will be useful for future applications.
\end{abstract}
\newpage 
\small
\tableofcontents
\normalsize
\newpage 
\section{Introduction}\label{sec:intro}
An \emph{error correcting code} $C$ is a set $C \subseteq \Sigma^N$, where $\Sigma$ is a finite \emph{alphabet} and $N > 0$ is the \emph{block length}.  The goal is to be able to recover a \emph{codeword} $c \in C$, given a corrupted version of $c$.  In \emph{list-decoding}, the level of corruption may be very large, and to compensate a decoder is allowed to return a short list of possible codewords.  Formally, we have the following definition.
\begin{definition}[List Decoding]
    Let $C\subseteq \Sigma^N$.  Let $\rho > 0$ and let $L$ be a positive integer.  We say that $C$ is \emph{$(\rho, L)$-list-decodable} if for any $y \in \Sigma^N$, $|\{ c \in C\,:\, \delta(c,y) \leq \rho\}| \leq L$.  
\end{definition}
That is, $C$ is list-decodable if, given $y \in \Sigma^N$ that is sufficiently close to some codeword $c \in C$, it is possible to recover a short list of codewords with the guarantee that $c$ is in the list.  The algorithmic goal is to recover this list efficiently, given $y$. 

When $\rho < \frac{\delta(C)} 2$ is less than half the minimum distance of the code,\footnote{Here and throughout the paper, we use $\delta(C) = \min_{c \neq c' \in C}\delta(c,c')$ to denote the \emph{minimum (relative) distance} of~$C$, where $\delta(c,c') = \frac 
{\{i \in [N] \mid c_i \neq c'_i\} } N$ is the \emph{(relative)  Hamming distance} between $c$ and $c'$.}  $C$ is always list-decodable up to radius $\rho$ with list size $L = 1$; this is called \emph{unique decoding}.  However, when $\rho > \delta(C)/2$, (worst-case) unique decoding is impossible, and we have list-decoding with $L > 1$.  List-decoding (and its generalization, list-\emph{recovery}\footnote{List-recovery is a generalization of list-decoding, where instead of an input vector $y$, the input is a list of sets $\cS = (S_1, \ldots, S_N) \in {\Sigma \choose \leq \ell}^N$, where each $S_i \subseteq \Sigma$, with $|S_i| \leq \ell$.  The goal is to return all of the codewords $c \in C$ so that $c_i \in S_i$ for all but a $\rho$-fraction of $i \in [N]$.  List-decoding is the special case when $\ell = 1$.  See \Cref{def:LR} for a formal definition.  Our results for list-decoding extend to list-recovery as well.}) 
have become useful primitives not just in error correcting codes, but throughout theoretical computer science; see for example~\cite{madhu_survey,venkat_survey,salil_survey,nic_venki_survey} for surveys on list-decodable and list-recoverable codes and some of their applications. 

By now, we have constructions of near-optimal list-decodable codes with fast algorithms.  For example, \cite{HRW19,KRSW18,KRRSS20,GHKS24,GHKS25,ST25}
give near-linear-time algorithms for list-decoding (or list-recovering) particular codes up to \emph{capacity}.  This means that they achieve the optimal trade-off between the \emph{radius} $\rho$ and the \emph{rate} $R := \frac{1}{N} \log_{|\Sigma|}(|C|)$ of the code.    The rate $R$ measures how much redundancy the code adds; it is a number between $0$ and $1$, and we want it to be as close to $1$ as possible.  The radius $\rho$ captures how many errors the code can correct; it is also between $0$ and $1$, and we also want it to be as close to $1$ as possible.  Capacity-achieving list-decodable codes are list-decodable up to radius $\rho$ approaching $1 - R$, which is known to be the best possible.

While the results mentioned above obtain fast near-linear time algorithms and near-optimal parameters, they do not focus on space.  Indeed, of the works above, those that are deterministic have near-linear space.  The algorithms in \cite{HRW19, KRSW18} can be implemented with sublinear space, but they are randomized.

\begin{remark}[Determinism and time-efficiency]\label{rem:randomness} 
 We note that the ``deterministic'' requirement is important here.  If we wanted a near-linear time and sublinear space \emph{randomized} list decoding algorithm, then existing 
\emph{local} list decoding algorithms~\cite{HRW19, KRSW18} would work, by correcting one symbol at a time.  However, since it is well-known that local (list) decoding algorithms cannot be deterministic (more on this later), these constructions seemed inherently randomized if we want to preserve a notion of locality.

The ``near-linear-time'' requirement is also important.  If we view a local list-decoding algorithm as a near-linear time global list-decoding algorithm by correcting one symbol at a time, as above, then the straightforward derandomization achieved by enumerating over all possible random seeds would require at least quadratic time, since it can be shown that the number of seeds must be at least linear in the block length $N$.

Thus, the challenge is to obtain low-space \emph{deterministic} list decoding algorithms, with \emph{near-linear-time} algorithms.

This challenge is partially motivated by a desire to understand the role/necessity of randomization in local decoding and in time- and space-efficient algorithms.  Additionally, deterministic constructions are useful as they can be composed more easily and used as building blocks in other pseudorandom constructions. 
Finally, as mentioned above, locally list-decodable codes are useful for several applications in theoretical computer science, in areas such as pseudorandomness and cryptography, where deterministic and low-space algorithms may be desirable.
\end{remark}

The problem of space-efficient deterministic algorithms for error-correcting codes was introduced recently by Cook and Moshkovitz in \cite{CM25_enc,CM25}. 
In \cite{CM25_enc}, the authors gave a space-efficient and near-linear time deterministic \emph{encoding} algorithm for error-correcting codes, while in \cite{CM25}, the authors gave such an algorithm for \emph{unique decoding}. 
In a follow-up work~\cite{CM26}, Cook and Moshkovitz adapted the techniques from \cite{CM25} to obtain a deterministic list decoding algorithm with space $N^{\tau}$ and time $N^{1 + \tau}$ for any small constant $\tau >0$.  This algorithm works up to radius $\rho = 1 -\gamma$ for any small constant $\gamma$.  However, the guarantees are far from capacity, meaning that the rate is asymptotically much less than $\gamma$.\footnote{More precisely, the rate of these codes is $(\gamma \tau)^{O(1/\tau^3)}$.}

In this work, we present deterministic list-decoding algorithms with space $N^{\tau}$, running time $N^{1 + \tau}$, which can approach capacity for any rate.
Our main result is the following theorem.
\begin{theorem}[Informal, see \Cref{cor:main}]\label{thm:main_informal}
Fix any $R \in (0,1)$, and any sufficiently small constant $\tau > 0$.  For infinitely many values of $N$, there is a code $C \subseteq \Sigma^N$ of rate at least $R$, with $|\Sigma| = O(1)$, so that:
\begin{itemize}
    \item $C$ is $(\rho,L)$-list-decodable for $\rho = 1 - R - \tau$ and $L = O(1)$.
    \item There is a deterministic list-decoding algorithm for $C$ with the above parameters that runs in time $N^{1 + \tau}$ and uses space $N^{\tau}$.
\end{itemize}
\end{theorem}
\begin{remark}[Results also hold for list-recovery]
    As noted above, our results are more general, and imply capacity-achieving \emph{list-recovery}, a generalization of list-decoding.  See \Cref{def:LR} for a definition and \Cref{cor:main} for the full statement.
\end{remark}

Our approach is in fact a version of the approach that we dismissed in \Cref{rem:randomness}.  That is, it is well-known that local (list-)decoding algorithms, while low-space, cannot be deterministic.  However, following prior work~\cite{CM25,CM26}, we observe that if some (sublinear-space, near-linear-time) pre-processing is allowed, then there is such a thing as ``deterministic local (list-)decoding,'' and it can be used to develop deterministic space- and time-efficient global list-decoding algorithms.   We explain this in more detail next in \Cref{sec:tech}.

\subsection{Technical Overview}\label{sec:tech}
As noted above, our results hold for list-recovery as well. Since list-recovery (as opposed to list-decoding) is also required for some steps of our approach, we generalize our discussion to list-recovery in this section. 

List-recovery is a generalization of list-decoding where instead of a corrupted codeword $y \in \Sigma^N$, we receive as input a list of sets $\cS = (S_1, \ldots, S_N)$, where each $S_i \subseteq \Sigma$ and $|S_i| \leq \ell$ for some parameter $\ell \geq 1$.  The goal is to return all $c \in C$ so that $c_i \in S_i$ for all but at most a $\rho$ fraction of indices $i \in [N]$; and in particular we hope that there are not more than $L$ such codewords.  Thus, list-decoding is the special case of list-recovery when $\ell = 1$.

The technical meat of our result is a \emph{deterministic} locally list-recoverable code.  
Before we explain this, we explain what a locally list-recoverable code is, why it cannot be deterministic, and how we get around this. 

\paragraph{Locally List-Recoverable Codes and \emph{Deterministic} Locally List-Recoverable Codes.}
The notion of \emph{locally list-recoverable codes}  was first formally defined in \cite{GKORS18}, following similar notions for list-decoding \cite{GL89,STV01} and unique  decoding \cite{BFLS91,STV01,KT00}. 
Let $C \subseteq \Sigma^N$, and suppose that the input to our list-recovery problem is a list of sets $\cS \in {\Sigma \choose \leq \ell}^N$.  We would like to return the list $\cL = \inset{ c \in C \,:\, \delta(c,\cS) \leq \rho }$,
where $\delta(c,\cS) = \Pr_{i \in [N]}[ c_i \not\in S_i ]$.  A \emph{local list-recovery algorithm} $\cP$  takes input $\mathbf{1}^N$ and outputs a list of \emph{local algorithms} $A_1, \ldots, A_L$.  
Each one of the local algorithms $A_j$ takes an input $i \in [N]$, is allowed to make $Q$ randomized queries to the input $\cS$, and outputs an element of $\Sigma$.  The idea is that each local algorithm $A_j$ should correspond to some $c \in \cL$.  That is, for all $c \in \cL$, there should be some $j \in [L]$ so that for all $i \in [N]$, $\Pr[A_j(i) = c_i] \geq 2/3$, where the probability is over the random choices of $A_j$.  See \Cref{defn:LLR} for a formal definition.

Existing local list-recovery algorithms can generally be turned into \emph{randomized} low-space global list-recovery algorithms.  Indeed, for each $j \in [L]$, we form $A_j$, and run it repeatedly on each $i \in [N]$ until we have each symbol $c_i$ for the codeword $c \in \cL$ that $A_j$ corresponds to (assuming that the decoding error is sufficiently small, which can be guaranteed by repetition).  We can output these symbols one at a time, re-using the space in between symbols.  The fact that $A_j$ makes only a few queries to $\cS$ typically means that it also uses sublinear space, and we are done.

However, as noted in \Cref{rem:randomness}, this approach seems inherently randomized.  We \emph{cannot} have a deterministic local list-recovery algorithm, at least if $L \cdot Q \ll N$ and $\rho = \Omega(1)$.  Indeed, in that case an adversary could design an input $\cS$ that is garbage on the $L\cdot Q \ll N$ coordinates that the algorithm could ever read, without significantly affecting which codewords are at most $\rho$-far from $\cS$.  Thus, such an algorithm could only return garbage.

To get around this, we modify the definition of a local list-recovery algorithm.  Instead of the algorithm $\cP$ taking as input $\mathbf{1}^N$, we now give it (global) access to the input $\cS \in {\Sigma\choose \leq \ell}^N$, and allow it to do some time- and space-efficient pre-processing.  (This explains why we call this algorithm ``$\cP$'': It stands for ``$\cP$re-$\cP$rocessing.'')  Now, the obstacle to determinism is gone.  
We call our modified definition a \emph{Deterministic Locally List-Recoverable Code} (DLLRC).  That is, a DLLRC is a code so that there is a time- and space-efficient deterministic algorithm $\cP$ with access to the input $\cS \in {\Sigma \choose \leq \ell}^N$; $\cP$ outputs a list of \emph{deterministic} local algorithms $A_1, \ldots, A_L$ so that, for all $c \in C$ with $\delta(c,\cS) \leq \rho$, there is some $j \in [L]$ so that $A_j(i) = c_i$ for \emph{all} $i \in [N]$.    See \Cref{def:DALLR} for the formal definition. 

As explained above, if we could construct a DLLRC, we could use it as a deterministic time- and space-efficient list-recovery algorithm.  Thus, we set our sights on constructing a (high-rate) DLLRC.  (We will see in a moment why we want it to be high-rate; briefly, it is so that we can amplify the list-recovery radius to achieve capacity later.)
We note that this approach is similar to what \cite{CM25} did for unique decoding: 
The special case of a DLLRC with $\ell=L=1$, which  we call a ``Deterministic Locally Correctable Code'' (DLCC) in \Cref{def:DLC}, is implicit in their work, and their main constructions are in fact DLCCs.  Similarly, a version for list-decoding is implicit in \cite{CM26}.

\begin{remark}[The value of DLLRCs] 
    As noted above, a DLLRC naturally gives rise to our goal of a deterministic time- and space-efficient list-recovery (or list-decoding) algorithm. 
    Further, we believe that the notion of DLLRC is interesting on its own.  While similar notions are already implicit in prior work \cite{CM25,CM26}, we hope that making this notion explicit will be valuable.  In particular, as we will make use of extensively in this paper, DLLRCs  behave nicely under black-box transformations that are common in coding theory, like composition, intersection, concatenation, the AEL distance amplification procedure, tensoring, and so on.  If we were to treat our DLLRC constructions as only time- and space-efficient list-recovery algorithms with no additional structure, we would not be able to manipulate them in a black-box way.
\end{remark}

\paragraph{Constructing a high-rate D(A)LLRC: Tensor Codes.}
In order to construct our high-rate DLLRC, we start with a Deterministic \emph{Approximate} Locally List-Recoverable Code (DALLRC); see \Cref{def:DALLR}.\footnote{The authors pronounce ``DALLRCs'' as ``Dollar Codes,'' as in the currency, although the reader is of course free to choose their own pronunciation.} 
This is the same as a DLLRC, except that we relax the guarantee that a local algorithm should satisfy $A_j(i) = c_i$ for all $i \in [N]$.  Instead, we ask that $A_j(i) = c_i$ for \emph{most} $i \in [N]$, specifically, for at least a $1 - \eps$ fraction.  This mirrors the notion of (randomized) approximate local list-recovery that was first formally defined and used in~\cite{HRW19}.

To construct our DALLRC, we follow the recipe of \cite{HRW19}, who used \emph{tensor codes} to construct the randomized analog.  For a code $C \subseteq \Sigma^n$, the \emph{tensor code} $C \otimes C \subseteq \Sigma^{n \times n}$ is the set of $n \times n$ matrices whose rows and columns all lie in $C$.  The $t$'th order tensor code $\Cot \subseteq \Sigma^{[n]^t}$ is defined similarly: It is the set of all $t$-dimensional $n \times n \times \cdots \times n$ tensors all of whose axis-parallel lines lie in $C$.  See \Cref{subsec:tensor} for a formal definition.

To explain our DALLR algorithm for $\Cot$, we assume that we have a DALLR algorithm $\cP'$ for $C^{\otimes (t-1)}$, and give an overview of what the algorithm $\cP$ for $\Cot$ looks like.  We view an input $\cS \in {\Sigma \choose \leq \ell}^{[n]^t}$ as a matrix of input lists with $n^{t-1}$ rows and $n$ columns.  Thus, our goal is to return a list of (local algorithms approximately representing) codewords in $\Sigma^{[n]^t} = \Sigma^{[n]^{t-1} \times [n]}$ so that every row is a codeword $c \in C$, and every column is a codeword in $C^{\otimes (t-1)}$.

Before we explain what our algorithm does, we explain the randomized approach of \cite{HRW19}, which was inspired in turn by a global randomized list-decoding algorithm for tensor codes of \cite{GGR11}.  There, the idea is the following.  We first run $\cP'$ on a random subset $H = \inset{h_1, \ldots, h_m} \subseteq [n]$ of $m$ columns of the input $\cS$.  This gives us, for each $a \in [m]$,  a list $\cL_a$ of local algorithms for $C^{\otimes (t-1)}$.  Then for each possible combination of these local algorithms $$\vec{A} = (A'_1, \ldots, A'_m) \in \cL_1 \times \cL_2 \times \cdots \times \cL_m,$$
we define a new local algorithm $A$ to add to the new output list.  

The local algorithm $A$ defined by $\vec{A}$ does the following.  On input $i = (i', i_t) \in [n]^{t-1} \times [n]$, it runs $A'_a(i')$ for each $a \in [m]$ to obtain a guess $v_a$ for the $(i',h_a)$-th coordinate of a codeword in $\Cot$.  Then it runs a global list-recovery algorithm for $C$ on the $i'$-th row of $\cS$ to obtain a list $\mathcal{K}$ of codewords in $C$ that are close to $\cS|_{\{i'\} \times [n]}$.  Then, it searches $\mathcal{K}$ for the codeword $c^* \in C$ that agrees the most with the $m$ guesses $v_a$ for $a \in [m]$ that it has already produced, and returns the $i_t$ symbol of $c^*$, namely $c^*_{i_t}$, as its guess for the $i = (i', i_t)$-th coordinate.

We have just described the randomized algorithm from \cite{HRW19}.  Our goal is to derandomize this algorithm, while keeping its time and space complexity small.  The first step is to derandomize the choice of the set $H$.  This is done via a \emph{sampler}, a standard pseudorandom object (see \Cref{def:samplerGamma}).  Informally, a sampler $\Gamma$ takes a short seed $\sigma \in \{0,1\}^r$ and returns a pseudo-random set $H \subseteq [n]$, with the guarantee that for most seeds $\sigma$, $H$ behaves like a random set.  Our deterministic algorithm $\cP$ enumerates over all such seeds $\sigma\in \{0,1\}^r$.  For each such seed, it recursively runs the pre-processing algorithm $\cP'$ for $C^{\otimes (t-1)}$ on the columns in $H = \Gamma(\sigma)$, to obtain a list $\cL$ of local algorithms.  Now, for each seed $\sigma$, and each combination of local algorithms 
\[ \vec{A} = (A_1',\ldots, A_m') \in \cL_1 \times \cdots \times \cL_m,\]
$\cP$ outputs a local algorithm $A$, which is parameterized by $(\sigma, \vec{A})$.  
(We note that each of the level-$(t-1)$ local algorithms in $\vec{A}$ itself has this form, defined recursively; so at the end of the day the local algorithm $A$ is essentially represented by a large collection of seeds).

Now, the local algorithm $A$ represented by $(\sigma, \vec{A})$ acts similarly to before.  It runs $\Gamma$ on its seed $\sigma$ to recover the set $H$, then runs the level-$(t-1)$ local algorithms in its $\vec{A}$ to get guesses $v_a$ for the $(i',h_a)$-entry for any $a \in [m]$.  Then it runs a global list-recovery algorithm for $C$ on the $i'$-th row of $\cS$ to obtain a list $\mathcal{K} \subseteq C$, chooses the codeword $c^* \in \mathcal{K}$ that agrees the most with its guesses, and returns $c^*_{i_t}$.

So far, we have derandomized the approach of \cite{HRW19}; this is similar to the approach taken by \cite{KRRSS20}, who gave a high-space derandomization of that algorithm using similar techniques.  However, there is a problem, which is that, as stated, the list size will become very large, too large to efficiently deal with. 

\footnote{In more detail, naively we would have a list size $L_t$ at the $t$-th level that grows like $L_t = L_{t-1}^m 2^r$.  However, in order for the sampler to work, the parameter $r$ will need to be logarithmic in 
$n$, which means that $L_t$ will be polynomial in $N=n^t$; and we will have to iterate over all $\poly(N)$ seeds.  We cannot afford this, so we need to prune down the list to keep it constant-sized.  We note that \cite{HRW19} did not have to deal with this problem, as they did not need to exhaust over a list of possible seeds.  The work \cite{KRRSS20} that derandomized \cite{HRW19} did have to deal with this problem in order to obtain near-linear running time, but they resolved it in a high-space way; see the discussion in \Cref{sec:related}.} 
Fortunately, the true list size is not so large; rather, the procedure described above may return many spurious local algorithms $A$.  That is, it may produce local algorithms $A$ that either do not correspond to any codeword in the true list; or local algorithms $A$ that correspond to the same codeword as an algorithm already in the output list.  In order to keep the list size small, we need to prune out these spurious local algorithms.  

Such a pruning step exists in many (approximate) list-recovery algorithms.  The three things we need to check are (1) that the output of the local algorithm $A$ is actually $\eps$-close to some codeword (that is, $A(i) = c_i$ for at least a $1-\eps$ fraction of the $i \in [n]^t$, for some $c \in \Cot$); (2) that this codeword $c \in \Cot$ is actually at most $\rho$-far from the input $\cS$; and (3) that there is no other local algorithm $\tilde{A}$ already in the output list that is also $\eps$-close to this $c$.  Typically these are straightforward to check.  (1) can be checked by running a unique decoding algorithm; (2) can be checked by direct comparison after running that decoding algorithm; and (3) can be checked by pairwise comparisons within the output list.  However, by default none of these approaches are low-space.  To get around this, we devise three tests, which we call $\Tone, \Ttwo, \Tthree$ (see  \Cref{subsec:test} for a formal definition), to (approximately) test each of the three things.  These tests leverage the fact that the local algorithms $A$ are indeed local, so we can simulate queries to a candidate string represented by $A$ in low space.  Combining this with techniques for (randomized) local testing of tensor codes from \cite{Vid2015}, we are able to implement suitable deterministic approximations to each of these three tests, in a low-space way. 

This completes the description of our pre-processing algorithm $\cP$.  It iterates over all possible seeds $\sigma \in \{0,1\}^r$; for each $\sigma$, it runs $\cP'$ on the columns in $H = \Gamma(\sigma)$ to obtain a list $\cL_a$ for each $a \in [m]$; then for each $\vec{A} \in \cL_1 \times \cdots \times \cL_m$, it constructs a representation of a local algorithm $A = (\sigma, \vec{A})$.  Now, for each candidate local algorithm, it runs the above suite of tests.  If the local algorithm passes all the tests, it is added to the output list.  The pseudocode for $\cP$ can be found in \Cref{alg:main}.  The pseudocode for the local algorithms $A$ that it outputs can be found in \Cref{alg:local-alg}.

Finally, we observe that $\cP$ satisfies the time and space requirements.  Indeed, it turns out that the space needed is dominated by the amount of space required at the bottom of the recursion, which is essentially a global list-recovery algorithm for $C \subseteq \Sigma^n$ that runs in time and space $\poly(n)$. Thus, the space ends up being dominated by $\poly(n) = N^{O(1/t)}$, where $N = n^t$ is the length of $\Cot$.  By choosing $t$ to be a large constant,  the space is bounded by $N^\tau$ for an arbitrarily small constant $\tau$.  A similar calculation shows that the time is bounded by $N^{1 + \tau}$ for an arbitrarily small constant $\tau > 0$.\footnote{%
We believe that,
similarly to prior work \cite{HRW19,KRRSS20}, one can obtain time $N^{1 + o(1)}$ and space $N^{o(1)}$ by choosing $t$ to be very slowly growing, at the cost of super-constant (though still non-trivial) output list size. For simplicity, we elected to present the result with time $N^{1 + \tau}$ and space $N^\tau$ for an arbitrarily small constant $\tau>0$, as was done in \cite{CM26}. 
}  
We prove the correctness and time/space bounds for our DALLRC in \Cref{thm:mainLocal}.

It remains to instantiate the DALLRC by choosing a suitable base code $C$. 
We need a high-rate linear code with an efficient (polynomial-time) global list-recovery algorithm.  Fortunately, we are able to find such a code (nearly) off-the-shelf from \cite{ST25}.  This code is actually linear over a smaller field, not over its alphabet, but we remedy this by concatenating with a constant-length high-rate linear code that is combinatorially list-recoverable; as it is constant-length, we list-recover it efficiently by brute force.  The resulting DALLRC is presented in \Cref{thm:instantiate_DALLR}.  We note that since $C$ is sufficiently high-rate, then $\Cot$ is also high-rate.  We will see next why we need this.

\paragraph{From high-rate DALLRCs to high-rate DLLRCs to capacity-achieving deterministic low-space list-recovery.}
Above, we explained how to construct a high-rate deterministic \emph{approximate} locally list-recoverable code (DALLRC).  The approximation is needed because at each step of tensoring, there could be some rows that the adversary completely corrupts, which means that we cannot hope to recover 100\% of the entries with our algorithm.  To obtain our final results, there are two more steps.  First, we transform our DALLRC to a DLLRC, to remove the approximation factor.  Second, we amplify the list-recovery radius of this DLLRC to get a capacity-achieving DLLRC.  We discuss both of these steps below.
\vspace{.2cm}

First, we discuss how to turn our DALLRC into a DLLRC.  To do this, we will intersect our DALLRC with a \emph{Deterministic Locally Correctable Code} (DLCC).  A DLCC $C \subseteq \Sigma^N$ can be viewed as the special case of a DLLRC when $\ell = L = 1$.  That is, there is a deterministic low time and space pre-processing algorithm $\cP$ that has access to an input $y \in \Sigma^N$, and outputs a \emph{single} local algorithm $A$ that takes an input $i \in [N]$, makes at most $Q$ queries to $y$, and outputs a guess for $c_i$, where $c$ is the unique codeword sufficiently close to $y$.  See \Cref{def:DLC} for a formal definition.

Now, if we have a  DALLRC $C_1$ with algorithm $\cP_1$ and a DLCC $C_2$ with algorithm $\cP_2$, consider the code $C = C_1 \cap C_2$.  If both $C_1$ and $C_2$ are linear codes of high rate, then $C$ will have high rate as well.  Moreover, we can use the properties of $C_2$ to correct the $\eps$ fraction of errors left over from $\cP_1$.  This results in a high-rate DLLRC, as desired.

For our DLCC, we adapt a result from \cite{CM25}.  In that paper, they show that \emph{Lifted Reed-Solomon Codes}~\cite{GKS13} are in fact DLCCs with near-linear time and sublinear space (although they do not use that language).  We adapt the approach of \cite{CM25} for our purposes in \Cref{thm:DLCC}; the only change from \cite{CM25} is that their stated result only guarantees a code of rate at least $1/2$, while we need a code with rate arbitrarily close to $1$, so we work out the parameters to obtain that.\footnote{We note that \cite{CM25} contains multiple constructions, some with uniform decoding algorithms and some with non-uniform decoding algorithms.  Their algorithm for lifted RS codes is uniform, which is why we chose it.  In this paper, all of our algorithms will be uniform.}

\vspace{.2cm}
At this point, we have a high-rate DLLRC, and our final step is to amplify the radius to obtain a capacity-achieving DLLRC.  For this, we can use a standard technique due to Alon, Edmonds and Luby~\cite{AEL95} (AEL).  This is a distance-amplification technique based on bipartite expander graphs, and it has been used in the past to turn high-rate (randomized) locally list-recoverable codes into capacity-achieving locally list-recoverable codes (e.g.~\cite{GKORS18,HRW19,KRSW18}).  The work \cite{CM25} implicitly showed that the AEL transformation can be applied in the DLCC setting. We show that the AEL transformation can also be applied in the DLLRC setting.  By applying it to our high-rate DLLRC, this results in  capacity-achieving DLLRCs (\Cref{thm:main}), and hence also capacity-achieving list-recoverable (and list-decodable) codes with deterministic near-linear time and sublinear space  list-recovery algorithms, as in \Cref{thm:main_informal}.  The formal statement is presented in \Cref{cor:main}.  We note that the AEL transformation could not have been applied to the low-space list-decodable codes of \cite{CM26}, as those codes have very low-rate, while to achieve capacity, AEL requires a high-rate list-recoverable code to begin with.

\subsection{Related Work}\label{sec:related}

\paragraph{Low-Space Encoding and Decoding.}
The study of low-space (deterministic) algorithms for error correcting codes was initiated by Cook and Moshkovitz in \cite{CM25_enc,CM25}.  In \cite{CM25_enc}, the authors gave  constructions of error-correcting codes that are \emph{encodable} in time $N^{1 + o(1)}$ and space $\polylog(N)$.  Then in \cite{CM25}, they turned their attention to decoding, constructing codes that are uniquely decodable in time $N^{1 + o(1)}$ and space $N^{o(1)}$.  The key idea is that of \emph{improving sets} (see \Cref{def:improving_set}).  Informally, an improving set $\cI$ for a code $C \subseteq \Sigma^N$ is a set of functions $I:\Sigma^N \to \Sigma^N$ so that if a $y \in \Sigma^N$ is close to a codeword $c \in C$, then \emph{on average} over all $I \in \cI$, $I(y)$ is even closer to $c$.  They show that such an improving set (where all of the functions $I$ run in low space) is essentially all one needs for deterministic low-space decoding (see \Cref{lem:imp_means_DLCC}); the idea is to iteratively apply the improvers $I$ to get closer and closer to the correct codeword, using techniques from local testing (it turns out that good testers follow from an improving set $\cI$ as well) to identify which choice of $I \in \mathcal{I}$ is a ``good'' one. 

While our main construction of a DALLRC does not rely on improving sets, we do use the machinery of \cite{CM25} in order to construct the DLCC that we use to intersect with our DALLRC at the end of our argument. This DLCC is essentially the same as the codes constructed in \cite{CM25}, and we follow essentially the same argument; the only difference is that we want our DLCC to have rate arbitrarily close to 1, while \cite{CM25} establishes rate at least $1/2$.  

In a follow-up work, \cite{CM26} applied the same ideas to list-decoding.  Their main result is an explicit family of codes (Reed-Muller codes) of length $N$, which, for any constants $\gamma, \tau > 0$, are $(\rho, L)$-list-decodable for $\rho = 1 - \gamma$ and $L = O(1/\gamma)$ in time $N^{1 + \tau}$ and space $N^\tau$.  However, the rate of these codes is $(\gamma \tau)^{O(1/\tau^3)}$, 
while a capacity-achieving rate would approach $\gamma$.  The work \cite{CM26} also provides results for list-recovery; and by concatenating their codes down to a smaller alphabet size, they can obtain constant or even binary alphabets, at the cost of an even smaller rate.

We note that it seems difficult to use the approach of \cite{CM26} to obtain a high-rate code, as would be needed to apply the AEL transformation to achieve capacity.  One reason is the use of Reed-Muller codes, which themselves do not achieve list-decoding capacity.\footnote{Indeed, the distance of the code should be at least $1 - \gamma$, which for a $m$-variate Reed-Muller code necessitates a degree of $r = \gamma |\F|$, leading to rate at most $(\gamma/m)^m$ or so.  Since the approach of \cite{CM26} takes $m = 1/\tau^3$, the rate of $(\gamma \tau)^{O(1/\tau^3)}$ is roughly as good as it can be for this code.  It is an interesting question whether the rate can be improved by moving to lifted Reed-Solomon codes, as was done in \cite{CM25}.  While this is not clear (to us), we note that \cite{CM26} chooses the degree $r$ of the Reed-Muller code to be small, substantially smaller than $\gamma|\F|$ ($r$ is chosen to be approximately $\frac{\gamma^8\tau^{15}}{10^{20}}\cdot |\F|$).  As lifted Reed-Solomon codes only have non-trivial rate (that is, rate larger than the associated RM code) for $r \geq |\F|/2$~\cite{RS96}, it seems more changes would be needed in the approach of \cite{CM26} to afford high-rate list-decodable  codes.}  Another reason is that, while \cite{CM26} does provide results for list-recovery, the cost is a factor of $1/\poly(\ell)$ in the rate; while to use AEL for list-decoding/recovery, we need list-recoverable codes with rate arbitrarily close to $1$ even for reasonably large $\ell$.  

The approach of \cite{CM26} is similar to that of \cite{CM25}: They adapt the idea of improving sets to \emph{list-improving sets}, and show how to construct these for Reed-Muller codes.  In contrast, our approach goes in a different direction, and we do not use improving sets to construct our DALLRCs.

Despite the difference in approaches, the concept of a DLLRC is implicit in \cite{CM26}.  Indeed, their main guarantee (see \cite[Theorem 34]{CM26}) is the same as that of a (low-rate) Deterministic Locally List-Decodable Code (aka, a DLLRC with $\ell = 1$).

\paragraph{Algorithmic List-Decoding and List-Recovery.}
If we do not restrict our attention to deterministic space-efficient algorithms, there are by now  plenty of time-efficient algorithms for list-decoding and list-recovery approaching capacity.  We have already mentioned \cite{HRW19}, which gave the first near-linear-time algorithm for list-decoding (and list-recovery) approaching capacity, based on tensor codes.  The global result follows from a local list-recovery algorithm, and thus is randomized.
In a similar vein, the work \cite{KRSW18} gave local-list-recovery algorithms for multivariate multiplicity codes, which also leads to near-linear-time randomized global list-decoding algorithms up to capacity.

The global list-decoding algorithm of \cite{HRW19} was later derandomized by \cite{KRRSS20}.  Our approach of derandomizing \cite{HRW19} is similar to theirs, in the sense that we use a good sampler to choose the set $H$ of columns to recurse on.  As mentioned in \Cref{sec:tech}, our approach diverges from \cite{KRRSS20} in how we keep the list size small.  The work \cite{KRRSS20} also needed to prune their list as they were building it, but they did this in a high-space way, by including a step to uniquely decode the approximate codeword and comparing it to the existing output list before deciding to return it.  In contrast, our approach implements a suite of three low-space deterministic testers to accomplish this.

Another line of work has studied deterministic (time-)efficient global list-decoding algorithms for algebraic codes, including Reed-Solomon (RS) codes, Folded RS Codes, and Multiplicity Codes.  The first capacity-achieving codes with efficient (polynomial-time) algorithms were Folded RS codes~\cite{GR08} followed by univariate multiplicity codes~\cite{Kop15,GW13}.  Later, near-linear-time algorithms were given for list-decoding and list-recovery of these codes~\cite{GHKS24,GHKS25}.  As written, these algorithms use high space, and it is not clear (to us) how to implement them with sublinear space while maintaining the near-linear running time.  Indeed, all these algorithms involve interpolating a multivariate polynomial, which involves solving a large linear system that needs to be kept in memory. 

Finally, we mention~\cite{ST25}, who give near-linear-time list-decoding and list-recovery algorithms for graph-based codes up to capacity, without considering space requirements.  The near-linear-time version is randomized, and the best deterministic version (also given in that paper) runs in time $O(N^{3.5})$.  In particular, making that algorithm near-linear-time and deterministic is open, even before space considerations.

\subsection{Discussion and Open Questions}

We have given a near-linear time and sub-linear space algorithm for list-decoding up to capacity, based on tensor codes.  Several interesting open questions remain; we mention three here.

\begin{enumerate}
    \item In \cite{CM25}, Cook and Moshkovitz give a general transformation that turns \emph{any} LCC into a DLLC (possibly with non-uniform algorithms).  Is a similar transformation possible for list-decoding/recovery?  That is, can one turn any locally list-recoverable code into a D(A)LLRC in a black-box way?  We have shown how to do this for tensor codes, but perhaps a more general statement is possible.

\item In particular, high-rate \emph{multivariate multiplicity codes} were shown in \cite{KRSW18} to be locally list recoverable with query complexity $Q=\exp( \log^{\gamma} n)$ for some $\gamma <1$. Is it possible to derandomize these algorithms? If possible, then this would lead to capacity-achieving list-recoverable (and list-decodable) codes with deterministic list recovery algorithms running in near-linear time, and space approximately
$Q$.

    \item In this paper we have made explicit a notion of a ``deterministic local code.'' Are such codes useful for any of the applications of locally (list-)decodable codes in theoretical computer science, for example in pseudorandomness or cryptography?
\end{enumerate}

\subsection{Organization}
In \Cref{sec:prelim} we introduce basic definitions and notation, including standard notions of (local) decoding.  In \Cref{sec:local_defs} we introduce our new definitions of \emph{deterministic} local codes, including DALLRCs and DLCCs, and we prove some useful statements about how they behave under certain black-box transformations.  In \Cref{sec:alg} we introduce our DALLR algorithm for tensor codes.  We analyze this algorithm (without instantiating it) in \Cref{sec:main}.  Finally, in \Cref{sec:DLLR}, we instantiate our theorems to construct our final code: First we choose a base code to instantiate our high-rate DALLRC, then we
instantiate the high-rate DLCC, then we intersect the two to obtain a high-rate  DLLRC, and finally we apply the AEL transformation to obtain a capacity-achieving DLLRC, which yields our final capacity-achieving list-recoverable (and list-decodable) codes with near-linear time  and sublinear space deterministic list-recovery algorithms.

\section{Preliminaries}\label{sec:prelim}

We start with some basic notation and definitions.
Throughout, we use the convention that an empty product (e.g., $\prod_{j=i+1}^i x_j$) is equal to $1$.  For a set $\Sigma$ and an integer $\ell$, we use ${\Sigma \choose \leq \ell}$ to denote the collection of subsets of $\Sigma$ of size at most $\ell$.  For two vectors $x,y \in \Sigma^N$, we use $\delta(x,y)$ to denote the \emph{(relative) Hamming distance} between them:
$\delta(x,y) = \frac{1}{N}|\{ i \in [N] : x_i \neq y_i \}|.$
For a vector $x \in \Sigma^N$ and a list of sets $\cS = (S_1, \ldots, S_N)$ with $S_i \subseteq \Sigma$, we write
\[ \delta(x,\cS) = \frac{1}{N}  |\{i \in [N] : x_i \not\in S_i \}|.\]

\subsection{Error-correcting codes}
A \emph{code} $C$ of \emph{block length} $N$ over an \emph{alphabet} $\Sigma$ is a subset $C \subseteq \Sigma^N$.  The (relative) \emph{distance} of $C$ is defined by
$\delta(C) := \min_{c \neq c' \in C} \delta(c,c').$
The \emph{rate} of the code $R(C)$ (usually denoted $R$ when $C$ is clear from context) is defined by
$R(C) := \frac{ \log_{|\Sigma|}(|C|)}{N}.$ For a finite field $\F$, we say that $C$ is an $\F$-\emph{linear code} if 
$\Sigma$ is a vector space over $\F$, and $C$ is linear over $\F$. In the special case that $\Sigma=\F$, we say that $C$ is \emph{linear}.

An \emph{encoding map} for $C$ is a bijection $E_{C}:\Sigma^{k}\to C$,
where $\left|\Sigma\right|^{k}=|C|$. For an $\F$-linear code, we further require that the
encoding map $E_C$ is linear over $\F$.
We say that an infinite family
of codes $\left\{ C_{n}\right\} _{n}$ is \emph{explicit} if there
is a polynomial time algorithm in $n$ that computes the encoding
maps $E_{C}$ of all the codes in the family.

\paragraph{List recovery.}

Our list decoding algorithms will be obtained through the more general notion of list recovery, defined as follows.

\begin{definition}[List-Recovery]\label{def:LR}
Let $C \subseteq \Sigma^N$ be a code, let $\rho \in [0,1)$ and let $\ell \leq L$ be positive integers.  We say that $C$ is $(\rho, \ell, L)$-list-recoverable if, for any $\cS = (S_1, \ldots, S_N)$ with $S_i \subseteq \Sigma$ and $|S_i| \leq \ell$, we have
\[ |\{ c \in C: \delta(c, \cS) \leq \rho \}| \leq L. \]

A $(\rho, \ell, L)$-\emph{list-recovery algorithm} for $C$ is an algorithm that given as input $\cS = (S_1, \ldots, S_N)$, outputs a list $\cL$ of size at most $L$ so that $\{c \in C : \delta(c, \cS) \leq \rho \} \subseteq \cL$.  
\end{definition}
Note that \emph{unique decoding} corresponds to the special case where $\ell = L=1$, while  \emph{list-decoding} corresponds to the special case where $\ell = 1$.

\paragraph{Local decoding.}
Intuitively, a code is said to be \emph{locally correctable} 
if, given a codeword $c\in C$ that has been corrupted by some errors,
it is possible to decode any coordinate of $c$ by only reading a
small part of the corrupted version of $c$.

\begin{definition}[Locally correctable code (LCC)]\label{def:lcc} 
Let $C\subseteq\Sigma^{N}$ be a code, let $\rho \in [0,1)$ so that $\rho <\frac{\delta(C)} 2$, and let $Q$ be a positive integer. We say that $C$ is a $(Q,\rho)$-\emph{locally correctable code (LCC)} 
if there is  a randomized algorithm $A$ that gets query access to a string $ w\in\Sigma^{N}$ and receives as input a coordinate $i \in [N]$. 
On input $i \in [N]$, $A$ makes at most $Q$ queries to
$w$ and outputs an element of $\Sigma$.  If there is some $c \in C$ so that $\delta(w,c) \leq \rho$, then $\Pr[ A(i) = c_i ] \geq 2/3$. 
\end{definition}

The following definition generalizes the notion of locally correctable codes to the setting of list decoding/recovery.
In this setting, the local list recovery algorithm is required to output  in an implicit sense all  codewords that are consistent with most of the input lists.

\begin{definition}[Locally list recoverable code (LLRC)]\label{defn:LLR} 
Let $C \subseteq \Sigma^N$ be a code,  let $\rho \in [0,1)$, and let $\ell \leq L$ and $Q$ be positive integers.  We say that $C$ is a  $(Q,\rho, \ell, L)$-\emph{locally list recoverable code (LLRC)} 
if there is an algorithm $\cApre$ so that the following holds.
\begin{itemize}
    \item $\cApre$ is a randomized algorithm that, on input $\mathbf{1}^N$, outputs a list of local algorithms $(A_1, \ldots, A_L)$.
    \item Each local algorithm $A_j$ is a randomized algorithm with query access to $\mathcal{S} \in {\Sigma \choose \leq \ell}^{N}$ and receives as input a coordinate $i \in [N]$.  On input $i \in [N]$, $A_j$ makes at most $Q$ queries to $\mathcal{S}$.
    \item For every codeword $c \in C$ so that $\delta(c,\mathcal{S}) \leq \rho$, with probability at least $\frac 2 3$ over the randomness of $\cApre$, there is some $j \in [L]$ so that
 for all $i \in [N]$, 
\begin{equation}\label{eq:local-list-rec}
 \Pr[A_j(i) = c_i] \geq \frac{2}{3},
 \end{equation}
 where the probability is over the internal randomness of $A_j$.

 \end{itemize}
\end{definition}

We shall make use of an \emph{approximate} version of the above definition in which each local algorithm $A_j$ is only required to correctly recover \emph{most} of the codeword entries. By averaging, in this case it can be assumed without loss of generality that the local algorithms $A_j$ are \emph{deterministic} (but $\cApre$ is randomized).

\begin{definition}[Approximate Locally List-Recoverable Codes (ALLRC)]\label{def:ALLR} 
Let $C \subseteq \Sigma^N$ be a code,  let $\rho \in [0,1)$, let $\ell \leq L$ and $Q$ be positive integers, and let $\eps > 0$.  We say that $C$ is  $(Q,\epsilon,\rho, \ell, L)$-\emph{approximately  locally list recoverable code (ALLRC)} 
 
if there is an algorithm $\cApre$ so that the following holds.
\begin{itemize}
    \item $\cApre$ is a randomized algorithm that, on input $\mathbf{1}^N$, outputs a list of local algorithms $(A_1, \ldots, A_L)$.
    \item Each local algorithm $A_j$ is a deterministic algorithm with query access to $\mathcal{S} \in {\Sigma \choose \leq \ell}^{N}$ and receives as input a coordinate $i \in [N]$.  On input $i \in [N]$, $A_j$ makes at most $Q$ queries to $\mathcal{S}$.
    \item For every codeword $c \in C$ so that $\delta(c,\mathcal{S}) \leq \rho$, with probability at least $\frac 2 3$ over the randomness of $\mathcal{P}$, there is some $j \in [L]$ so that
    \[ \Pr_{i \in [N]}[ A_j(i) = c_i ] \geq 1 - \eps.\]
\end{itemize}
\end{definition}

\subsection{Tensor codes}\label{subsec:tensor}
Our construction is based on \emph{tensor codes}.  For a base code $C \subseteq \Sigma^n$, the \emph{$t$-dimensional tensor code} $C^{\otimes t} \subseteq \Sigma^{[n]^t}$ is a code whose coordinates are indexed by $[n]^t$; it is given by
\[ \Cot = \inset{ c \in \Sigma^{[n]^t} \,:\, c|_V \in C \text{ for all axis-aligned lines $V$ }}, \]
where an \emph{axis-aligned line} $V \subseteq [n]^t$ is a set of the form 
\[V = \{a_1\} \times \{a_2\} \times \cdots \times \{a_{i-1}\} \times [n] \times \{a_{i+1}\} \times \cdots \times \{a_t\}\]
for some $i \in [t]$ and some $a_1, \ldots, a_t \in [n]$, and where $c|_V$ denotes restriction. 
We note that $\Cot$ can equivalently be defined recursively by $C^{\otimes t} = C^{\otimes (t-1)}\otimes C.$\footnote{For $C_1 \subseteq \Sigma^{n_1}$ and $C_2 \subseteq \Sigma^{n_2}$, $C_1 \otimes C_2 \subseteq \Sigma^{n_1 \times n_2}$ can be viewed as the set of $n_1 \times n_2$ matrices so that every column lies in $C_1$ and every row lies in $C_2$.}

When $C$ is a linear code with rate $R(C)$ and distance $\delta(C)$, the tensor code $C^{\otimes t}$ has rate $(R(C))^t$ and distance $(\delta(C))^t$.  

We next state a result of \cite{HRW19} about local list recovery of tensor codes and a corollary of it that will be useful for us later.

\begin{lemma}[\cite{HRW19}, Lemma 4.1]\label{lem:HRW}
Fix $\delta, \rho \in  [0,1)$, $\eps > 0$ sufficiently small, and any $L \geq 1$.  Then there is some $\tilde{\kappa} = \poly(1/\delta, 1/\rho, 1/\eps)$ so that the following holds.  Suppose that $C \subseteq \Sigma^n$ is a linear code of distance $\delta$ that is (globally) $(\rho, \ell, L)$-list-recoverable.  Then for any $t \geq 1$, $\Cot \subseteq \Sigma^{[n]^t}$ is a $(Q, \eps, \rho \cdot \tilde{\kappa}^{-t^2}, \ell, L^{\tilde{\kappa}^{t^2}\log^t L})$-ALLRC with query complexity $Q=n \cdot \tilde{\kappa}^{t^2\cdot \log^t L}$.
\end{lemma}

The above lemma in particular implies the following combinatorial bound on the output list size for $\Cot$.
\begin{corollary}[$\Cot$ is globally list-recoverable]\label{cor:HRW}
    Fix $\delta, \rho \in  [0,1)$,  and any $L \geq 1$.  Then there is some $\kappa = \kappa(\delta, \rho) = \poly(1/\delta, 1/\rho)$ so that the following holds.   Suppose that $C \subseteq \Sigma^n$ is a linear code of distance $\delta$ that is (globally) $(\rho, \ell, L)$-list-recoverable.  Then for any $t \geq 1$, $\Cot \subseteq \Sigma^{[n]^t}$ is (globally) $(\rho \cdot \kappa^{-t^3}, \ell, L^{\kappa^{t^3} \log^t L})$-list-recoverable. 
\end{corollary}
\begin{proof}
    We first plug in $\eps \gets \delta^t/4$ into \Cref{lem:HRW}.  This means that the parameter $\tilde{\kappa}$ from \Cref{lem:HRW} is $\tilde{\kappa} = \poly(1/\rho,1/\delta^t)$.
    Let $\kappa = 10 \tilde{\kappa}^{1/t}$, so $\kappa = \poly(1/\delta, 1/\rho)$.  
    Then \Cref{lem:HRW} implies that $\Cot$ is  $(Q, \eps=\frac {\delta^t} {4} ,\rho \cdot 10^{t^3}\kappa^{-t^3}, \ell, L^{(\kappa/10)^{t^3} \log^t L})$-ALLRC.  Let $\mathcal{L}$ be the list of local algorithms output by the approximate local list-recovery algorithm $\cApre$.  For each $A \in \mathcal{L}$, let $x(A) \in \Sigma^{[n]^t}$ be the string that $A$ corresponds to: $x(A)_i = A(i)$ for all $i \in [n]^t$. 

   Now consider the ``correct'' list
    \[ \cL^* = \inset{ c \in C^{\otimes t} : \delta(c, \cS) \leq \rho \cdot \kappa^{-t^3} }.\]
    The guarantee of \Cref{lem:HRW} (along with the fact that 
$\rho \cdot 10^{t^3} \kappa^{-t^3} > \rho\cdot  \kappa^{-t^3}$) implies that, for any $c \in \cL^*$, with probability at least $\frac 2 3$ over the randomness of $\cApre$, there is some $A \in \cL$ so that 
   $ \delta(x(A), c) \leq \frac{\delta^t}{4} $.
    Thus,  
    \[ \mathbb{E}_{\cApre}|\cL| \geq \sum_{c \in \cL^*} \Pr_{\cApre}\left[\exists A \in \cL \text{ s.t. } \delta(x(A), c) \leq \frac{\delta^t}{4} \right] \geq |\cL^*| \cdot \frac 2 3,\]
    where in the first inequality we are using the fact that for any algorithm $A$ there is at most one codeword $c \in C$ so that  $\delta(x(A), c) \leq \frac{\delta^t}{4}$ (otherwise, by triangle inequality, there would be two distinct codewords $c, c'\in C$ so that $\delta(c,c')\leq \frac{\delta^t} {2} <  (\delta(C))^t$).

    Since we also have $|\cL| \leq L^{(\kappa/10)^{t^3 \log^t L}}$ by the conclusion of \Cref{lem:HRW}, this implies that
    \[ |\cL^*| \leq \frac{3}{2} L^{(\kappa/10)^{t^3} \log^t L}. \]
    For $L \geq 2$, this implies that
    \begin{equation}\label{eq:listBoundCor} |\cL^*| \leq L^{\kappa^{t^3} \log^t L}, \end{equation}
    as desired.
    On the other hand if $L = 1$, then the above reads $|\cL^*| \leq 3/2$, and since $|\cL^*|$ is an integer we conclude that $|\cL^*| \leq 1$, so \Cref{eq:listBoundCor} still holds.  This proves the corollary.
\end{proof}

\subsection{Samplers}
We will make use of \emph{samplers}.  Below, we give a definition, and record parameters for good samplers.
\begin{definition}[Sampler]\label{def:samplerGamma} An $(n, \eta, \nu)$-\emph{sampler}  with \emph{seed length} $r$ and \emph{sample size} $m$ is a randomized algorithm $\Gamma$ that tosses $r$ random coins and outputs a subset $I \subseteq [n]$ of size $m$ so that the following holds.  For any function $f:[n] \to [0,1]$, with probability at least $1 - \eta$ over the choice of~$I$, 
\[ \inabs{\EE_{i \in I}[f(i)] - \EE_{i \in [n]}[f(i)]} \leq \nu.\]
\end{definition}

\begin{theorem}[\cite{goldreich2011sample}, Corollary 5.6]\label{thm:sampler}
    For any $\eta, \nu > 0$ and an integer $n$, there exists an $(n, \eta, \nu)$-sampler with randomness $\log(n/\nu)$, sample size $m = O\left( \frac{1}{\eta \nu^2}\right)$, and running time $\poly(\log n, 1/\eta, 1/\nu)$.
\end{theorem}

\section{Deterministic local codes}\label{sec:local_defs}

In this section we formally define and discuss the notions of a \emph{deterministic (approximate) locally list recoverable code} (D(A)LLRC) and a \emph{deterministic locally correctable code} (DLCC), which will be the main objects of study in this work. Then, in Section \ref{subsec:dllr_transform} below we present several transformations on such codes that we shall make use of in our constructions.
To motivate the above definitions, we first discuss their implications to time- and space-efficient decoding algorithms.
 We start by defining our computational model for low-space algorithms.

\begin{definition}[Computational model for low-space algorithms, \cite{CM25_enc, CM25,CM26}]\label{def:model}
Our time- and space-efficient algorithms are uniform algorithms that have access to a read-only input tape, a read/write working tape, and a write-only output tape.  The \emph{space} of the algorithm refers to the amount of space on the read/write working tape that it uses.  

We assume that our algorithms have random access to the input tape and to the working tape, but that they must write to the output tape in sequential order. For example, if an algorithm outputs a codeword $c \in \Sigma^n$, it must output the first symbol $c_1$, then $c_2$, and so on; if it outputs a list of representations of local algorithms $A_1, A_2, \ldots, A_L$, then it must first output a canonical description of $A_1$, then of $A_2$, and so on.  In the case of an algorithm that has oracle access to some input (say a set of input lists $\cS$) as well as another input (say a query position $i \in [N]$), we imagine that both are written on the input tape in some canonical order. 
\end{definition}

Our eventual goal is to obtain \emph{deterministic} list-recovery (and list-decoding) algorithms with sub-linear space and near-linear time. That is, a deterministic algorithm  that has access to a tuple of input lists $\cS = (S_1, \ldots, S_N)$, and outputs a short list containing $\{c \in C : \delta(c, \cS) \leq \rho \}$ in near-linear time and sublinear space, in the computational model described above.

Note that a locally list recoverable code (LLRC, Definition \ref{defn:LLR}) gives a \emph{randomized} such algorithm by executing each of the local algorithms $A_j$ on any input coordinate $i \in [N]$ (assuming the decoding error is sufficiently small, which can be achieved by repeating the decoding process independently  and  outputting the majority value).  However, as noted in the introduction, the resulting algorithm must be randomized, since randomness is necessary for a local list-recovery algorithm (with constant $L,Q \ll N$ and $\rho = \Omega(1)$).

To obtain a \emph{deterministic} time- and space-efficient list recovery algorithm, we consider a \emph{deterministic} version of a locally list recoverable code. 
We formally define this below, noting that a similar version of the definition for list-decoding was also implicitly used in \cite{CM26}.  We hope that there is value in making this definition explicit.  In particular, as we will see later, DALLRCs can be naturally transformed in standard ways while preserving near-linear time and sublinear space, while a time- and space-efficient algorithm with no additional structure is not so easy to manipulate.

\begin{definition}[Deterministic Approximate Locally List-Recoverable Codes (DALLRC), and Deterministic Locally List-Recoverable Codes (DLLRC)]\label{def:DALLR}
Let $C \subseteq \Sigma^N$ be a code, let  $\rho \in [0,1)$, let $\ell \leq L$ and $Q$ be positive integers, and let $\eps > 0$.  We say that $C$ is a 
$(Q,\eps,\rho, \ell,L)$-\emph{deterministic approximate locally list-recoverable code}  (DALLRC) with
preprocessing time $\Tmpre$, preprocessing space $\Sppre$, evaluation time $\Tmeval$, evaluation space $\Speval$, and output length $\Spop$ 
if there exists an algorithm $\cApre$ (referred to as a \emph{DALLR algorithm}) so that the following holds.  
\begin{itemize}
    \item $\cApre$ is a \emph{deterministic} algorithm which receives as input $\mathcal{S} \in {\Sigma \choose \leq \ell}^{N}.$
    \item In time $\Tmpre$ and space $\Sppre$, $\cApre$ outputs a list of (representations of)\footnote{As we will see later, in order to make sure that the space stays small, we need to be careful about how we represent the local algorithms $A_j$ in memory.} local algorithms $(A_1, \ldots, A_L)$, where the total amount written on the output tape (that is, the space to store $L$ representations of local algorithms) is at most $\Spop$.
    \item Each local algorithm $A_j$ is a deterministic algorithm with query access to $\mathcal{S} \in { \Sigma \choose \leq \ell}^N$ and receives as input a coordinate $i \in [N]$.  On input $i \in [N]$, $A_j$ makes at most $Q$ queries to $\cS$, and runs in time $\Tm^{(Eval)}$ and space $\Sp^{(Eval)}$.  
    \item For every codeword $c \in C$ so that $\delta(c,\cS) \leq \rho$, there is some $j \in [L]$ so that
    \[ \Pr_{i \in [N]}[ A_j(i) = c_i ] \geq 1-\eps. \]
\end{itemize}
In the special case that $\eps=0$, we refer to a $(Q,\eps=0,\rho, \ell,L)$-$\ADLLR$  as a \emph{$(Q, \rho, \ell, L)$-deterministic locally list-recoverable code} $($\DLLR$)$, and to the  algorithm $\cApre$ as a DLLR algorithm.
\end{definition}

When counting the number of queries, we assume that each position is queried at most once. This assumption is without loss of generality as long as $\Speval \geq Q$, which will always be the case for our  algorithms.

\begin{remark}[On the ``locality'' of $\ADLLR$s]\label{rem:local} While the local algorithms $A_j$ in \Cref{def:DALLR} are local in the sense that they do not make too many queries to $\cS$, the overall algorithm $\cApre$ may be global, meaning that as preprocessing, it may read all of $\cS$.  As noted above, for a deterministic algorithm, it is \emph{necessary} that some part of the algorithm be global.  If $\cApre$ had unlimited time and space, \Cref{def:DALLR} would not be very interesting, as $\cApre$ could just run a global list-recovery algorithm and hard-code a local algorithm $A_j$ for each codeword in the output list, to obtain a $\DLLR$ algorithm with $Q=0$.  However, the fact that $\cApre$ has limited resources makes the notion in \Cref{def:DALLR} non-trivial (and useful).

\end{remark}

The following simple lemma states that a DLLR algorithm  gives a deterministic global list recovery algorithm with similar time and space bounds.

\begin{lemma}\label{lem:dllr_to_global}
Suppose that $C \subseteq \Sigma^N$ is a $(Q,\rho, \ell,L)$-\DLLR\  with preprocessing time $\Tmpre$, preprocessing space $\Sppre$, evaluation time $\Tmeval$, evaluation space $\Speval$, and output length $\Spop$. Then $C$ is $(\rho,\ell,L)$-globally list recoverable deterministically in time $\Tmpre + N \cdot L \cdot \Tmeval$ and space $\Sppre + \Spop + \Speval$.
\end{lemma}

\begin{proof}
The deterministic global list recovery algorithm $\mathcal{A}$ executes the DLLR algorithm $\cApre$, and writes the description of the $L$ local algorithms $A_1, \ldots, A_L$ it outputs on the working tape. Then for $j=1,\ldots,L$, $\mathcal{A}$ executes each of the local algorithm $A_j$ on input $i$ for $i=1,\ldots,N$ and writes all the outputs sequentially on the output tape. The time and space bounds follow.
\end{proof}

In the special case that $\eps=0$ and $\ell=L=1$, we say that a DALLRC as in \Cref{def:DALLR} is a \emph{deterministic locally correctable code}  (DLCC).  We state this as its own definition below. A similar definition was also implicit in \cite{CM25}.

\begin{definition}[Deterministic Locally Correctable Code (DLCC)]\label{def:DLC}
Let $C \subseteq \Sigma^N$ be a code, let $\rho \in [0,1)$ so that $\rho < \frac{\delta(C)}{2}$, and let $Q$ be a positive integer.  We say that $C$ is a $(Q,\rho)$-\emph{Deterministic Locally Correctable Code (DLCC)} with preprocessing time $\Tm^{(Pre)}$, preprocessing space $\Sp^{(Pre)}$, evaluation time $\Tm^{(Eval)}$, evaluation space $\Sp^{(Eval)}$, and output length $\Spop$  if there exists an algorithm $\cApre$ (referred to as a $\mathrm{DLC}$ algorithm) so that the following holds.
\begin{itemize}
    \item $\cApre$ is a \emph{deterministic} algorithm which receives as input a string $w \in \Sigma^N$.
    \item In time $\Tm^{(Pre)}$ and space $\Sp^{(Pre)}$, $\cApre$ outputs a representation of a local algorithm $A$, where the total amount written on the output tape (that is, the space to store the representation of $A$) is at most $\Spop$. 
    \item The local algorithm $A$ is a \emph{deterministic} algorithm with query access to $w \in \Sigma^N$ and receives as input a coordinate $i \in[N]$.  On input $i \in [N]$, $A$ makes at most $Q$ queries to $w$, and runs in time $\Tm^{(Eval)}$ and space $\Sp^{(Eval)}$.
    \item If there is some $c \in C$ so that $\delta(c,w) \leq \rho$, then $A(i) = c_i$.  
\end{itemize}
\end{definition}

\begin{remark}[On the ``locality'' of DLCCs]
    As with local list recovery, it is impossible to have a deterministic LCC (in the traditional sense) that is completely local; indeed, an adversary with a budget of $\rho n$ corruptions can completely corrupt the $Q$ symbols that such a deterministic algorithm would read.  Instead, our definition allows the preprocessing algorithm $\cApre$ to be global.  As in \Cref{rem:local}, the time and space restrictions on $\cApre$ make \Cref{def:DLC} non-trivial. 
\end{remark}

Finally, we introduce some notation that will be useful.  
\begin{definition}[String corresponding to a local algorithm]\label{def:alg_to_string}
    Let $A$ be a deterministic local algorithm, as returned by a DALLR or DLC algorithm, for a code $C \subseteq \Sigma^N$.  We associate $A$ with a string $x(A) \in \Sigma^{N}$ in the natural way:  That is, for $i \in [N]$, define
    $ x(A)_i := A(i).$
\end{definition}

\subsection{D(A)LLRC transformations}\label{subsec:dllr_transform}

In this section we present some useful transformations on D(A)LLRCs.
The first transformation shows that the intersection of a \ADLLR\ and a \DLCC\ is a \DLLR\ with roughly the same time and space bounds.

\begin{lemma}[The intersection of a \ADLLR\ and a \DLCC\ is a \DLLR]\label{lem:mainDLLR}
    Suppose that $C_1 \subseteq \Sigma^N$ is a linear code of rate $R_1 = 1 - \zeta_1$ that is a $(Q_1, \eps, \rho, \ell, L)$-\ADLLR\  with preprocessing time and space $(\Tmcon_1 , \Spcon_1)$, evaluation time and space $(\Tmeval_1, \Speval_1)$, and output length $\Spop_1$.
    
    Suppose that $ C_2 \subseteq \Sigma^N$ is a linear code of rate $R_2 =  1-\zeta_2$ that is a $(Q_2,\eps)$-\DLCC\  with preprocessing time and space  $(\Tmcon_2, \Spcon_2)$, evaluation time and space $(\Tmeval_2, \Speval_2)$, and output length $\Spop_2$.

    Suppose that $\eps < \frac{\delta(C_2)} 2$, and let 
$\Cstar = C_1\cap C_2 \subseteq \Sigma^N$.  Then $\Cstar$ is a linear code of rate $R^{\star} \geq 1 - \zeta_1 - \zeta_2$ and distance $\delta(C^\star) \geq \max\{\delta(C_1), \delta(C_2)\}$ that is a $(Q = Q_1Q_2, \rho, \ell, L)$-\DLLR\ 
with preprocessing time and space 
\[
\Tmcon_{\star} = O\left(\Tmcon_1 + L \cdot \Tmcon_2 \cdot \Tmeval_1\right) \qquad \Spcon_{\star} = O\left(\Spcon_1 + \Spcon_2 + \Speval_1 + \Spop_1\right),
\]
evaluation time and space 
\[
\Tmeval_{\star} = O\inparen{\Tmeval_2 + Q_2 \cdot \Tmeval_1}  \qquad \Speval_{\star} = O\inparen{\Speval_2 + \Speval_1}, 
\]
and output length
\[
 \Spop_{{\star}} = \Spop_1 + L \cdot \Spop_2.\]
\end{lemma} 
\begin{proof} 
First, we observe that the code $C^{\star}$ does indeed have rate at least $1 - \zeta_1 - \zeta_2$.  Indeed, as $C_1$ and $C_2$ are both linear codes, they are defined by $\zeta_1 N$ and $\zeta_2 N$ linear constraints, respectively.  Thus, $C_1 \cap C_2$ is defined by at most $(\zeta_1 + \zeta_2)N$ linear constraints, and thus has rate at least $1 - \zeta_1 - \zeta_2$.
Similarly, the bound on the distance follows since $C^\star \subseteq C_1$ and $C^\star \subseteq C_2$.

Let $\cApre_1, \cApre_2$ be the DALLR and DLC algorithms for $C_1, C_2$, respectively.
Next we will describe the DLLR algorithm $\cApre^\star$ for $C^\star$; then prove that it is correct; then analyze its space and time requirements.

\paragraph{The algorithm $\cApre^{\star}$.} 
Before we describe $\cApre^{\star}$, we describe the representations of the local algorithms that it outputs.

A local algorithm $A^{\star}$ output by $\cApre^{\star}$ is of the form \[\rep(A^\star) = (\rep(A), \rep(\bar{A})),\] where $\rep(A)$ is the representation of a local algorithm output by $\cApre_1$, and $\rep(\bar{A})$ is the representation of a local algorithm output by $\cApre_2$.  Given this representation, the local algorithm $A^{\star}$ is implemented as follows.  On input $i \in [N]$ and with oracle access to $\cS \in {\Sigma \choose \leq \ell}^N$, to compute $A^{\star}(i)$:
\begin{itemize}
    \item Use $A$ to implement oracle access to the string $x(A)$ computed by $A$: to query $x(A)_j$, we run $A$ with input $j$ and oracle access to $\cS$.
    \item Run $\bar{A}$ on the input $i$, with oracle access to $x(A)$ provided as above.
\end{itemize}  

We observe that each local algorithm $A^\star$ has query complexity at most $Q = Q_1 \cdot Q_2$.  Indeed, if $\rep(A^\star) = (\rep(A), \rep(\bar{A}))$, then $\bar{A}$ makes at most $Q_2$ queries to $x(A)$, and for each of those $A$ will make at most $Q_1$ queries to $\cS$.

Now that we know what the representations of local algorithms look like, we can describe the DLLR algorithm $\cApre^*$.  We do that in \Cref{alg:intersection} below.

\begin{algorithm}[H]
\caption{$\cApre^\star$: DLLR algorithm for $C^\star = C_1 \cap C_2$}
\label{alg:intersection}
\begin{algorithmic}[1]
\Require Oracle access to $\cS \in {\Sigma \choose \leq \ell}^N$
\Ensure Outputs a list $\ctL$ of $L$ representations of local algorithms.
\State Run $\cApre_1$ with oracle access to $\cS$ to obtain a list $\ctL_1$ of representations of local algorithms.
\State $\ctL = \emptyset$
\ForAll{$\rep(A) \in \ctL_1$}
    \State Run $\cApre_2$ with oracle access to $x(A)$ to get a representation of a local algorithm $\rep(\bar{A})$.
    \State \Comment{Implement oracle access to $x(A)$ using $A$: to query $x(A)_j$, run $A(j)$ with oracle access to~$\cS$}
    \State $\ctL \gets \ctL \cup \{ (\rep(A), \rep(\bar{A})) \}$
\EndFor
\Return $\ctL$
\end{algorithmic}
\end{algorithm}

\paragraph{The algorithm $\cApre^\star$ is correct.}  Next, we show that $\cApre^\star$ as in \Cref{alg:intersection} is correct.  Let $\cL^\star = \inset{ c \in C^\star : \delta(c, \cS) \leq \rho }$.  We would like to show that for all $c \in \cL^\star$, there is some $\rep(A^\star) \in \ctL$ so that $x(A^\star) = c$.

Fix $c^\star \in \cL^\star$.  Let $\ctL_1$ be the output of running $\cApre_1$ with oracle access to $\cS$, as in \Cref{alg:intersection}.  Since $\cL^\star \subseteq C^\star \subseteq C_1$, the correctness of $\cApre_1$ (\Cref{def:DALLR}) implies that there is a local algorithm $\rep(A) \in \ctL_1$ so that $\delta(x(A), c^\star) \leq \eps$.  
Now let $\rep(\bar{A})$ be the result of running $\cApre_2$ with oracle access to $x(A)$, as in \Cref{alg:intersection}.  Note that $c^\star \in C^\star \subseteq C_2$, and by the above $\delta(x(A), c^\star) \leq \eps < \frac{\delta(C_2)} 2$, so by the correctness of $\cApre_2$ (\Cref{def:DLC}), $x(\bar{A}) = c^\star$.  But $x(\bar{A})$ is precisely $x(A^\star)$, so we conclude that $x(A^\star) = c^\star$, as desired.

\paragraph{Resources.}  Next, we record the pre-processing and evaluation time and space for $\cApre^\star$.  
First consider the pre-processing time $\Tmcon_{\star}$.  To run $\cApre^\star$, we first run $\cApre_1$, with a cost of $\Tmcon_1$.  Then for all $L$ local algorithms $\rep(A)$ in $\ctL_1$, we run $\cApre_2$ with oracle access to $x(A)$.  Running $\cApre_2$ takes time $\Tmcon_2$, and in particular we make at most $\Tmcon_2$ queries to $x(A)$.  Each query (simulated using $A$) takes time $\Tmeval_1$, for a total of
\[ \Tmcon_{\star} = O\inparen{\Tmcon_1 + L\cdot \Tmcon_2 \cdot \Tmeval_1}.\]
Next, we claim that the space required for $\cApre^\star$ is 
\begin{equation}\label{eq:spacecon_bound} \Spcon_{\star} = O(\Spcon_1 + \Spcon_2 + \Speval_1 + \Spop_1).\end{equation}
To see why, we first recall that $\cApre_1$ outputs the list $\ctL_1$ on its output tape, with space $\Spop_1$ that gets written on $\cApre^\star$'s work tape.  We also need space $\Sppre_1$ to account for the space to run $\cApre_1$.  
Now, for each local algorithm $\rep(A)$ output by $\cApre_1$, we run $\cApre_2$, using $A$ to simulate oracle access to $x(A)$.  Running $\cApre_2$ requires space $\Sppre_2$.  Each time we query $x(A)$ requires space $\Speval_1$; notice that we can re-use the space from each query, because if $\cApre_2$ stored the query it would be accounted for in $\Sppre_2$.  Notice that we can also run each instance of $\cApre_2$ iteratively, re-using the same space.  This gives the final bound on $\Sppre_{\star}$.

Next, we check the evaluation time $\Tmeval_{\star}$.  To run a local algorithm $A^\star$ with representation $\rep(A^\star) = (\rep(A), \rep(\bar{A}))$, we run $\bar{A}$, using $A$ to simulate query access to $x(A)$.  The local algorithm $\bar{A}$ takes time $\Tmeval_2$ to run, and it makes at most $Q_2$ queries to $x(A)$, each of which takes time $\Tmeval_1$ to simulate.  So the total is
\[ \Tmeval_{\star} = O\left(\Tmeval_2 + Q_2 \cdot \Tmeval_1 \right).\]
Next, we check the space $\Speval_{\star}$.  As above, we need to run the local algorithm $\bar{A}$ with space $\Speval_2$, and then we need space to simulate the queries to $x(A)$, which is $\Speval_1$.  Notice that we can re-use the space for these queries, since if $\bar{A}$ would store them, that would count as part of the space $\Speval_2$.  Thus, the total space is
\[ \Speval_{\star} = O\inparen{ \Speval_1 + \Speval_2 }.\]

Finally, we consider the output length.  The algorithm $\cApre^\star$ outputs $L$ representations of local algorithms, where each representation is of the form $(\rep({A}), \rep(\bar{A}))$.  All of the $L$ representations $\rep(A) \in \ctL_1$ together take space $\Spop_1$.  Each of the representations $\rep(\bar{A})$ takes space $\Spop_2$.  So the total space is 
\[ \Spop_{\cApre^{\star}} = \Spop_1 + L \cdot \Spop_2.\]

This completes the proof.
\end{proof}

The next lemma gives a \emph{concatenation} procedure for DLLRCs, which can be used to reduce the alphabet size of these codes.

\begin{lemma}[Concatenation for \DLLR]\label{lem:DLLR_concat}
Suppose the codes $C_{out}$ and $C_{in}$ exist with the following parameters: 
\begin{itemize}
\item $C_{\out}$ is an $\F$-linear code of block length $N_{\out}$, alphabet 
$\Sigma_{\out}$, rate $R_{\out}$, and distance $\delta_{\out}$
that is encodable in time $\Tm_{\out}^{(enc)}$, and is
a $(Q_{\out}, \rho_{\out}, \ell_{\out}, L_{\out})$-$\DLLR$ with
preprocessing  
time $\Tm_{\out}^{(Pre)}$,  preprocessing space $\Sp_{\out}^{ (Pre)}$, evaluation time 
$\Tm_{\out}^{(Eval)}$, evaluation space $\Sp_{\out}^{(Eval)}$, and output length $\Spop_{\out}$. 
\item $C_{\inn}$ is an $\F$-linear code of block length $N_{\inn}$, alphabet $\Sigma_{\inn}$,
rate $R_{\inn}$, and distance $\delta_{\inn}$ that is encodable in time $\Tm_{\inn}^{(enc)}$ and space $\Sp_{\inn}^{(enc)}$, and is
$(\rho_{\inn},\ell_{\inn},L_{\inn})$-(globally) list recoverable deterministically in time $\Tm_{\inn}$ and space  $\Sp_{\inn}$.  
Suppose furthermore that the encoding map for $C_{\inn}$ can be inverted in time $\Tm_{\inn}^{(unenc)}$ and space  $\Sp_{\inn}^{(unenc)}$.
\end{itemize}

If the parameters of $C_{out}$ and $C_{in}$ satisfy 
$|\Sigma_{\out}|=|\Sigma_{\inn}|^{R_{\inn}\cdot N_{\inn}}$ and $L_{\inn}\leq\ell_{\out}$,
then there exists an $\F$-linear code $C$ of block length $N_{\out}\cdot N_{\inn}$, alphabet
 $\Sigma_{\inn}$, rate $R_{\inn}\cdot R_{\out}$, and 
distance at least $\delta_{\out} \cdot \delta_{\inn}$ that is encodable in time 
${\Tm}_{out}^{(enc)} + N_{out} \cdot {\Tm}_{in}^{(enc)},$
and is a $(Q_{\out} \cdot N_{\inn},\rho_{\out} \cdot \rho_{\inn},\ell_{\inn},L_{\out})$-$\DLLR$  with preprocessing time  
$$\Tm_{\out}^{(Pre)} \cdot (\Tm_{\inn} +L_{\inn} \cdot \Tm_{\inn}^{(unenc)}),$$ preprocessing space 
$$ \Sp_{\out}^{(Pre)} + \Sp_{\inn} + N_{\inn} \cdot (\ell_{\inn} +L_{\inn}) \cdot \log(|\Sigma_{\inn}|) + \Sp_{\inn}^{(unenc)},$$
evaluation time 
$$T_{\out}^{(Eval)} \cdot (\Tm_{\inn} + L_{\inn} \cdot T_{\inn}^{(unenc)})+   T_{\inn}^{(enc)} ,$$ 
evaluation space $$ \Sp_{\out}^{(Eval)} + \Sp_{\inn} +    \Sp_{\inn}^{(enc)} + N_{\inn} \cdot (\ell_{\inn} + L_{\inn})\cdot \log(|\Sigma_{\inn}|)+ \Sp_{\inn}^{(unenc)},$$
and output length $\Spop_{\out}+ O(1)$. 
\end{lemma}

\begin{proof}

We first describe the construction of the code $C$ and its properties, and then we describe the corresponding DLLR algorithm and analyze it.

\paragraph{Construction of the code $C$.}
The code $C$ is the concatenation of $C_{\out}$ with $C_{\inn}$.  Formally,
we construct $C$ by giving a bijection from $C_{\out}$ to $C$. Let 
$\Sigma_{\out}, \Sigma_{\inn}$ denote the alphabets of $C_{\out}, C_{\inn}$ respectively.
Given a codeword $c_{\out}\in C_{\out}$, one obtains the corresponding codeword
$c\in C$ as follows: 
 View each codeword symbol in $\Sigma_{\out}$ as a vector of length $R_{\inn}\cdot N_{\inn}$ over $\Sigma_{\inn}$ and encode it via the code $C_{\inn}$. Each codeword symbol gets mapped to a string in  $\Sigma_{\inn}^{N_{\inn}}$. We denote the resulting string by $c \in \Sigma_{\inn}^{N_{\inn}\cdot N_{\out}}$ and the various resulting codewords of $C_{\inn}$ by $B_1, B_2, \ldots B_{N_{\out}} \in \Sigma_{\inn}^{N_{\inn}}$.  

\paragraph{Code parameters:}
It readily follows that $C$ is an $\F$-linear code of blocklength $N:=N_{\out} \cdot N_{\inn}$ and alphabet size $\Sigma_{\inn}$, and that the encoding time is as stated. The rate and distance analysis is standard (see e.g., \cite[Section 10]{GRS_book}).
We now turn to describe and analyze the DLLR algorithm for the code $C$. 

\paragraph{DLLR algorithm for $C$.}
We will now describe the $(Q,\rho_{\out} \cdot \rho_{\inn},\ell_{\inn},L_{\out})$-DLLR algorithm $\cApre$ for the code $C$. 
Let $\bar \cApre$ be the DLLR algorithm for $C_{\out}$. $\bar \cApre$ is a deterministic algorithm that given 
a tuple $\overline \cS = (\overline S_1, \ldots , \overline S_{N_{\out}}) \in {\Sigma_{\out} \choose {\leq \ell_{\out}}}^{N_{\out}}$
 outputs a list of $L_{\out}$ deterministic local algorithms $\bar A_1, \bar A_2, \ldots, \bar A_{L_{\out}}$. 

We now describe $\cApre$. 
Suppose the algorithm $\cApre$ is invoked on a tuple $\cS = (S_1, \ldots, S_N)\in {\Sigma_{\inn} \choose {\leq \ell_{\inn}}}^{N}$. Then $\cApre$ invokes the algorithm $\bar \cApre$ and emulates $\bar \cApre$ in the natural way. 
Recall that $\bar \cApre$ expects to be given access to a tuple $\bar S \in  {\Sigma_{\out} \choose {\leq \ell_{\out}}}^{N_{\out}}$. 
For any $k \in [N_{\out}]$, whenever  $\bar \cApre$  queries the $k$-th element of the sequence $\bar S_1, \ldots, \bar S_{N_{\out}} \in  {\Sigma_{\out} \choose {\leq \ell_{\out}}}$, the algorithm $ \cApre$ does the following. The algorithm $ \cApre$ first invokes the global list-recovery algorithm for $C_{\inn}$ with the lists $S^{(k,r)} \in  {\Sigma_{\inn} \choose {\leq \ell_{\inn}}}$ which correspond to the $r$-th entry of $B_k$ for each $r \in [N_{\inn}]$. The output of this algorithm is a list of size at most $L_{\inn} \leq \ell_{\out}$ with elements from $\Sigma_{\inn}^{ N_{\inn}}$. We denote by $\bar S_k$ the set of messages in $\Sigma_{\inn}^{r_{\inn} N_{\inn}} = \Sigma_{\out}$ corresponding to the codewords in this list.
This is what $ \cApre$ feeds to $\bar \cApre$. 

Suppose that $\bar \cApre$ outputs a list of $L_{\out}$ deterministic local algorithms $\bar A_1, \bar A_2, \ldots, \bar A_{L_{\out}}$. For any $\bar A_j$ output by $\bar \cApre$, the  algorithm $ \cApre$ outputs a local algorithm $A_j$ defined as follows. 
Each $A_j$ takes as input a coordinate in $[N_{\out} \cdot N_{\inn}]$ and also gets access to the lists $S^{(k,r)}$. 
On input coordinate $(k,r)$ which corresponds to the $r$-th entry in the $k$-th block for $r \in [N_{\inn}]$ and $k \in [N_{\out}]$, $A_j$ emulates $\bar A_j$ on input coordinate $k$, where each query of $\bar A_j$ to $\bar S$ is emulated as described above. The output of $\bar A_j$ is a value $v \in \Sigma_{\out} = \Sigma_{\inn}^{r_{\inn} N_{\inn}}$, and the local algorithm $A_j$ encodes $v$ using the code $C_{\inn}$ and outputs the $r$-th entry of the encoding.

\medskip

\paragraph{Correctness:}
Clearly, the query complexity of each algorithm~$ A_j$ is at most $N_{\inn}$ times the query complexity of~$\bar A_j$, and hence it is at most $Q_{\out} \cdot N_{\inn}$.

Next assume that $c_{\out} \in C_{\out}$ is such that the corresponding 
codeword $c$ of $C$ (as given by the bijection above) satisfies 
$\dist(c,S)\leq \rho_{\out} \cdot \rho_{\inn}$. 
We claim that the tuple  $\bar S :=(\bar S_1, \bar S_2, \ldots, \bar S_{N_{\out}})$ as defined  above satisfies $\dist(c_{\out}, \bar S) \leq \rho_{\out}$.

To see the above, let $T \subseteq [N_{\out}]$ be the subset of all $k \in [N_{\out}]$ so that for at most $\rho_{\inn}$-fraction of  $r \in [N_{\inn}]$ we have that $(B_k)_r \notin S^{(k,r)}$. Then by Markov's inequality, we have that $|T| \geq (1 - \rho_{\out}) \cdot N_{\out}$. 

Next observe that for any $k \in T$, the encoding of $c_{\out_k}$ via the code $C_{\inn}$ (which we call $B_k$) agrees with various $S^{(k,r)} $ for at least a $(1-\rho_{\inn})$-fraction of $r\in [N_{\inn}]$. Consequently, the global list recovery algorithm of $C_{\inn}$ will succeed in outputting the encoding of $c_{\out_k}$, and so 
 $c_{\out_k} \in \bar S_k$. 
Since $|T| \geq (1-\rho_{\out}) \cdot N_{\out}$, this shows that $\dist(c_{\out},\bar S) \leq \rho_{\out}$.  

Consequently, by the correctness of the algorithm $\bar \cApre$, we have that 
$\bar A_j(k)=c_{\out_k}$ for any   $k \in [N_{\out}]$, and so $A_j(k,r)$ is the 
$r$-th entry in the $k$-th block of $c$ for any $k \in [N_{\out}]$ and $r \in [N_{\inn}]$. 

\paragraph{Time and space bounds:}
Finally, we analyze the running time and space of the algorithm $ \cApre$ and each of the local algorithms $ A_j$. To this end, note that the algorithm $ \cApre$ ($ A_j$, respectively) emulates the algorithm $\bar \cApre$ ($\bar A_j$, respectively), where in each time step of  $\bar \cApre$ ($\bar A_j$, respectively), the algorithm $\cApre$ ($ A_j$, respectively) needs to store the values of $S^{(k,r)}$ for $k \in [N_{\out}]$ and $r \in [N_{\inn}]$, which takes space at most  $N_{\inn} \cdot \ell_{\inn} \cdot \log(|\Sigma_{\inn}|)$ (noting that we only need to store one $k$ block at a time), then invoke the global list recovery algorithm for $C_{\inn}$, and find and store the preimages of the codewords in the output list under the encoding map of $C_{\inn}$. Additionally, each local algorithm $A_j$ has to invoke the encoding algorithm for $C_{\inn}$ on the output $v$ of $\bar A_j$. 

Consequently,  $ \cApre$ has running time at most $$ T_{\out}^{(Pre)} \cdot (T_{\inn} + L_{\inn} \cdot T_{\inn}^{(unenc)}) $$
and space at most
$$ \Sp_{\out}^{(Pre)} + \Sp_{\inn} + N_{\inn} \cdot (\ell_{\inn}+L_{\inn}) \cdot \log(|\Sigma_{\inn}|) + \Sp_{\inn}^{(unenc)}.$$ Similarly, each local algorithm $A_j$ has running time at most $$T_{\out}^{(Eval)} \cdot (T_{\inn}+   L_{\inn} \cdot T_{\inn}^{(unenc)}) +   T_{\inn}^{(enc)} $$ and space at most $$ \Sp_{\out}^{(Eval)} + \Sp_{\inn} +    \Sp_{\inn}^{(enc)} + N_{\inn} \cdot (\ell_{\inn}+L_{\inn}) \cdot \log(|\Sigma_{\inn}|) + \Sp_{\inn}^{(unenc)}.$$
The output length is also clearly $\Spop_{\out}+ O(1)$. 

\end{proof}

Finally, we state the following lemma, which shows how to transform a \emph{high-rate} \DLLR\ into a \emph{capacity-achieving} \DLLR. The transformation is based on the AEL transformation of \cite{AEL95}. A similar transformation was implicitly given in \cite[Section 4.6]{CM25} 
for the \DLCC\ setting, and explicitly in \cite[Lemma 5.4]{GKORS18}
for the setting of  (randomized) LLRC. For completeness, we provide a full proof of this lemma in Appendix \ref{sec:AEL}.

\begin{restatable}{lemma}{AEL}\emph{[Distance amplification for \DLLR]}\label{lem:AEL}

Suppose the codes $C_{out}$ and $C_{in}$ exist with the following parameters: 
\begin{itemize}
\item $C_{\out}$ is an $\F$-linear code of block length $N_{\out}$, alphabet size
$\Sigma_{\out}$, rate $R_{\out}$, and distance $\delta_{\out}$ that is encodable 
in time $\Tm_{\out}^{(enc)}$, 
and is a $(Q_{\out}, \rho_{\out}, \ell_{\out}, L_{\out})$-$\DLLR$ with 
preprocessing  
time $\Tm_{\out}^{(Pre)}$,  preprocessing space $\Sp_{\out}^{ (Pre)}$, evaluation time 
$\Tm_{\out}^{(Eval)}$, evaluation space $\Sp_{\out}^{(Eval)}$, and output length $\Spop_{\out}$. 
\item $C_{\inn}$ is an $\F$-linear code of block length $N_{\inn}$, alphabet size $\Sigma_{\inn}$,
rate $R_{\inn}$, and distance $\delta_{\inn}$ that is encodable in time $\Tm_{\inn}^{(enc)}$ and space $\Sp_{\inn}^{(enc)}$, and is
$(\rho_{\inn},\ell_{\inn},L_{\inn})$-(globally) list recoverable deterministically in time $\Tm_{\inn}$ and space  $\Sp_{\inn}$. 
Suppose furthermore that the encoding map for $C_{\inn}$ can be inverted in time $\Tm_{\inn}^{(unenc)}$ and space  $\Sp_{\inn}^{(unenc)}$.
\end{itemize}

There exists a $d=d(\delta_{out},\rho_{out},\gamma)=(1/\delta_{out}+1/\rho_{out}+1/\gamma)^{O(1)}$ such that if the parameters of $C_{out}$ and $C_{in}$ satisfy $N_{\inn}\geq d$,
$|\Sigma_{\out}|=|\Sigma_{\inn}|^{R_{\inn}\cdot N_{\inn}}$ and $L_{\inn}\leq\ell_{\out}$,
then there exists an $\F$-linear code $C$ of block length $N_{\out}$, alphabet
size $\Sigma_{\inn}^{N_{\inn}}$, rate $R_{\inn}\cdot R_{\out}$, and 
distance at least $\delta_{\inn}-2\gamma$ that is encodable in time 
$${\Tm}_{out}^{(enc)} + N_{out} \cdot {\Tm}_{in}^{(enc)} + 
	N_{out} \cdot \poly(N_{in}, \log(N_{out})),$$
and is a $(Q_{\out} \cdot N_{\inn}^2,\rho_{\inn}-\gamma,\ell_{\inn},L_{\out})$-$\DLLR$  with preprocessing time  
$$ \Tm_{\out}^{(Pre)} \cdot \left(\Tm_{\inn} + L_{\inn} \cdot T_{\inn}^{(unenc)}+\poly(N_{\inn}, \log(N_{\out})) \right)  ,$$ 
preprocessing space 
$$ \Sp_{\out}^{(Pre)} + \Sp_{\inn} +  \poly(N_{\inn}, \log(N_{\out})) + N_{\inn} \cdot (\ell_{\inn}+L_{\inn}) \cdot \log(|\Sigma_{\inn}|) + \Sp_{\inn}^{(unenc)},$$
evaluation time 
$$ N_{\inn} \cdot (\Tm_{\out}^{(Eval)} \cdot \left(\Tm_{\inn} + L_{\inn} \cdot T_{\inn}^{(unenc)}+ \poly(N_{\inn}, \log(N_{\out})) \right)  + \Tm_{\inn}^{(enc)}),$$ 
evaluation space $$\Sp_{\out}^{(Eval)} + \Sp_{\inn}^{(enc)} + \Sp_{\inn}  + 
N_{\inn} \cdot (\ell_{\inn} + L_{\inn}) \cdot \log(|\Sigma_{\inn}|)+ 
\poly( N_{\inn}, \log(N_{\out})) + \Sp_{\inn}^{(unenc)},$$
and output length $\Spop_{\out}+O(1)$. 
\end{restatable}

\section{Deterministic Approximate Local List-Recovery Algorithm}\label{sec:alg}
\label{sec:DALLR}

In this section we introduce our DALLR algorithm for tensor codes.  
We present the DALLR algorithm in \Cref{sec:DALLRC_tensor} below. The algorithm uses a space-efficient testing procedure for tensor codes whose description is deferred to \Cref{subsec:test}. The analysis of the DALLR algorithm is given in \Cref{sec:main}.

\subsection{DALLR Algorithm for Tensor Codes}\label{sec:DALLRC_tensor}

Our goal in this section is to present our DALLR algorithm for a tensor code $C^{\otimes t}$.  To this end, in what follows, fix a linear code $C \subseteq \Sigma^n$ with distance at least $\delta(C)$ that is  $(\rho_0,\ell,L_0)$-globally list-recoverable using an algorithm $\cA_0$. Further assume that $C \otimes C$ has a unique decoder $\mathrm{Dec}_{C \otimes C}$ that can uniquely decode $C \otimes C$ up to radius $\rho^{(!)}_2 \in \left(0, \frac{(\delta(C))^2} 2 \right)$.

Recall that a DALLR algorithm for $C^{\otimes t}$ is a deterministic algorithm $\DALLR{t}$ that outputs a list $A_1^{(t)}, \ldots, A_L^{(t)}$ of deterministic local algorithms.  (We decorate the algorithms with $\cdot^{(t)}$, where $t$ is the order of the tensor, because we will define them recursively relative to the algorithms for lower-order tensors.) Each local algorithm 
$A_j^{(t)}$ receives as input a codeword entry $i=(i_1, \ldots, i_t) \in [n]^t$, and also gets  oracle access to a string of input lists
$\mathcal{S} \in {\Sigma \choose \leq \ell}^{[n]^t}$, and outputs a value in $\Sigma$. The guarantee is that for any codeword $c \in C^{\otimes t}$ that agrees with many of the input lists, there is some local algorithm $A_j^{(t)}$ that correctly recovers most of the entries of $c$. 

We start by describing the DALLR algorithm $\cApre^{(2)}$ for $C \otimes C$ in Section \ref{sec:DALLRC_tensor_base} below, and then recursively describe the algorithm $\cApre^{(t)}$ for $t\geq 3$ in Section \ref{sec:DALLRC_tensor_recurse}.  In fact, we will see that $C \otimes C$ is actually a DLLRC (that is, a DALLRC with $\eps = 0$).  The approximation parameter $\eps_t$ will start small and grow with $t$. 

\subsubsection{Base case: DLLR Algorithm for $C \otimes C$}\label{sec:DALLRC_tensor_base}

We first describe the DLLR algorithm $\DALLR{2}$ for $C \otimes C$. Recall that $C \otimes C \subseteq \Sigma^{[n] \times [n]}$ is a code whose codewords are all $[n] \times [n]$ matrices so that all their rows and columns belong to the base code $C$.

To describe the algorithm $\DALLR{2}$, we first need to explain how we will represent the local algorithms $A^{(2)}$ that $\DALLR{2}$ generates. The representation for a local algorithm $A^{(2)}$  
is  given by $$\rep(A^{(2)}) = (\sigma^{(2)}, \tau=(\tau_1, \ldots, \tau_{m_2})) \in \{0,1\}^{r_2} \times [L_0]^{m_2},$$
where $r_2,m_2$ will be determined later in \Cref{def:params2}. Intuitively, the seed $\sigma^{(2)}$ is used to select a subset of $m_2$ columns out of the $n$ columns of $C \otimes C$ (using the sampler given by Definition \ref{def:samplerGamma}), while each $\tau_a \in [L_0]$ is used for selecting a list element out of the $L_0$ list elements output by the global list recovery algorithm for $C$, when applied to the $a$-th selected column.
The pre-processing algorithm $\DALLR{2}$ iterates over all possible choices of $\sigma^{(2)} \in \{0,1\}^{r_2}$ and $\tau \in [L_0]^{m_2}$, and outputs a corresponding local algorithm $A^{(2)}$ represented by $(\sigma^{(2)}, \tau)$ for each such choice. 

The local algorithm $A^{(2)}$ represented by $(\sigma^{(2)}, \tau)$ operates (roughly) as follows.
First it applies the sampler $\Gamma$ given by Definition \ref{def:samplerGamma} on the seed 
$\sigma^{(2)}$ to select $m_2$ columns out of the $n$ columns. Then, on input $i=(i_1,i_2)$, $A^{(2)}$ runs the global list recovery algorithm for $C$ on each of the $m_2$ selected columns, and uses the advice $\tau$ to pick one output list element  for each of these columns; Let  $c^{(1)}, \ldots,c^{(m_2)}$ denote the selected codewords.
Finally, $A^{(2)}$ runs the global list recovery algorithm for $C$ once more on the $i_1$-th row and  chooses the codeword $\tilde c$ from the output list that agrees the most with $c^{(1)}, \ldots,c^{(m_2)}$. The output of $A^{(2)}$ is the $i_2$-entry of $\tilde c$.

Next we present the formal description of the DLLR algorithm $\DALLR{2}$ and the local algorithms $A^{(2)}$ that it outputs. We start with the description of the local algorithms below in \Cref{alg:local-alg-two}. The algorithm follows the high-level description above, but includes an additional \emph{unique decoding step}; the reason for this is so that $\DALLR{2}$ becomes a DLLR algorithm, rather than a DALLR algorithm.  This means that in our recursion, we can treat $\DALLR{2}$ as a DALLR algorithm with arbitrarily small approximation parameter $\epsilon_2$.  (Briefly, this is needed because the approximation parameter will start very small, and grow to reach $\eps$ for our final code $\Cot$).  We note that this unique decoding step is only for the base case of $t=2$, as we do not have space to do it for larger $t$.

\begin{algorithm}[H]
\caption{$A^{(2)}$: Local alg. for $C \otimes C$ represented by $\rep(A^{(2)}) = \left(\sigma^{(2)}, \tau\right)$}
\label{alg:local-alg-two}
\begin{algorithmic}[1]
\Require This algorithm is parameterized by a seed $\sigma^{(2)} \in \{0,1\}^{r_2}$, and a vector $\tau = (\tau_1, \ldots, \tau_{m_2}) \in [L_0]^{m_2}$, which are both on the input tape.  On the input tape, it also gets an index $i = (i_1, i_2) \in [n]\times [n]$, and the input lists $\cS \in {\Sigma \choose \leq \ell}^{[n]\times [n]}$.  Finally, it can call a sampler $\Gamma_2$ that takes seeds in $\{0,1\}^{r_2}$ and samples $m_2$ elements of $[n]$; and a deterministic $(\rho_0, \ell, L_0)$-global list-recovery algorithm $\mathcal{A}_0$ for $C$; and a unique decoding algorithm $\mathrm{Dec}_{C \otimes C}$ for $C \otimes C$ that works up to radius $\rho_2^{(!)}$.
\Ensure Outputs an element of $\Sigma$.  \Comment{The local algorithm $A^{(2)}$ implicitly corresponds to some string $x \in \Sigma^{[n]\times [n]}$; the output is supposed to be $x_i$.}
\State $H_2 = \{h_1, \ldots, h_{m_2}\} \gets \Gamma_2(\sigma^{(2)})$ \Comment{$H_2 \subseteq [n]$; we view each $h_a$ as an index of a column of $C \otimes C$}
\For{$a \in [m_2]$}
    \State Let $\cS_a = \cS|_{[n] \times \{h_a\}}$ be the restriction to the $h_a$'th column of $\cS$.
    \State Call $\cA_0$ with input $\cS_a$ to get a list $\mathcal{L}_a \subseteq C$ with $|\cL_a| \leq L_0$.
    \State Let $c^{(a)} \in C$ be the $\tau_a$'th element of $\cL_a$.  \Comment{If $|\cL_a| < \tau_a$, set $c^{(a)} = \bot^n$, for some $\bot \not\in \Sigma$.}
\EndFor
\For{$b \in [n]$}
\State Run $\mathcal{A}_0$ on $\cS|_{\{b\} \times [n]}$ to get a list $\mathcal{K}_{b} = \{ c \in C \,:\, \delta(c, \cS_{\{b\}\times [n]}) \leq \rho_0 \} \subseteq C.$
\State Let $\tilde{c}^{(b)} = \argmin_{c \in \mathcal{K}_{b}} |\{ a \in [m_2] : c^{(a)}_{b} \neq c_{h_a} \}|$.
\Comment{If $\mathcal{K}_b = \emptyset$, set $\tilde{c}^{(b)} = \bot^n$}
\EndFor
\State Let $w \in \Sigma^{[n] \times [n]}$ be the matrix whose $b$'th row is given by $\tilde{c}^{(b)}.$
\State Compute $c^* = \mathrm{Dec}_{C \otimes C}(w) \in C \otimes C$.  \Comment{Unique decoding step}
\State Return $c^*_{i}$.
\end{algorithmic}
\end{algorithm}

The preprocessing algorithm $\DALLR{2}$ is formally described in Algorithm \ref{alg:main_base} below. The algorithm follows the high-level description above, but includes an additional \emph{testing step}; the goal of this step is to prune the output list and eliminate spurious local algorithms which do not correspond to a close-by codeword, as well as duplicates.
The parameter $\rho_2$ used in the testing procedure will be defined later in Definition \ref{def:params}.

\begin{algorithm}[H]
\caption{\DALLR{2}: DLLR algorithm for $C \otimes C$}\label{alg:main_base}
\label{alg:approx-list-recovery_base}
\begin{algorithmic}[1]
\Require On the input tape: input lists $\cS \in {\Sigma \choose \leq \ell}^{[n] \times [n]}$;
         parameters $m_2, r_2, L_0$ and $\rho_2$.
\Ensure Outputs a list $\ctL^{(2)}$ of representations of local algorithms.
\State $\ctL^{(2)} \gets \emptyset$
\ForAll{seeds $\sigma^{(2)} \in \{0,1\}^{r_2}$}
    \ForAll{$\tau \in [L_0]^{m_2}$}
        \State $\rep(A^{(2)}) = (\sigma^{(2)}, \tau)$
        \State $\mathrm{dist} \gets \frac{1}{n^2}\sum_{i \in [n]\times [n]} \mathbf{1}[A^{(2)}(i) \not\in S_i]$ \Comment{Check that $x(A^{(2)})$ is close to $\cS$}
        \State $\mathrm{sim} \gets \min_{ \rep(\tilde{A}^{(2)}) \in \ctL^{(2)}} \sum_{i \in [n]\times [n]} \mathbf{1}[ A^{(2)}(i) \neq \tilde{A}^{(2)}(i) ]$ 
        \Comment{Check for duplicates; if $\ctL^{(2)} = \emptyset$ then $\mathrm{sim} \gets \infty$}
        \If{$\mathrm{dist} \leq \rho_2$ and $\mathrm{sim} > 0$}
            \State $\ctL^{(2)} \gets \ctL^{(2)} \cup \{\rep(A^{(2)})\}$.
        \EndIf
  \EndFor
\EndFor
\State \Return $\ctL^{(2)}$ 
\end{algorithmic}
\end{algorithm}

\subsubsection{DALLR algorithm for $C^{\otimes t}$}\label{sec:DALLRC_tensor_recurse}

Next we recursively describe the DALLR algorithm $\DALLR{t}$ for $C^{\otimes t}$ for $t \geq 3$. To this end, we view $C^{\otimes t}$ as a $2$-dimensional tensor $C^{\otimes (t-1)} \otimes C$, so the columns are codewords of $C^{\otimes (t-1)}$ while the rows are codewords of $C$. At a high level, the algorithm $\DALLR{t}$ and the local algorithms $A^{(t)}$ it generates follow the same footprint as the $t=2$ case, except that we recursively use the corresponding algorithms for $C^{\otimes (t-1)}$ for list recovering the columns, and we use low-space testers. 

In more detail, we first explain how we represent the local algorithms $A^{(t)}$ that $\DALLR{t}$ generates. 
Informally, the representation for a local algorithm $A^{(t)}$ is given by
\begin{itemize}
    \item a seed $\sigma^{(t)} \in \{0,1\}^{r_t}$ for some parameter $r_t$; and
    \item a list of $m_t$ local algorithms $A_1^{(t-1)}, \ldots, A_{m_t}^{(t-1)}$, where each $A_a^{(t-1)}$ is a (representation of a) local algorithm for $C^{\otimes (t-1)}$, 
\end{itemize} 
where $r_t$ and $m_t$ will be determined later in \Cref{def:params2}.
Formally, we let
\[ \rep(A^{(t)}) = \inparen{\sigma^{(t)}, \rep(A_1^{(t-1)}), \ldots, \rep(A_{m_t}^{(t-1)}) },\]
where $\sigma^{(t)} \in \{0,1\}^{r_t}$ and each of the $\rep(A_i^{(t-1)})$ are defined recursively. 
As in the $t=2$ case, the seed $\sigma^{(t)}$ will be used for selecting $m_t$ columns out of the $n$ columns of $C^{\otimes (t-1)} \otimes C$ using the sampler given by Definition \ref{def:samplerGamma}. The representation $\rep(A_a^{(t-1)})$ will be used for selecting a  local algorithm out of the list of local algorithms output by the DALLR algorithm $\DALLR{t-1}$, when applied to the $a$-th selected column.

We now proceed to the formal description of $\DALLR{t}$, and the local algorithms $A^{(t)}$ that it outputs. We start with the definition of the local algorithms below in \Cref{alg:local-alg}. The algorithm is similar to the $t=2$ case, except that we apply on each column recursively the local algorithms for $C^{\otimes (t-1)}$ instead of the global list recovery algorithm for $C$. Additionally, for $t \geq 3$ we do not include the final unique decoding step, as it would require too much space.

\begin{algorithm}[H]
\caption{$A^{(t)}$: Local alg. for $C^{\otimes t}$ with $\rep(A^{(t)}) = \left(\sigma^{(t)}, \rep(A_1^{(t-1)}), \ldots, \rep(A_{m_t}^{(t-1)})\right)$}\label{alg:rep}
\label{alg:local-alg}
\begin{algorithmic}[1]
\Require This algorithm is parameterized by a seed $\sigma^{(t)} \in \{0,1\}^{r_t}$ and by representations of local algorithms $A_a^{(t-1)}$ for $C^{\otimes (t-1)}$, for $a \in [m_t]$; this is given on the input tape.  
On the input tape, it also gets an index $i = (i_1, \ldots, i_t) \in [n]^t$, and the input lists $\cS  \in {\Sigma \choose \leq \ell}^{[n]^t}$.  Finally, it can call a sampler $\Gamma_t$ that takes seeds in $\{0,1\}^{r_t}$ and samples $m_t$ elements of $[n]$; and a deterministic  $(\rho_0, \ell, L_0)$-global list-recovery algorithm $\mathcal{A}_0$ for $C$.
\Ensure Outputs an element of $\Sigma$.  \Comment{The local algorithm $A^{(t)}$ implicitly corresponds to some string $x \in \Sigma^{[n]^t}$; the output is supposed to be $x_i$.}
\State Let $i' = (i_1, \ldots, i_{t-1}) \in [n]^{t-1}$, so $i = (i', i_t)$.
\State $H_t = \{h_1, \ldots, h_{m_t}\} \gets \Gamma_t(\sigma^{(t)})$ \Comment{$H_t \subseteq [n]$;  each $h_a$ indexes a column of $C^{\otimes (t-1)} \otimes C$}
\For{$a \in [m_t]$}
    \State Let $\cS_a = \cS|_{[n]^{t-1} \times \{h_a\}}$ be the restriction to the $h_a$'th ``column'' of $\cS$. 
    \State Call $A_a^{(t-1)}$ with the input $i'$ and with oracle access to $\cS_a$ to get $v_a \in \Sigma$.
\EndFor
\State Run $\mathcal{A}_0$ on $\cS|_{\{i'\} \times [n]}$ to get a list $\mathcal{K}_{i'} = \{ c \in C \,:\, \delta(c, \cS_{\{i'\}\times [n]}) \leq \rho_0 \} \subseteq C.$
\State Let $c^* = \argmin_{c \in \mathcal{K}_{i'}} |\{ a \in [m_t] : c_{h_a} \neq v_a \}|$.
\State Return $c^*_{i_t}$.
\end{algorithmic}
\end{algorithm}

The preprocessing algorithm $\DALLR{t}$ is formally described in Algorithm \ref{alg:main} below. Similarly to the $t=2$ case, the algorithm iterates over all choices for seed and representations of local algorithms for $C^{\otimes (t-1)}$. For a large $t$, we cannot afford for the local algorithm to have a unique decoding step, as in the $t=2$ case. 
Instead we now use a testing procedure \Call{Test}{} in the pre-processing algorithm that allows us to skip the unique decoding step in the local algorithm.  \Call{Test}{} can be implemented in low-space, and is based on local testing of tensor codes \cite{Vid2015}. The formal description of the testing procedure is deferred to Section \ref{subsec:test} below (cf. Algorithm \ref{alg:test}).

\begin{algorithm}[H]
\caption{\DALLR{t}: DALLR  alg. for $C^{\otimes t} = C^{\otimes (t-1)} \otimes C$, assuming such an algorithm for $C^{\otimes (t-1)}$}\label{alg:main}
\label{alg:approx-list-recovery}
\begin{algorithmic}[1]
\Require On the input tape: input lists $\cS \in {\Sigma \choose \leq \ell}^{[n]^{t-1} \times [n]}$.  This algorithm can call
         an $(n, \eta_t, \nu_t)$-sampler $\Gamma_t$ with seed length $r_t$,  and sample size $m_t$ (see \Cref{def:params2} for how to set parameters); and a \ADLLR\ 
        algorithm $\DALLR{t-1}$ for $C^{\otimes (t-1)}$.
\Ensure Outputs a list $\ctL^{(t)}$ of representations of local algorithms.
\State $\ctL^{(t)} \gets \emptyset$
\ForAll{seeds $\sigma^{(t)} \in \{0,1\}^{r_t}$}
  \State $H_t = \{h_1,\ldots,h_{m_t}\} \gets \Gamma_t(\sigma^{(t)})$
        \Comment{$H_t \subseteq [n]$; each $h_a$ indexes a column of $C^{\otimes (t-1)} \otimes C$}
  \For{$a = 1,\ldots,m_t$}
    \State Run $\DALLR{t-1}$ on the column $\cS|_{[n]^{t-1} \times \{h_a\}}$;
           let $\cL^{(t-1)}_{a}$ be the output list of representations of local algorithms for $C^{\otimes (t-1)}$.
  \EndFor
  \ForAll{tuples $\vec{A} = (\rep(A^{(t-1)}_1),\ldots,\rep(A^{(t-1)}_{m_t})) \in \cL^{(t-1)}_{1} \times \cdots \times \cL^{(t-1)}_{m_t}$}
    \State $\rep(A^{(t)}) = (\sigma^{(t)}, \vec{A})$ \label{line:localalg}
    \If{\Call{Test}{$\rep(A^{(t)}$), $\ctL^{(t)}$}}
        \State $\ctL^{(t)} \gets \ctL^{(t)} \cup \{\rep(A^{(t)})\}$.
    \EndIf
  \EndFor
\EndFor
\State \Return $\ctL^{(t)}$ 
\end{algorithmic}
\end{algorithm}

\subsection{Testers to Prune the List}\label{subsec:test}

In this section we describe the space-efficient testing procedure \Call{Test}{} used in Algorithm \ref{alg:main}, and analyze its time and space complexity.
\Call{Test}{} will consist of three different testers whose goal is to eliminate spurious local algorithms and keep the list size small. 

Before we present our testing algorithms, we set a few parameters.      Recall that in Section \ref{sec:DALLRC_tensor} we have fixed a linear code $C \subseteq \Sigma^n$ with  distance at least $\delta(C)$ that is  $(\rho_0,\ell,L_0)$-globally list-recoverable using an algorithm $\cA_0$, and we have further assumed that $C \otimes C$ has a unique decoder $\mathrm{Dec}_{C \otimes C}$ that can uniquely decode $C \otimes C$ up to radius $\rho^{(!)}_2 \in \left(0, \frac{(\delta(C))^2} 2 \right)$. 

\begin{definition}[Setting parameters $\gamma_t$,  $\epsilon_t$, and $\rho_t$]\label{def:params}    
    For a positive integer $t$, define the \emph{testing parameter}
    \begin{equation}\label{eq:defgamma}
        \gamma_t = \frac{(\delta(C))^{2t}}{18^{\log_{1.5}(t)}}.
    \end{equation}      
Our objective will be to show that $C^{\otimes t}$ is $(Q_t,\epsilon_t, \rho_t, \ell, L_t)-\ADLLR$, where 
  $\{\eps_t\}_t$ and $\{\rho_t\}_t$ are defined recursively as follows (we shall defer the definitions of the parameters $Q_t$ and $L_t$ to Definition \ref{def:params2}). 

    For the base case of $t=2$, 
    let
  \begin{equation} \label{eq:rho2}
            \rho_2 = \rho_0 \cdot \min \left\{ \frac{ \delta(C)}{8}, \frac{ \rho_2^{(!)}}{2}, \kappa^{-8} \right\}, 
        \end{equation}  
        where $\kappa = \kappa(\frac 1 {\delta(C)}, \frac 1 {\rho_0})$ is the parameter from \Cref{cor:HRW}, 
      and let $\epsilon_2>0$ be a parameter.

   Finally, let $d_0\geq 14$ be an absolute constant, and for any $t \geq 3$, define
    \begin{equation}\label{eq:epsessrhot}
    \eps_{t} = \frac{2 d_0 \eps_{t-1}}{\gamma_t \delta(C) \rho^{(!)}_2} \qquad \text{and} \qquad \rho_t = \eps_{t-1}\rho_{t-1}.
   \end{equation}

   \end{definition}

Each local algorithm $A^{(t)}$ that we want to test corresponds to a string $x$ that should be in the true output list.  For simplicity, we first define and analyze testers $T_1, T_2, T_3$ that are functions of \emph{strings} $x$, rather than as functions of representations of algorithms $\rep(A^{(t)})$. 
After that, we will show how to implement these three testers to run on the representations $\rep(A^{(t)})$.

\subsubsection{Distance Approximators}
For the first of our testers, we will need to solve the following problem: Given a string $x \in \Sigma^{[n]^t}$, estimate the distance from $x$ to $C^{\otimes t}$, efficiently and using low space.  We will do this in two steps: First we will use an estimator of \cite{Vid2015} to estimate the distance from two-dimensional restrictions of $x$ to $C \otimes C$; then we will average these to obtain an estimator for the distance from $x$ to $C^{\otimes t}$.  

For the two-dimensional estimator, we use the following lemma.

\begin{lemma}[\cite{Vid2015}, Claim 6.2]\label{lem:rowcolumn}
Let $C \subseteq \Sigma^n$ be a linear code. 
    For $x \in \Sigma^{[n] \times [n]},$ let
 $$
    \delta_{RC}(x) := \frac{1}{2n}\left( \inabs{\inset{i \in [n] : x|_{\{i\} \times [n]} \not\in C }} + \inabs{\inset{j \in [n] : x_{[n] \times \{j\}} \not\in C}}\right).
$$
    Then $\frac{1}{2}\delta(x, C \otimes C) \leq \delta_{RC}(x) \leq 1$.  Further, if $C$ has deterministic algorithm $\mathrm{Detect}_C$ running in time $\Tm_C$ and space $\Sp_C$ so that $\mathrm{Detect}_C(x)$ returns \textsc{True} if  $x \in C$, and \textsc{False} otherwise, then there is a deterministic algorithm to compute $\delta_{RC}(x)$ that runs in time $O(n \Tm_C)$ and space $O(\Sp_C + \log n)$.  
\end{lemma}

\begin{remark}
Note that $\delta_{RC}(x)$ is the rejection probability of the 
 (randomized) local testing algorithm for $C \otimes C$, which given oracle access to a string $x \in \Sigma^{[n] \times [n]}$, picks 
 a random row or a random column of $x$, with probability $\frac 1 2$ each, and accepts if and only if that row or column belongs to $C$.  Claim 6.2 of \cite{Vid2015}, which the author refers to as folklore, shows that the rejection probability of this tester is at least  $\frac{1}{2}\delta(x,C\otimes C)$, which in particular implies that this algorithm is a (strong) local tester for $C\otimes C$.  We note that it is known that this local tester is generally not \emph{robust} \cite{Valiant05}, where a robust local tester is a local tester in which the local view of the tester is far from an accepting view on average if the string is far from the code; However, for our purposes we shall not require the stronger robustness property out of this local tester.

Finally, we note that the (deterministic) algorithm referred to in \Cref{lem:rowcolumn} is to iterate over all rows and columns of $x \in \Sigma^{[n] \times [n]}$ and to return the fraction of them that belong to $C$, in order to compute $\delta_{RC}(x)$.  In particular, it is not a local testing algorithm, which would be randomized. 
\end{remark}

We will also use the following lemma about the average over all such two-dimensional estimators.

\begin{theorem}[\cite{Vid2015}, Implicit in the proof of Theorem 3.1]\label{thm:viderman_tester}
    Let $t \geq 3$, and let $C \subseteq \Sigma^n$ be a linear code, and let $x \in \Sigma^{[n]^t}$.  Then
    \[\gamma_t \cdot \delta(x, \Cot) \leq \frac{1}{{t \choose 2}n^{t-2}} \sum_P \delta(x|_P, C\otimes C) \leq \delta(x, \Cot),\]
    where $\gamma_t$ is as in \Cref{eq:defgamma}, and where the sum is over all ${t \choose 2}n^{t-2}$ two-dimensional axis-aligned planes $P \subseteq [n]^t$;  that is, over all planes $P$ of the form 
    \begin{equation}\label{eq:planes} P = \{a_1\} \times \{a_2\} \times \cdots \times \{a_{i-1}\} \times [n] \times \{a_{i+1}\} \times \cdots \times \{a_{j-1}\} \times [n] \times \{a_{j+1}\} \times \cdots \times \{a_t\},\end{equation}
    where $i < j \in [t]$ and $a_r \in [n]$ for $r \in [t]\setminus\{i,j\}$.
\end{theorem}

\begin{remark} 
The above lemma implies that for $t \geq 3$, the (randomized) local testing algorithm for $C^{\otimes t}$,  which given oracle access to a string $x \in \Sigma^{[n]^t}$, picks a 
random two-dimensional axis-aligned plane, and accepts if and only if the restriction to that plane belongs to $C\otimes C$, is a \emph{robust} local tester.
\end{remark}

To obtain an estimator to the distance of a string $x$ to $C^{\otimes t}$, our basic idea is
 to use $\delta_{RC}$ to estimate $\delta(x|_P, C\otimes C)$ for each $P$, as per \Cref{lem:rowcolumn}, and then average them to estimate $\delta(x, \Cot)$ as per \Cref{thm:viderman_tester}.
We can improve on this idea by noting that if $x|_P$ is sufficiently close to $C\otimes C$, then we can do better by decoding $x|_P$ to the nearest codeword and computing $\delta(x|_P, C\otimes C)$ exactly.   This modified two-dimensional approximation, $\tilde{\delta}_{C\otimes C}(x)$, appears in \Cref{alg:appx2} below.
\begin{algorithm}[H]
\caption{$\tilde{\delta}_{C\otimes C}$: Approximates distance to $C\otimes C$}\label{alg:appx2}
    \begin{algorithmic}[1]
        \Require{Input $x \in \Sigma^{[n] \times [n]}$; a unique decoder $\mathrm{Dec}_{C\otimes C}$ for $C\otimes C$ that works up to radius $\rho_2^{(!)}$, and an error detection algorithm $\textsc{Detect}_{C}$ for $C$.}
        \Ensure{Output an estimate $\tilde{\delta}_{C\otimes C}(x)$ for $\delta(x, C\otimes C)$.}
        \State $y \gets \mathrm{Dec}_{C\otimes C}(x)$
        \Comment{If decoding fails, $y \gets \bot^{n\times n}$}
        \State Use \textsc{Detect}$_C$ on all rows and columns of $y$ to decide if $y \in C \otimes C$.
        \If{$ y \in C\otimes C$ \textbf{and} $\delta(x,y) \leq \rho_2^{(!)}$}{
        \Return{$\delta(x,y)$}}
        \Else{
        \Return{$\delta_{RC}(x)$}
        }
        \EndIf
    \end{algorithmic}
\end{algorithm}
We note that \Cref{alg:appx2} can be implemented in time $O(nT_C + T_{C\otimes C})$ and space $O( \Sp_C + n^2 + \Sp_{C\otimes C})$, where $T_{C\otimes C}$ and $\Sp_{C\otimes C}$ are the time and space required to run $\mathrm{Dec}_{C\otimes C}$ respectively.
Indeed, in addition to the time and space required to compute $\delta_{RC}(x)$ from the above, we also account for the $O(n^2)$ space required to store $y$.  We can check whether or not $y \in C \otimes C$ in time $n \Tm_{C}$ and space $\Sp_C$ by checking whether each row and column belongs to $C$.  Finally, the time to compute $\delta(x,y)$ is $O(n^2)$, which is dominated by the $n T_{C}$ term.

Now we can define our approximator for $\delta(x, C^{\otimes t})$, which we call $\tilde{\delta}_{\Cot}$:

\begin{definition}\label{def:feasible_appx}
    Let $C \subseteq \Sigma^n$ be a linear code.  For any $x \in \Sigma^{[n]^t}$, define
    \[ \tilde{\delta}_{\Cot}(x) := \frac{1}{{t \choose 2} n^{t-2}}  \sum_{P} \tilde{\delta}_{C \otimes C}(x|_P),\]
    where the sum is over all two-dimensional axis-aligned planes $P \subseteq [n]^t$, as in \Cref{eq:planes}. 
\end{definition}
\begin{lemma}\label{lem:goodappx}
  Let $\tilde{\delta}_{\Cot}$ be as in \Cref{def:feasible_appx}.  Then for all $x \in \Sigma^{[n]^t}$, 
  \[ \frac{\gamma_t}{2} \delta(x, \Cot) \leq \tilde{\delta}_{\Cot}(x) \leq \frac{1}{\rho_2^{(!)}} \delta(x, \Cot),\]
  where
  $\gamma_t$ is as defined in \Cref{def:params}.
  Moreover, we can compute $\tilde{\delta}_{\Cot}(x)$ in time $O\left( t^2 ( n^{t-1}\Tm_C + n^{t-2}\Tm_{C \otimes C}) \right)$ and space $O\left( \Sp_C + t \log n + n^2 + \Sp_{C \otimes C}\right)$, where $\Tm_C$ and $\Sp_C$ (resp. $\Tm_{C \otimes C}$ and $\Sp_{C \otimes C}$) are the time and space for \textsc{Detect}$_C$ (resp.  $\mathrm{Dec}_{C \otimes C}$).
\end{lemma}
\begin{proof}
   Suppose that $\delta(x|_P, C \otimes C) \leq \rho_2^{(!)}$.  Then $\tilde{\delta}_{C \otimes C}(x|_P) = \delta(x|_P, C \otimes C)$, because in \Cref{alg:appx2}, the algorithm successfully uniquely decodes $C \otimes C$.  On the other hand, if $\delta(x|_P, C \otimes C) > \rho_2^{(!)}$, then 
    \[ \tilde{\delta}_{C \otimes C}(x|_P) = \delta_{RC}(x|_P) \geq \frac{1}{2} \delta(x|_P, C \otimes C),
    \] 
    where the equality is by definition and the inequality is by \Cref{lem:rowcolumn}.  In the case that $\delta(x|_P, C \otimes C) > \rho_2^{(!)}$, we also have
    \[ \rho_2^{(!)} \cdot \tilde{\delta}_{C \otimes C}(x|_P) \leq \rho_2^{(!)} \leq \delta(x|_P, C \otimes C) \] as $\tilde{\delta}_{C \otimes C}(x|_P) \in [0,1]$, which implies that
    \[ \tilde{\delta}_{C \otimes C}(x|_P) \leq \frac{1}{\rho_2^{(!)}} \delta(x|_P, C \otimes C). \]
    Thus, in either case, we have
    \[ \frac{1}{2} \delta(x|_P, C \otimes C) \leq \tilde{\delta}_{C \otimes C}(x|_P) \leq \frac{1}{\rho_2^{(!)}} \delta(x|_P, C \otimes C).\]
    Averaging over all such planes $P$, we have that
    \[ \frac{1}{2} \left( \frac{1}{{t\choose 2}{n^{t-2}}}\sum_P \delta(x|_P, C \otimes C)\right) \leq \tilde{\delta}_{\Cot}(x) \leq \frac{1}{\rho_2^{(!)}}\left(\frac{1}{{t\choose 2}{n^{t-2}}}\sum_P \delta(x|_P, C \otimes C)\right). \]
    Finally applying \Cref{thm:viderman_tester} completes the proof of the bound.

    For the running time and space, we have already established that $\tilde{\delta}_{C \otimes C}$ can be computed in time $O(n \Tm_C + \Tm_{C \otimes C})$ and space $O(\Sp_C + n^2 + \Sp_{C \otimes C}).$  To obtain the running time to compute $\tilde{\delta}_{\Cot}(x)$, we multiply by ${t \choose 2}n^{t-2} = O(t^2 n^{t-2})$, the number of planes $P$ iterated over.  The space is the same as for $\tilde{\delta}_{C \otimes C}$, with an additional $O(t \log(n) + \log(t)) = O(t \log n)$ bits to keep track of the current plane $P$, and another $O(t \log n)$ bits to keep track of the sum $\sum_P n^2 \cdot \tilde{\delta}_{C \otimes C}(x|_P)$. 
\end{proof}

\subsubsection{Three Testers for Strings}
Finally we are ready to define our three testers, $T_1, T_2, T_3$.  As noted above, we define these first as functions of strings $x \in \Sigma^{[n]^t}$, and explain later how to apply them to local algorithms.  The three testers play the three roles described in \Cref{sec:tech}.  That is, we need to test (1) if $x$ is close to some $c \in \Cot$; (2) if $c$ is close to the input lists $\cS$; and (3) if there is some other $x'$ already in the list that corresponds to the same $c$.  For item (2), in fact we test whether or not $x$ is close to $\cS$; if item (1) holds and $x$ is close enough to $c$, then the triangle inequality will imply a good enough approximation to (2).

\begin{definition}[Testers $T_1, T_2, T_3$]\label{def:testers}
    For $t \geq 3$, define
    \begin{equation}\label{eq:thresh} \beta_1^{(t)} = \frac{d_0 \eps_{t-1}}{\delta(C) \rho_2^{(!)}} \qquad 
\beta_2^{(t)} = \rho_t + \frac{d_0 \eps_{t-1}}{\delta(C)} \qquad 
\beta_3^{(t)} = \frac{(\delta(C))^t}{2},\end{equation}
where we recall the notation from \Cref{def:params}. 

For any $x, x' \in \Sigma^{[n]^t}$, define:
\begin{align*}
    T_1(x) &= \mathbf{1}[ \tilde{\delta}_{\Cot}(x) \leq \beta_1^{(t)} ] \qquad \text{(is $x$ close enough to $\Cot$?)}\\
    T_2(x) &= \mathbf{1}[\delta(x, \cS) \leq \beta_2^{(t)}] \qquad \text{(is $x$ close enough to the input $\cS$?)} \\
    T_3(x,x') &= \mathbf{1}[\delta(x,x') \geq \beta_3^{(t)} ]\qquad \text{(is $x$ far enough away from another $x'$?)}
\end{align*}

We note that $T_1, T_2, T_3$ all depend on $t$; we suppress this dependence for notational clarity. 

\end{definition}
Next, we prove a few properties of these testers.  

\begin{lemma}[Properties of $T_1$]\label{lem:T1}
Fix parameters as in \Cref{def:params}, and suppose that
$
\eps_{t} < \frac{(\delta(C))^t}{2}.
$
Let $x \in \Sigma^{[n]^t}$. Then the following hold.
\begin{enumerate}
    \item \textbf{Completeness:} If $\delta(x, \Cot) \leq \frac{d_0 \eps_{t-1}}  {\delta(C)}$, then $T_1(x) = 1$.
    \item \textbf{Soundness:} If $T_1(x) = 1$, then $\delta(x, \Cot) \leq \eps_t < \frac{(\delta(C))^t} 2$, and in particular there is a unique closest codeword $\Phi(x) \in \Cot$ to $x$.
\end{enumerate}
\end{lemma}
\begin{proof}
    We prove each item.
    \begin{enumerate}
        \item \textbf{Completeness:} Suppose that $\delta(x, \Cot) \leq \frac {d_0 \eps_{t-1}} {\delta(C)} = \rho_2^{(!)} \beta^{(t)}_1$.  By \Cref{lem:goodappx}, we have $\tilde{\delta}_{\Cot}(x) \leq \frac{\delta(x, \Cot)}{\rho_2^{(!)}} \leq \beta^{(t)}_1$, which implies that $T_1(x) = 1$.
        \item \textbf{Soundness:} Suppose that $T_1(x) = 1$.  Then by definition $\tilde{\delta}_{\Cot}(x) \leq \beta_1^{(t)}$, so \Cref{lem:goodappx} implies that
        \[ \frac{\gamma_t}{2} \delta(x,\Cot) \leq \tilde{\delta}_{\Cot}(x) \leq \beta_1^{(t)},\]
        hence
        \[ \delta(x, \Cot) \leq \frac{2 \beta_1^{(t)}}{\gamma_t} =  \frac{2d_0 \eps_{t-1}}{\gamma_t \delta(C) \rho_2^{(!)}} = \eps_{t} < \frac{ (\delta(C))^t}{2},\]
        where the last inequality is by assumption.
        Thus, $\delta(x, \Cot)$ is less than half the minimum distance of $\Cot$.  This implies that there is a unique $\Phi(x) \in \Cot$ so that $\delta(x, \Phi(x)) < \frac{ (\delta(C))^t}{2}.$
    \end{enumerate}
\end{proof}

\begin{lemma}[Properties of $T_2$]\label{lem:T2}
Fix parameters as in \Cref{def:params}, and suppose that $\epsilon_t < \frac{ (\delta(C))^t} 2$.  Let $x \in \Sigma^{[n]^t}$, and let $\cS \in {\Sigma \choose \leq \ell}^{[n]^t}$ be a collection of input lists.  
Suppose that $T_1(x) = 1$, and let $\Phi(x) \in \Cot$ be the unique closest codeword to $x$, as guaranteed by \Cref{lem:T1}.  Then the following hold.
\begin{enumerate}
    \item \textbf{Completeness:} Suppose that $\delta(\Phi(x), \cS) \leq \rho_t$ and suppose that $\delta(x, \Phi(x)) \leq \frac {d_0 \eps_{t-1}}  {\delta(C)}.$  Then $T_2(x) = 1$.
    \item \textbf{Soundness:} If $T_2(x) = 1$, then $\delta(x, \cS) \leq \rho_t + \frac{ d_0 \eps_{t-1}} { \delta(C)}.$
\end{enumerate}
\end{lemma}
\begin{proof}
    We prove each item.
    \begin{enumerate}
        \item \textbf{Completeness:} It follows from the triangle inequality that
        \[ \delta(x,\cS) \leq \delta(x, \Phi(x)) + \delta(\Phi(x), \cS) \leq \frac{d_0 \eps_{t-1}} { \delta(C)} + \rho_t = \beta_2^{(t)}, \]
        so $T_2(x) = 1$. 
        \item \textbf{Soundness:} If $T_2(x) = 1$, then by definition $\delta(x,\cS) \leq \beta_2^{(t)} = \rho_t + \frac {d_0\eps_{t-1}} {\delta(C)}.$
    \end{enumerate}
\end{proof}

\begin{lemma}[Properties of $T_3$]\label{lem:T3}
Fix parameters as in \Cref{def:params}. and suppose that
$ \eps_{t} < \frac{(\delta(C))^t}{4}.$ Let $x,x' \in \Sigma^{[n]^t}$. 
Suppose that $T_1(x) = T_1(x') = T_2(x) = T_2(x') = 1$, and let $\Phi(x) \in \Cot$ and $\Phi(x') \in \Cot$ be the closest codewords to $x$ and $x'$, respectively.  Then the following hold.
\begin{enumerate}
    \item \textbf{Completeness:} If $\Phi(x) = \Phi(x')$, then $T_3(x,x') = 0$.
    \item 
    \textbf{Soundness:} If $\Phi(x) \neq \Phi(x')$, then $T_3(x,x') = 1$.
\end{enumerate}
\end{lemma}
\begin{proof}
First, since we are assuming that $T_1(x) = 1$, note that \Cref{lem:T1} implies that \begin{equation}\label{eq:close}
\delta(x, \Phi(x)) = \delta(x, \Cot) \leq \eps_{t},
\end{equation}  and the same is true for $x'$.  Now, we prove each item.
\begin{enumerate}
    \item \textbf{Completeness:}   Suppose that $\Phi(x) = \Phi(x') = c$.  From \eqref{eq:close}, the triangle inequality implies that
    \[ \delta(x,x') \leq \delta(x,c) + \delta(x',c) \leq 2\eps_{t} < \frac{(\delta(C))^t} 2 = \beta_3^{(t)},\]
    using the assumption that $\eps_{t} < \frac{(\delta(C))^t} 4$.  But then by definition $T_3(x,x') = 0$.
    
    \item \textbf{Soundness:}
    If $\Phi(x) \neq \Phi(x')$, then from the triangle inequality we have
    \[ \delta(x,x') \geq \delta(\Phi(x), \Phi(x')) - \delta(x,\Phi(x)) - \delta(x', \Phi(x')),\]
    and from \eqref{eq:close}, the fact that $\delta(\Phi(x), \Phi(x')) \geq \delta(\Cot) = (\delta(C))^t$, and our assumption that $ \eps_{t} < \frac{(\delta(C))^t}{4}$,
    we conclude that
    \[ \delta(x,x') \geq \delta(\Phi(x), \Phi(x')) - 2 \epsilon_t > \delta(\Phi(x), \Phi(x')) - \frac{(\delta(C))^t} 2 \geq \frac{(\delta(C))^t } 2 = \beta_3^{(t)},\]
\end{enumerate}
and so $T_3(x,x')=1$.   
\end{proof}

Now that we have defined and analyzed the testers $T_1, T_2, T_3$, which act on strings, we explain how to implement them on representations of local algorithms $A^{(t)}$ as in \Cref{alg:local-alg}.  

\begin{proposition}\label{prop:tests_of_algs}
    Let $T_1, T_2, T_3$ be as in \Cref{def:testers}.  Then there are algorithms $\Tone, \Ttwo, \Tthree$ that take as an input representations of local algorithms $\rep(A^{(t)})$ (as in \Cref{alg:local-alg}, formed by \Cref{alg:main}), so that for $i \in \{1,2\}$, 
    \[ \textsc{Test}_i(\rep(A^{(t)})) = T_i(x(A^{(t)})),\]
    and so that
    \[ \Tthree(\rep(A^{(t)}), \rep(\tilde{A}^{(t)})) = T_3(x(A^{(t)}), x(\tilde{A}^{(t)})).\]
    Above, for a local algorithm $A^{(t)}$, the string $x(A^{(t)}) \in \Sigma^{[n]^t}$ is as in \Cref{def:alg_to_string}. 
    
 Moreover, these algorithms satisfy the following time and space requirements, where $\Tm_C, \Sp_C$ are the time and space required to detect errors in $C$; $\Tm_{C \otimes C}$ and $\Sp_{C \otimes C}$ are the time and space required to uniquely decode $C \otimes C$ up to radius $\rho_2^{(!)}$.

    \begin{itemize}
        \item \Tone:
        \begin{align*} \text{Time} &= O\left( \Tmeval_{\DALLR{t}} \cdot t^2( n^{t-1} \Tm_C + n^{t-2} \Tm_{C \otimes C})\right)\\
        \text{Space} &=  O\left(\Speval_{\DALLR{t}} + \Sp_C + \Sp_{C \otimes C} + t\log n + n^2\right) \end{align*}
        \item \Ttwo\ and \Tthree:
         \begin{align*} 
         \text{Time} &=  O( n^t \cdot \Tmeval_{\DALLR{t}})\\
        \text{Space} &= O(\Speval_{\DALLR{t}} + t\log n)
        \end{align*}
    \end{itemize}
    Above, $\Tmeval_{\DALLR{t}}$ and $\Speval_{\DALLR{t}}$ are the evaluation time and space of $\DALLR{t}$, respectively.
\end{proposition}
\begin{proof}
    We implement the testers $T_1, T_2, T_3$ by simply evaluating $A^{(t)}(i)$ whenever we would need to query the $i$'th symbol of the input vector.  It is clear that this outputs the correct value for each test, so it remains to consider the running time and space usage.

    First we compute the time and space resources for $\Tone$.  Given access to $x \in \Sigma^{[n]^t}$, computing $T_1(x)$ takes the same time as computing $\tilde{\delta}_{C^{\otimes t}}(x)$, which by \Cref{lem:goodappx} takes time $O( t^2( n^{t-1} \Tm_C + n^{t-2} \Tm_{C \otimes C}))$.  Now to compute $\Tone(\rep(A^{(t)}))$, we compute $T_1(x(A^{(t)}))$ simulating query access to $x(A^{(t)})$ using $A^{(t)}$.  The number of queries to $x(A^{(t)})$ is bounded by the running time, and each one of these queries takes time $\Tmeval_{\DALLR{t}}$, so the total is
    \[ O\left( \Tmeval_{\DALLR{t}} \cdot t^2( n^{t-1} \Tm_C + n^{t-2} \Tm_{C \otimes C})\right). \]
    Similarly, by \Cref{lem:goodappx} the space required to compute $T_1(x)$ is $O(\Sp_C + \Sp_{C \otimes C} + t\log n + n^2)$.  Additionally we need space to compute one query to $x(A^{(t)})$ at a time (if $\Tone$ needs to store the response to multiple queries, that is accounted for in the space to compute $T_1(x)$), for a total of
    \[ O(\Speval_{\DALLR{t}} + \Sp_C + \Sp_{C \otimes C} + t\log n + n^2).\]
    For $\Ttwo$, we need to compute $\delta(x(A^{(t)}),\cS)$, and for $\Tthree$, we need to compute $\delta(x(A^{(t)}), x(\tilde{A}^{(t)}))$ for two local algorithms $A^{(t)}$ and $\tilde{A}^{(t)}$.  Each of these require $O(n^t)$ queries to $x(A^{(t)})$ and $x(\tilde{A}^{(t)})$, and each query takes time $\Tmeval_{\DALLR{t}}$, so for both the time is $O( n^t \Tmeval_{\DALLR{t}})$.  Similarly, for both the space is $O(\Speval_{\DALLR{t}} + t\log n)$, as we need space to compute a query to $x(A^{(t)})$ and to store a counter.
\end{proof}

Now, we will show how to put the three tests together to form an algorithm \textsc{Test}.  This is presented in \Cref{alg:test} below.  As mentioned above, the point of the algorithm \textsc{Test} will be to filter out spurious local algorithms, so that the final output list does not grow too large.
\begin{algorithm}[H]
\caption{\textsc{Test}: Decide whether to add a candidate to  output list.}
\label{alg:test}
\begin{algorithmic}[1]
\Require{Representation of a local algorithm $\rep(A^{(t)})$; list $\tilde{L}^{(t)}$ of representations of other local algorithms; access to algorithms $\textsc{Test}_i$ as in \Cref{prop:tests_of_algs}.}
\Ensure{Returns True if $\rep(A^{(t)})$ should be added to the list $\ctL^{(t)}$ and False otherwise.}

  \If{$\Tone(\rep(A^{(t)})) = 0$} \Return False \Comment{$x(A^{(t)})$ is not close enough to $C^{\otimes t}$}
  \EndIf
  \If{$\Ttwo(\rep(A^{(t)})) = 0$} \Return False \Comment{$x(A^{(t)})$ is not close enough to $\cS$}
  \EndIf
  \For{$\rep(\tilde{A}^{(t)}) \in \ctL^{(t)}$}
    \If{$\Tthree(\rep(A^{(t)}), \rep(\tilde{A}^{(t)})) = 0$} \Return False \Comment{too close to an existing list element}
    \EndIf
  \EndFor
  \State \Return True
\end{algorithmic}
\end{algorithm}

\begin{corollary}\label{cor:test_resources}
    Using the same notation as in \Cref{prop:tests_of_algs}, the algorithm \textsc{Test} (\Cref{alg:test}) running on representations of local algorithms $A^{(t)}$ can be implemented to use time 
    \[ \Tm_{\Testt} = O\inparen{\Tmeval_{\DALLR{t}} \cdot \inparen{t^2( n^{t-1} \Tm_C + n^{t-2}\Tm_{C \otimes C}) + |\ctL^{(t)}| n^t}} \]
    and space 
    \[ \Sp_{\Testt} = O\inparen{\Speval_{\DALLR{t}} + \Sp_{C} + \Sp_{C \otimes C} + t\log n + n^2}.\]
\end{corollary}
\begin{proof}
    The statement follows from adding together the bounds in \Cref{prop:tests_of_algs}.
\end{proof}

\section{Analysis of Deterministic Approximate Local List-Recovery Algorithm}\label{sec:main}

In this section we analyze the correctness and time and space requirements of \Cref{alg:main}. To this end, we first analyze the time and space requirements of the local algorithms $A^{(t)}$ in Section \ref{subsec:time_space_local} below. Then in Sections \ref{sec:inductive_step} and \ref{sec:unwind} we analyze the correctness and time and space requirements of the pre-processing algorithm $\DALLR{t}$. 

\subsection{Time and space requirements of local algorithms}\label{subsec:time_space_local}

We begin with analyzing the time and space requirements of the local algorithms $A^{(t)}$ (Algorithm \ref{alg:local-alg}).

\begin{proposition}[Time and space requirements of local algorithms]\label{prop:local_timespace}
    Let $t \geq 3$, and let $A^{(t)}$ be a local algorithm, as in \Cref{alg:local-alg}.  
    Suppose that the running time of the sampler $\Gamma_i$ is $\Tm_{\Gamma_i}$, and let $\Tm_{\Gamma} = \max_{2 \leq i \leq t} \Tm_{\Gamma_i}$.
    Let $M_t = \prod_{j=2}^t m_j$, let $\widetilde{M}_t = \sum_{j=2}^t m_j$, and let $\tilde{r}_t = \sum_{j=2}^t r_j$.
    
    If $\rep(A^{(t)})$ is the representation of a local algorithm output by a DALLR algorithm $\DALLR{t}$, recall that the time and space used by $A^{(t)}$ is denoted $\Tmeval_{\DALLR{t}}$ and $\Speval_{\DALLR{t}}$ respectively.  Then we have:
    \begin{itemize}
        \item $\Tmeval_{\DALLR{t}} = O\inparen{M_t\left( t(\Tm_{\Gamma} + n\Tm_{\cA_0} + L_0) + \Tm_{\mathrm{Dec}_{C \otimes C}}\right)}$.
        \item $\Speval_{\DALLR{t}} = O\inparen{\widetilde{M}_t\log (n|\Sigma|) + t\inparen{\Sp_{\cA_0} + \Tm_{\Gamma} + L_0 n^2 \log|\Sigma|} + \Sp_{\mathrm{Dec}_{C \otimes C}}}$.
    \end{itemize}
    Further, the space required to store $\rep(A^{(t)})$ is
    \[ \mathrm{Len}_{A^{(t)}} = O\inparen{ M_t (\tilde{r}_t + \log L_0 )}.  \]

    Above, for an algorithm $\mathcal{D}$, $\Tm_{\mathcal{D}}$ denotes its running time and $\Sp_{\mathcal{D}}$ denotes its space.    
\end{proposition}
\begin{proof}
    We begin with the running time.  From \Cref{alg:local-alg}, we see that
    \[ \Tmeval_{\DALLR{t}} = \Tm_{\Gamma_t} + \Tm_{\cA_0} + m_t \Tmeval_{\DALLR{t-1}} + O(m_t L_0).\]
    Indeed, \Cref{alg:local-alg} calls each of $\Gamma_t$ and $\cA_0$  once.  It calls the level-$(t-1)$ local algorithms $A^{(t-1)}$ $m_t$ times, and it also takes time $O(m_t L_0)$ to search the list $\mathcal{K}_{i'}$ for $c^*$.  
    
    Unrolling the recursion and using the fact that $\Tm_{\Gamma_i} \leq \Tm_{\Gamma}$ for all $i \leq t$, we have
    \begin{align*} \Tmeval_{\DALLR{t}} &= \left(\prod_{j=3}^t m_j \right) \Tmeval_{\DALLR{2}} + \sum_{j=3}^t \inparen{\prod_{k=j+1}^t m_k}( \Tm_{\cA_0} + \Tm_{\Gamma} + O(m_j L_0) ) \\
    &= O\left( \left( \prod_{j=3}^t m_j \right) ( \Tmeval_{\DALLR{2}} + t(\Tm_{\cA_0} + \Tm_{\Gamma} + L_0))\right).
    \end{align*}
    Now we see from \Cref{alg:local-alg-two} that
    \begin{equation}\label{eq:eval_time_base}
    \Tmeval_{\DALLR{2}} = O\inparen{\Tm_{\Gamma_2} + (m_2 + n) \Tm_{\cA_0} + nm_2 L_0 + \Tm_{\mathrm{Dec}_{C \otimes C}}} \leq O\left(m_2 \left( \Tm_{\Gamma} + n \Tm_{\cA_0} +  nL_0 + \Tm_{\mathrm{Dec}_{C \otimes C}} \right) \right),
    \end{equation}
    and plugging this into the above proves the statement.

    Next we consider the space.  We first consider the space necessary to just {represent} $A^{(t)}$.  Each $A^{(t)}$ has description length $\mathrm{Len}_{A^{(t)}}$ satisfying
    \[ \mathrm{Len}_{A^{(t)}} = r_t + m_t \mathrm{Len}_{A^{(t-1)}},\]
    with $r_t$ bits for the seed $\sigma^{(t)} \in \{0,1\}^{r_t}$, and $m_t$ descriptions of local algorithms for $C^{\otimes (t-1)}$.  Thus,
    \[ \mathrm{Len}_{A^{(t)}} = \sum_{j=3}^t \inparen{r_j \prod_{i=j+1}^t m_j } + \inparen{\prod_{j=3}^t m_j }\mathrm{Len}_{A^{(2)}}. \]
    For the base case, we have
    \[ \mathrm{Len}_{A^{(2)}} = r_2 + m_2 \log(L_0).\]
    Plugging this in, we see that
    \[ \mathrm{Len}_{A^{(t)}} \leq M_t \inparen{\sum_{j=2}^t r_j + \log L_0} = M_t(\tilde{r}_t + \log L_0).\]
    
    Now we consider the space $\Speval_{\DALLR{t}}$ required to actually {run} $A^{(t)}$, assuming oracle access to the description.
    Looking at \Cref{alg:local-alg}, we see that the space required to run $A^{(t)}$ includes the following components:
    \begin{itemize}
        \item $m_t \log n$ bits to store $H_t$
        \item $\Sp_{\Gamma_t} \leq \Tm_{\Gamma_t} \leq \Tm_{\Gamma}$ to run the sampler $\Gamma_t$
        \item $\Speval_{\DALLR{t-1}}$ to recursively run each level $t-1$ local algorithm
        \item $m_t \log |\Sigma|$ to hold $v_a$ for all $a \in [m_t]$
        \item $\mathrm{Sp}_{\cA_0} + L_0 n \log|\Sigma|$ to run $\cA_0$ and hold the resulting list of $L_0$ codewords in $\Sigma^n$. 
        \item $O(\log (L_0 m_t) )$ additional space to compute the minimum. ($O(\log L_0)$ space to store a pointer to the best candidate so far, plus $O( \log m_t )$ space to store the smallest value so far).
    \end{itemize}
    Putting this together, we have the relationship
    \[ \Speval_{\DALLR{t}} = O\inparen{ m_t \log(n|\Sigma|) + \Sp_{\cA_0} + \Tm_\Gamma + L_0 n \log|\Sigma|} + \Speval_{\DALLR{t-1}},\]
    and unrolling this we see that
    \[ \Speval_{\DALLR{t}} = O\inparen{\left(  \sum_{j=3}^t m_j\right) \log(n|\Sigma|) + t( \Sp_{\cA_0} + \Tm_{\Gamma} + L_0 n \log|\Sigma|) + \Speval_{\DALLR{2}}}.\]
      A similar calculation for \Cref{alg:local-alg-two} shows that
    \begin{equation}\label{eq:eval_space_base}
    \Speval_{\DALLR{2}} = O\inparen{m_2 \log (n|\Sigma|) + \log|\Sigma|( n(L_0 + n) ) + \Tm_{\Gamma_2} + \mathrm{Sp}_{\cA_0} + \Sp_{\mathrm{Dec}_{C \otimes C}}}.
    \end{equation}
    Plugging this in above, and noting that $T_{\Gamma_2} \leq T_\Gamma$, establishes  
    $$\Speval_{\DALLR{t}} = O\inparen{\widetilde{M}_t\log (n|\Sigma|) + t\inparen{\Sp_{\cA_0} + \Tm_{\Gamma} + L_0 n \log|\Sigma|} + \Sp_{\mathrm{Dec}_{C \otimes C}} + n^2 \log|\Sigma|},$$
    which proves the claim.
\end{proof}

\subsection{Analyzing \DALLR{t}, assuming \DALLR{t-1}}\label{sec:inductive_step}

Next we analyze the correctness and time and space requirement of the algorithm \DALLR{t} (\Cref{alg:main}). Before we state our main theorem about the performance of \DALLR{t}, we first set a few more parameters.   Recall that in Section \ref{sec:DALLRC_tensor} we have fixed a linear code $C \subseteq \Sigma^n$ with distance at least $\delta(C)$ that is  $(\rho_0,\ell,L_0)$-globally list-recoverable using an algorithm $\cA_0$, and we have further assumed that $C \otimes C$ has a unique decoder $\mathrm{Dec}_{C \otimes C}$ that can uniquely decode $C \otimes C$ up to radius $\rho^{(!)}_2 \in \left(0, \frac{(\delta(C))^2} 2 \right)$. 
Further recall that our objective is to show that $C^{\otimes t}$ is $(Q_t,\epsilon_t, \rho_t, \ell, L_t)$-$\ADLLR$, where 
  $\{\eps_t\}_t$ and $\{\rho_t\}_t$ were set in \Cref{def:params}.
  It remains to set parameters for the sampler used in \Cref{alg:main}, and also the list size $L_t$ that \DALLR{t} will output and its query complexity $Q_t$.  We do this below.

\begin{definition}[Setting parameters for $\Gamma_t$, as well as the parameters $Q_t$ and $L_t$]\label{def:params2}
    For $t \geq 3$, define
    \begin{equation}\label{eq:sampler_params} \eta_t = \frac{\eps_{t-1}}{10 L_0} \qquad \nu_t = \min\left(\eps_{t-1}, \frac{\delta(C)}{4}\right) \qquad m_t = O\left( \frac{L_0}{\eps_{t-1} \nu_t^2} \right) \qquad r_t = \log\left(\frac{n}{\nu_t}\right). \end{equation}
    For $t=2$, define
    \[ \eta_2 = \frac{\rho_2^{(!)}}{20 L_0} \qquad \nu_2 = \frac{\delta(C)}{8} \qquad m_2 = O\inparen{\frac{1}{\eta_2 \nu_2^2}} \qquad r_2 = \log\left( \frac{n}{\nu_2}\right).\]
Above, the expression in the big-Oh notation around $m_t$ for $t \geq 2$ is chosen so that \Cref{thm:sampler} applies.

Next define $Q_2 = n^2$, and for $t \geq 3$, define
    \begin{equation}\label{eq:Qt} Q_t = nM_t(n+t),\end{equation}
    where $M_t = \prod_{j=2}^t m_j.$ 

   Finally, for $t \geq 2$, define
    \begin{equation}\label{eq:Lt_def} L_t = L_t(L_0, \rho_0, \delta(C)) = L_0^{\kappa^{t^3} \log^t L_0},\end{equation}
    where $\kappa = \kappa(1/\delta(C), 1/\rho_0)$ is the parameter from \Cref{cor:HRW}.

\end{definition}

To analyze \DALLR{t} (\Cref{alg:main}), we first prove, for any $t \geq 3$, that if \DALLR{t-1} is correct, then \DALLR{t} is correct, and we also prove time and space bounds on \DALLR{t} assuming time and space bounds on \DALLR{t-1}.
Then in \Cref{sec:unwind}, we will establish the base case for $t=2$, and unwind the parameters.
Our main result in this section is the following theorem.

\begin{restatable}[Main Technical Theorem]{theorem}{maintech}\label{thm:mainTech}
Fix $t \geq 3$, and suppose that $C \subseteq \Sigma^n$ is a linear code of distance $\delta(C)$ that is $(\rho_0, \ell,L_0)$-globally list recoverable. Suppose also that $C \otimes C$ is uniquely decodable up to radius $\rho_2^{(!)}$. 
Adopt the parameters from \Cref{def:params} and \Cref{def:params2}, 
and suppose that

\begin{equation}\label{eq:eps_t_requirement}
         0 < \eps_{t-1} \leq \min \left\{\frac{ 2 \rho_{0}}{\kappa^{t^3}} \cdot \frac{ \rho_2^{(!)} \gamma_t \delta(C)}{16 d_0}, \frac{(\delta(C))^t} 4 \right\},
    \end{equation}
    and that 
    $\rho_{t-1}\leq \rho_0$.

Assume that $C^{\otimes (t-1)}$ is a $(Q_{t-1}, \eps_{t-1}, \rho_{t-1}, \ell, L_{t-1})$-\ADLLR\ with algorithm $\DALLR{t-1}$.   
Then the following hold:
\begin{itemize}
    \item \textbf{Correctness:} $\Cot$ is a $(Q_t, \eps_t, \rho_t, \ell, L_t)$-\ADLLR\  with algorithm \DALLR{t} (\Cref{alg:main}).
    \item \textbf{Resources:} \DALLR{t} (\Cref{alg:main}) has preprocessing time and space
      \begin{align*} \Tmpre_{\DALLR{t}} &=  2^{r_t}\inparen{\Tm_{\Gamma} + m_t \Tmpre_{\DALLR{t-1}} + (L_{t-1})^{m_t} \TTestt + O(\mathrm{Len}_{A^{(t)}})}\\
     \Sppre_{\DALLR{t}} &= \Sppre_{\DALLR{t-1}} + O\inparen{m_t \log n + \Tm_{\Gamma} + m_t \cdot L_{t-1} \cdot \mathrm{Len}_{A^{(t-1)}} + L_t \cdot \mathrm{Len}_{A^{(t)}} + \Sp_{\Testt}}.
     \end{align*}

        Above, $\TTestt$ and $\Sp_{\Testt}$ are the time and space requirements for \textsc{Test} running on representations of level-$t$ local algorithms, as in \Cref{cor:test_resources}; $\Tm_{\Gamma_i}$ is the time to run the sampler $\Gamma_i$, and let $\Tm_\Gamma = \max_{2 \leq i \leq t} \Gamma_i$.  
 as in \Cref{prop:local_timespace}; And $\mathrm{Len}_{A^{(t)}}$ is the length of $\rep(A^{(t)})$.
\end{itemize}
\end{restatable}

We break the proof of \Cref{thm:mainTech} up into two parts: 
First, in \Cref{sec:correctness}, we prove that \DALLR{t} is correct.  
Second, in \Cref{sec:resourceProof}, we establish the running time and space requirements of \DALLR{t}.  Then we put them together in \Cref{sec:putItTogether} to prove the theorem.

\subsubsection{Correctness of \DALLR{t}, given \DALLR{t-1}}\label{sec:correctness}
To prove that \DALLR{t} is correct,  we first show in \Cref{lem:correct} that for any $\tilde{c}$ with $\delta(\tilde{c}, \cS) \leq \rho_t$, the algorithm \DALLR{t} includes some local algorithm that is $\eps_t$-consistent with $\tilde{c}$ in the output list.
Then in \Cref{lem:list_size}, we bound the list size $|\ctL^{(t)}|$ that it outputs.  Finally, in \Cref{lem:query_complexity}, we show that the query complexity is indeed $Q_t$.

We begin by showing that all close-by codewords are captured by a local algorithm in $\ctL^{(t)}$.
\begin{lemma}[All close-by codewords are represented in $\ctL^{(t)}$]\label{lem:correct}
Assume the hypotheses of \Cref{thm:mainTech}.
    Suppose that $\tilde{c} \in \Cot$ has $\delta(\tilde{c}, \cS) \leq \rho_t$.  Then there exists some $\rep(A^{(t)}) \in \ctL^{(t)}$ so that $\delta(\tilde{c}, x(A^{(t)})) \leq \eps_{t}.$  Here, $\ctL^{(t)}$ is the final list returned by \Cref{alg:main}.
\end{lemma}
\begin{proof}
    The proof is similar to the analogous proof in~\cite{HRW19}, who prove the correctness of a similar (but randomized) local list-recovery algorithm; and to the analogous proof  in~\cite{KRRSS20}, who derandomize \cite{HRW19} (but without small space).   The main difference between our proof and that in prior work is the analysis of the pruning step (that is, the call to \textsc{Test} in \Cref{alg:main}), as the algorithm \textsc{Test} is new to our work. 
    The part of the proof that is similar to prior work is encapsulated in the following claim.
     \begin{restatable}[There is some $A^{(t)}$ that is consistent with $\tilde{c}$]{claim}{goodAt} \label{cl:goodAt}
        Assume the hypotheses of \Cref{thm:mainTech}.
        Let $\cL_1^{(t-1)}, \ldots, \cL_{m_t}^{(t-1)}$ be the lists output by the calls to $\DALLR{t-1}$ in \Cref{alg:main}. 
        Fix $\tilde{c} \in \Cot$ so that $\delta(\tilde{c}, \cS) \leq \rho_t$. Then with probability at least $1/2$ over the choice of a uniform random $\sigma^{(t)} \in \{0,1\}^{r_t}$, there is some $\vec{A} \in \cL_1^{(t-1)} \times \cdots \times \cL_{m_t}^{(t-1)}$ so that if $A^{(t)}$ is the algorithm with representation 
        \[ \rep(A^{(t)}) = (\sigma^{(t)}, \vec{A}),\]
        then \begin{equation}\label{eq:Atclose}
     \Pr_{i \in [n]^t}[ A^{(t)}(i) = \tilde{c}_i ] \geq 1 - \frac{d_0 \eps_{t-1}}{\delta(C)},\end{equation}
     where the probability is only over the choice of uniform random $i \in [n]^t$. 
     \end{restatable}
     We note that \Cref{alg:main} is a deterministic algorithm that iterates over all seeds $\sigma^{(t)} \in \{0,1\}^{r_t}$, so the choice of a uniform random $\sigma^{(t)}$ is not actually random; rather, we treat $\sigma^{(t)}$ as random in \Cref{cl:goodAt} for the analysis only.
     The proof of \Cref{cl:goodAt} is similar to prior work~\cite{HRW19,KRRSS20}; for completeness we include it in \Cref{app:proofOfCorrectnessClaim_t}.
    \vspace{.5cm}

     Now, since \Cref{alg:main} iterates over all possible seeds $\sigma^{(t)} \in \{0,1\}^{r_t}$, \Cref{cl:goodAt} implies that at least once during \Cref{alg:main}, the algorithm forms the representation of some local algorithm $A^{(t)}$ that is consistent with $\tilde{c}$ (using the assumption \cref{eq:eps_t_requirement} that $\epsilon_{t-1} < \frac {\delta(C)} {d_0}$). 
     We next show that at least one such local algorithm $A^{(t)}$ survives the pruning step. 
     
     \begin{claim}[At least one $A^{(t)}$ consistent with $\tilde{c}$ survives the pruning step]\label{cl:Atsurvives}
     Suppose that the hypotheses of \Cref{thm:mainTech} hold, and let $\tilde{c} \in \Cot$ so that $\delta(\tilde{c}, \cS) \leq \rho_t$.  Suppose that, partway through \Cref{alg:main}, the algorithm forms the representation of a local algorithm $A^{(t)}$ so that 
     \Cref{eq:Atclose} holds.
     Suppose also that there is no local algorithm $\tilde{A}^{(t)}$ already present in $\ctL^{(t)}$ so that
     \[ \Pr_{i \in [n]^t}[\tilde{A}^{(t)}(i) = \tilde{c}_i] \geq 1- \eps_t.\]
   
     Then $\textsc{Test}(A^{(t)}, \ctL^{(t)})$ returns \textsc{True}.
         
     \end{claim}
     \begin{proof}
         We need to show that each of the three testers \Tone, \Ttwo,\ and  \Tthree \ called in \textsc{Test} (\Cref{alg:test}) return 1.  Let $x = x(A^{(t)})$ be the string that corresponds to the local algorithm $A^{(t)}$.  By the assumption that \Cref{eq:Atclose} holds, we have
         \[ \delta(x, \tilde{c}) \leq \frac{ d_0 \eps_{t-1}}{\delta(C)}.\]
         For any other local algorithm $\tilde{A}^{(t)}$ as in the statement of the claim, let $\tilde{x} = x(\tilde{A}^{(t)})$, so by assumption $\delta(\tilde{x}, \tilde{c}) > \eps_t$.  
         By \Cref{prop:tests_of_algs}, we have $$T_1(x) = \Tone(\rep(A^{(t)})), \qquad T_2(x) = \Ttwo(\rep(A^{(t)})), \qquad T_3(x,\tilde{x}) = \Tthree( \rep(A^{(t)}), \rep(\tilde{A}^{(t)})).$$ Thus, it suffices to show that $T_1(x) = 1, T_2(x) = 1$, and $T_3(x,\tilde{x}) = 1$ for any $x \in \Sigma^{[n]^t}$ with $\delta(x, \tilde{c}) \leq \frac{d_0 \eps_{t-1}} {\delta(C)}$ and any $\tilde{x} \in \Sigma^{[n]^t}$ with $\delta(\tilde{x}, \tilde{c}) > \eps_t$.  We prove each of these below. 

         \begin{enumerate}
             \item By \Cref{lem:T1}, for any $x$ so that $\delta(x, \Cot) \leq \frac{d_0 \eps_{t-1}}{\delta(C)}$, then $T_1(x) = 1$, so the first test passes.
             \item Given that $T_1(x) = 1$, \Cref{lem:T1} implies that there is a unique closest codeword $\Phi(x) \in \Cot$ to $x$; since 
             \[ \delta(x, \tilde{c}) \leq \frac{ d_0 \eps_{t-1}}{\delta(C)} = \frac{ \eps_t \gamma_t \rho_2^{(!)}}{2} \leq \eps_t < \frac{(\delta(C))^t} 2,\]
             we know that $\tilde{c} = \Phi(x)$.  (Above, in the equality we have plugged in the definition of $\eps_t$ from \Cref{def:params}, and the assumption \Cref{eq:eps_t_requirement}, that $\epsilon_t < \frac{(\delta(C))^t} 4$). 
             Since $\tilde{c} = \Phi(x)$, we have by assumption that $\delta(\Phi(x), \cS) \leq \rho_t$ and $\delta(x, \Phi(x)) \leq \frac{d_0 \eps_{t-1}} {\delta(C)}$.  Then  \Cref{lem:T2} implies that $T_2(x) = 1$. 
             \item Now we know that $T_1(x) = T_2(x) = 1$.  Since we are assuming that $\tilde{A}^{(t)} \in \ctL^{(t)}$, this means that $\textsc{Test}(\rep(\tilde{A}^{(t)})) = \textsc{True}$, and hence $T_1(\tilde{x}) = T_2(\tilde{x}) = 1$ as well.

             Then \Cref{lem:T1} implies that $\delta(\tilde{x}, \Cot) \leq \eps_t$, which implies that the closest codeword $\Phi(\tilde{x}) \in \Cot$ to $\tilde{x}$ must satisfy $\Phi(\tilde{x}) \neq \tilde{c}$, as $\delta(\tilde{x}, \tilde{c}) > \eps_t$ by assumption.  But since $\Phi(x) = \tilde{c}$, this implies in turn that $\Phi(x) \neq \Phi(\tilde{x})$, so \Cref{lem:T3} implies that $T_3(x,\tilde{x}) = 1$.  
         \end{enumerate}
         Thus, all three tests pass, which implies that $\rep(A^{(t)})$ will be added to the final list $\ctL^{(t)}$, as desired.
     \end{proof}
     Now \Cref{cl:goodAt} and \Cref{cl:Atsurvives} together prove \Cref{lem:correct}.  Indeed, \Cref{cl:goodAt} implies that there is at least one local algorithm $A^{(t)}$ considered so that $\delta( \tilde{c}, x(A^{(t)})) \leq \frac{d_0 \eps_{t-1}} {\delta(C)} \leq \eps_t$; and then \Cref{cl:Atsurvives} implies that, if there is not already some $\tilde{A}^{(t)} \in \ctL^{(t)}$ so that $\delta(\tilde{c}, x(\tilde{A}^{(t)})) \leq \eps_t$, then $\rep(A^{(t)})$ will be added to the list.  This completes the proof of \Cref{lem:correct}.

\end{proof}

Now that we have established that the output list $\ctL^{(t)}$ captures all of the close-by codewords, we show that the pruning step \textsc{Test} keeps $\ctL^{(t)}$ small.  Formally, we have the following lemma.
\begin{lemma}[The output list $\ctL^{(t)}$ is small]\label{lem:list_size}
    Suppose that the assumptions of \Cref{thm:mainTech} hold.  Let $\ctL^{(t)}$ be the list returned by \DALLR{t} (\Cref{alg:main}).  Then $|\ctL^{(t)}| \leq L_0^{\kappa^{t^3}\log^t L_0}.$
    \end{lemma}
\begin{proof}
    We will first show that for all local algorithms $A^{(t)} \in \ctL^{(t)}$,
    \begin{equation}\label{eq:within_radius} \delta( \Phi(x(A^{(t)})), \cS ) \leq \frac{\rho_0}{\kappa^{t^3}},\end{equation}
    where $x(A^{(t)}) \in \Sigma^{[n]^t}$ is the string corresponding to $A^{(t)}$, as in \Cref{def:alg_to_string}, and $\Phi(x)$ denotes the codeword in $\Cot$ closest to $x$.
    Indeed, if we can show \Cref{eq:within_radius}, then \Cref{cor:HRW} will imply that $\ctL^{(t)}$ is small.

    To establish \Cref{eq:within_radius}, let $A^{(t)} \in \ctL^{(t)}$, and let $x = x(A^{(t)})$.  Since $A^{(t)} \in \ctL^{(t)}$, that means that it passed the pruning step in \textsc{Test}.  In particular, $\Tone(A^{(t)}) = \Ttwo(A^{(t)}) = \textsc{True},$ which by \Cref{prop:tests_of_algs} implies that $T_1(x) = T_2(x) = 1$.  Then by \Cref{lem:T1}, we have
    \[ \delta(x, \Cot) \leq \eps_t < \frac{(\delta(C))^t}{2},\]
    so there is a unique codeword $\Phi(x) \in \Cot$ so that $\delta(x, \Phi(x)) \leq \eps_t$.  Further, by \Cref{lem:T2}, 
    \[ \delta(x, \cS) \leq \rho_t + \frac{ d_0 \eps_{t-1}}{\delta(C)}.\]
    Putting these together, the triangle inequality implies that
    \[ \delta(\Phi(x), \cS) \leq \delta(x, \Phi(x)) + \delta(x, \cS) \leq \eps_t + \rho_t + \frac{ d_0 \eps_{t-1} }{\delta(C)} \leq \rho_t + 2 \eps_t,\]
    where above we have used the fact that
    \[ \frac{d_0 \eps_{t-1} }{\delta(C)} = \frac{ \gamma_t \rho_2^{(!)} }{2} \eps_t \leq \eps_t,\]
    using \Cref{def:params} and the fact that $\gamma_t, \rho_2^{(!)} < 1$.
    Next, we will bound both terms $\rho_t$ and $2\eps_t$.  We have
    \begin{align*}
        2 \eps_t &= \frac{ 4 d_0 }{\gamma_t \delta(C) \rho_2^{(!)}} \cdot \eps_{t-1} \qquad \text{by \Cref{def:params}} \\
        &\leq \frac{ 4d_0 }{ \gamma_t \delta(C) \rho_2^{(!)}} \cdot \frac{ \rho_2^{(!)} \gamma_t \delta(C) \rho_{0}}{ 8 d_0 \kappa^{t^3}} \qquad \text{by \Cref{eq:eps_t_requirement}} \\
        &= \frac{ \rho_{0} }{2 \kappa^{t^3}}.
    \end{align*}
    Similarly, we can bound
    \begin{align*}
        \rho_t 
        &= \rho_{t-1}\eps_{t-1} \qquad \text{by \Cref{def:params}} \\
        &\leq \eps_{t-1}  \\
        &\leq \frac{ \rho_2^{(!)} \gamma_t \delta(C) \rho_0 }{ 8 d_0 \kappa^{t^3}} \qquad \text{by \Cref{eq:eps_t_requirement}}\\
        &\leq \frac{ \rho_0}{2 \kappa^{t^3}}.
    \end{align*}
    Thus, we have
    \[ \delta(\Phi(x), \cS) \leq \rho_t + 2 \eps_t \leq \frac{ \rho_0 }{\kappa^{t^3}},\]
    proving \Cref{eq:within_radius}. 

    Now we show that \Cref{eq:within_radius}, along with \Cref{cor:HRW}, proves the lemma.
    Let
    \[ \cL^* = \inset{ c \in \Cot \,:\, \delta(c, \cS) \leq \frac{ \rho_0 }{\kappa^{t^3}}}. \]
    By \Cref{cor:HRW}, $|\cL^*| \leq L_0^{\kappa^{t^3}\log^t L_0}$.
    Now consider the map $\Psi: \ctL^{(t)} \to \cL^*$ given by
    \[ \Psi(\rep(A^{(t)})) = \Phi(x(A^{(t)})),\]
    where $\Phi: \Sigma^{[n]^t} \to \Cot$ maps a string to its closest codeword.  We claim that $\Psi$ is injective.  To see this, consider any two $\rep(A^{(t)}) \neq \rep(\tilde{A}^{(t)}) \in \ctL^{(t)}$.  Since both $\rep(A^{(t)}))$ and $\rep(\tilde{A}^{(t)})$ are in $\ctL^{(t)}$, it means that the pruning step \textsc{Test} passed, and in particular $\Tthree(\rep(A^{(t)}), \rep(\tilde{A}^{(t)})) = \textsc{True}$. Then \Cref{prop:tests_of_algs} implies that $T_3(x(A^{(t)}), x(\tilde{A}^{(t)})) = 1$.  By \Cref{lem:T3}, this implies that $\Phi(x(A^{(t)})) \neq \Phi(x(\tilde{A}^{(t)})),$ which shows that $\Psi$ is injective, as desired.  But this means that
    \[ |\ctL^{(t)}| \leq |\cL^*| \leq L_0^{\kappa^{t^3} \log^t L_0},\]
    which proves the lemma.
\end{proof}

Finally, we establish that the query complexity of \DALLR{t} is indeed at most $Q_t$.
\begin{lemma}\label{lem:query_complexity}
    Suppose that the assumptions of \Cref{thm:mainTech} hold.  Let $\ctL^{(t)}$ be the list returned by \DALLR{t} (\Cref{alg:main}).  Then for each local algorithm $A^{(t)} \in \ctL^{(t)}$, the query complexity of $A^{(t)}$ is at most $Q_t$.
\end{lemma}
\begin{proof}
    Recall from \Cref{def:params2} that $Q_2 = n^2$, and $Q_t = nM_t(n+t)$, where $M_t = \prod_{j=2}^t m_j$.     We first observe that for the base case $t=2$, indeed a local algorithm $A^{(2)}$ (given in \Cref{alg:local-alg-two}) trivially has query complexity $n^2$, as it queries everything.  Now we analyze the case for $t \geq 3$.  
    
    Let $\tilde{Q}_t$ be the true query complexity of a local algorithm $A^{(t)}$, so we wish to show that $\tilde{Q}_t \leq Q_t$.  Observe that the $\tilde{Q}_t$ satisfy the recurrence relation
    \[ \tilde{Q}_t \leq m_t \tilde{Q}_{t-1} + n.\]
    This is because, as seen in \Cref{alg:local-alg}, $A^{(t)}$ calls $m_t$ different local algorithms $A^{(t-1)}$; and then queries the entire row $\cS_{\{i'\} \times [n]}$ (on input $i = (i', i_t)$) to run $\cA_0$ on that row.  
    Expanding the recurrence with the base case $\tilde{Q}_2 \leq Q_2 = n^2$, we see that
    \begin{align*} \tilde{Q}_t &\leq n^2 \prod_{j=3}^t m_j + n \sum_{k=3}^t \left( \prod_{j = k+1}^t m_j \right) \\
    &\leq M_t n^2 + ntM_t \\
    &= Q_t.
    \end{align*}
 
\end{proof}

\subsubsection{Time and Space Bounds on \DALLR{t}, given \DALLR{t-1}}\label{sec:resourceProof}
In this section we prove the following lemma.  We recall that we have already analyzed the time and space of the local algorithms (that is, $\Tmeval_{\DALLR{t}}$ and $\Speval_{\DALLR{t}}$) in \Cref{prop:local_timespace}, and so we focus on the pre-processing time and space.
\begin{lemma}[Pre-processing time and space for \DALLR{t}]\label{lem:resources}
    Let $t \geq 3$.  Suppose that $\DALLR{t-1}$ is a DALLR algorithm with list size $L_{t-1}$.  Let $\DALLR{t}$ be as in \Cref{alg:main}, and suppose that the output list size is bounded by $L_t$.  Then
    \begin{align*} \Tmpre_{\DALLR{t}} &=  2^{r_t}\inparen{\Tm_{\Gamma} + m_t \Tmpre_{\DALLR{t-1}} + (L_{t-1})^{m_t} \Tm_{\Testt} + O(\mathrm{Len}_{A^{(t)}})}\\
     \Sppre_{\DALLR{t}} &= \Sppre_{\DALLR{t-1}} + O\inparen{ m_t \log n + \Tm_{\Gamma} +  m_t \cdot L_{t-1} \cdot \mathrm{Len}_{A^{(t-1)}} + L_t \cdot \mathrm{Len}_{A^{(t)}} + \Sp_{\Testt}}. 
     \end{align*}
    Above, for an algorithm $\mathcal{D}$, $\Tm_\cD$ and $\Sp_{\cD}$ denote the amount of time and space required by $\cD$, respectively; and $\Tm_\Gamma = \max_{2 \leq i \leq t} \Tm_{\Gamma_i}$.  For a local algorithm $A^{(t)}$, $\mathrm{Len}_{A^{(t)}}$ represents the length of $\rep(A^{(t)})$; recall that this is bounded in \Cref{prop:local_timespace}.   Recall that $\Tm_{\Testt}$ and $\Sp_{\Testt}$ are given in \Cref{cor:test_resources}.

    Further, the output length is bounded by
    \[ \Spop_{\DALLR{t}} = L_t \cdot \mathrm{Len}_{A^{(t)}}.\]
\end{lemma}
\begin{proof}
To see the time bound, note that \Cref{alg:main} loops over all $2^{r_t}$ possible seeds $\sigma^{(t)}$.  For each, it runs the sampler $\Gamma_t$ to obtain $H_t$, taking time $\Tm_{\Gamma_t} \leq \Tm_\Gamma$.  Then it runs $\DALLR{t-1}$ on $m_t$ different columns.  Finally, it iterates over all $(L_{t-1})^{m_t}$ possible combinations $\vec{A}$ of local algorithms, forms a representation $\rep(A^{(t)})$ by concatenation, and runs $\textsc{Test}$ on it.  Finally, we add $O( \mathrm{Len}_{A^{(t)}})$ overhead for forming and manipulating representations of local algorithms throughout.

For the space, the space requirements are:
\begin{itemize}
    \item $r_t$ bits to store the current seed $\sigma^{(t)}$
    \item $\Sp_{\Gamma_t} \leq \Tm_{\Gamma_t} \leq \Tm_\Gamma$ bits to run the sampler $\Gamma_t$
    \item $m_t \log n$ bits to store $H_t$
    \item $\Sppre_{\DALLR{t-1}}$ to run $\DALLR{t-1}$.  Notice that we may do this sequentially, re-using the space between runs, and storing only the output.
        \item $m_t \cdot L_{t-1} \cdot \mathrm{Len}_{A^{(t-1)}}$ bits to store the lists $\cL_a^{(t-1)}$ for all $a \in [m_t]$ 
            \item $\Sp_{\Testt}$ bits to run $\textsc{Test}$.  Again, we can run all the iterations of \textsc{Test} sequentially, re-using the space.
    \item $L_t \cdot \mathrm{Len}_{A^{(t)}}$ bits to store the output list $\ctL^{(t)}$. 
\end{itemize}
Adding these up results in the final expression (noting that $r_t \leq T_\Gamma$ since $T_\Gamma$ has to read all of its input). 

Finally, we note that the output length follows by definition, as the algorithm outputs $L_t$ representations of level-$t$ local algorithms.
\end{proof}

     \subsubsection{Proof of \Cref{thm:mainTech}}\label{sec:putItTogether}
     Finally, we put together \Cref{lem:correct}, \Cref{lem:list_size}, \Cref{lem:query_complexity}, and \Cref{lem:resources} to prove \Cref{thm:mainTech}.
     \begin{proof}[Proof of \Cref{thm:mainTech}]
     We wish to show that \DALLR{t} (\Cref{alg:main}) is a correct $(Q_t, \eps_t, \rho_t, \ell, L_t)$-DALLR algorithm, with the desired running time and space.

     To show that it is correct, we step through the items in \Cref{def:DALLR} (the definition of \ADLLR).
    First, we observe that \DALLR{t} is indeed deterministic, and outputs a list $\ctL^{(t)}$ of representations of local algorithms.  We can see from \Cref{alg:local-alg} that each of the local algorithms is also deterministic.  \Cref{lem:query_complexity} shows that each local algorithm makes at most $Q_t$ queries to the input $\cS$, and \Cref{lem:list_size} shows that there are at most $|\ctL^{(t)}| \leq L_t$ local algorithms.
     \Cref{lem:correct} shows that for any codeword $\tilde{c} \in \Cot$ so that $\delta(\tilde{c}, \cS) \leq \rho_t$, there is some $A^{(t)}$ whose representation is in $\ctL^{(t)}$ so that 
     \[ \Pr_{i \in [n]^t}[ A^{(t)}(i) = \tilde{c}_i ] \geq 1 - \eps_t.\]
     This hits all the points of \Cref{def:DALLR}, and proves the first part of the theorem.

     Finally, \Cref{lem:resources} implies that \DALLR{t} requires pre-processing space and time as in \Cref{thm:mainTech}, proving the second part of the theorem.
     \end{proof}

\subsection{Unwinding the recursion}\label{sec:unwind}
Now that we have shown in \Cref{thm:mainTech} that if \DALLR{t-1} is correct then \DALLR{t} is correct, we unwind the recursion to prove that \DALLR{t} is correct, and work out the parameters.  Formally, we have the following theorem.

\begin{theorem}\label{thm:mainLocal}
Fix $t \geq 3$,  
 and suppose that $C \subseteq \Sigma^n$ is a linear code of distance $\delta(C)$ that is $(\rho_0, \ell, L_0)$-globally  list recoverable via a deterministic algorithm $\cA_0$. 
Suppose also that $C$ has an algorithm $\textsc{Detect}_C$ that detects errors (that is, $\textsc{Detect}_C(x) = \mathbf{1}[x \in C]$), and suppose that $C \otimes C$ has an algorithm $\textsc{Dec}_{C \otimes C}$ that uniquely decodes $C \otimes C$ up to radius $\rho_2^{(!)}$.  
Adopt the parameters from \cref{def:params} and \cref{def:params2}, 
let $\epsilon>0$
be a parameter so that
\begin{equation}\label{eq:eps_requirement}
\epsilon< \min \left\{\frac{\rho_0} {4\kappa^{t^3}}, \frac {(\delta(C))^t} 4 \right\},
\end{equation}
and let 
\begin{equation}\label{eq:eps_2}
\epsilon_2= \left( \prod_{j=3}^t \gamma_j\right) \left( \frac{\delta(C) \rho_2^{(!)}} {2 d_0}\right)^{t-2} \epsilon .
\end{equation}

Then the following hold.
\begin{itemize}
    \item \textbf{Correctness:} 
    Then there exists an absolute constant $b_0$ so that
    $\Cot$ is a $(Q, \eps, \rho, \ell, L)$-\ADLLR\ with algorithm \DALLR{t} (\Cref{alg:main}), for
    \begin{align*} 
    \rho &\geq \eps^{t-1} \left( \frac{ \rho_2^{(!)}}{2d_0 {t^{b_0}}} \right)^{t^2} \left({\delta(C)}\right)^{t^3} \rho_0 \min\left\{ \frac{ \delta(C)}{8}, \frac{\rho_2^{(!)}}{2}, \frac{1}{\kappa^8} \right\}, \\
    L &= L_0^{\kappa^{t^3} \log^t L_0}, \text{ and}\\
    Q &\leq n(n+t)\inparen{ \frac{b_0 L_0}{\eps_2^3}}^t
    \end{align*}

    \item \textbf{Resources:} 
        There is a constant
\[ c_0 = \poly\inparen{ \frac{1}{\delta(C)}, L_0, \frac{1}{\rho_2^{(!)}}, \frac{1}{\rho_0} }\]
so that the following holds.
    \DALLR{t} (\Cref{alg:main}) has pre-processing time and space
\begin{align*}
\Tmpre_{\DALLR{t}}
 &\leq n^{t+1} \frac{1}{\eps_2^{4t + O(1)}} \inparen{  n\Tm_{\cA_0} + \Tmdectwo}\inparen{ \frac{\Tmdetect}{n} + \frac{\Tmdectwo}{n^2} }  \cdot \exp\inparen{ c_0^{t^3} \eps_2^{-3} }.\\
\Sppre_{\DALLR{t}} &=
O\inparen{ t \cdot \Sp_{\cA_0} + \Spdetect + \Spdectwo + \frac{c_0 \cdot t \cdot n^2\polylog|\Sigma|}{\poly(\eps_2)} +  \eps_2^{-3t} \cdot \exp(c_0^{t^3} ) \cdot \log n },
\end{align*}
evaluation time and space is bounded by
\begin{align*}
    \Tmeval_{\DALLR{t}} &\leq O\inparen{\frac{c_0^t}{\eps_2^{3t}} \cdot \poly\inparen{\frac{\log n}{\eps_2}} + n \cdot \Tm_{\cA_0} + \Tmdectwo}\\
    \Speval_{\DALLR{t}} &\leq 
c_0\cdot t\cdot \inparen{\poly\inparen{\frac{\log(n|\Sigma|)}{\eps_2}} + n^2\log|\Sigma| }+ t \cdot \Sp_{\cA_0} + \Spdectwo,
\end{align*}
and output length is bounded by 
\[ \Spop_{\DALLR{t}} \leq \frac{\exp(c_0^{t^3}) \cdot \log n}{\eps_2^{3t}+O(1)}.\]

Above, for an algorithm $\mathcal{D}$, we let $\Tm_{\mathcal{D}}$ and $\Sp_{\mathcal{D}}$ denote the time and space requirements of $\mathcal{D}$, respectively.
\end{itemize}

\end{theorem}

  Again we prove \Cref{thm:mainLocal} in two steps.  We first establish the correctness in \Cref{sec:unwind_correctness}, and then we establish the time and space bounds in \Cref{sec:unwind_resources}.
    We put them together in \Cref{sec:pfmainLocal} to prove \Cref{thm:mainLocal}.

\subsubsection{Correctness of \DALLR{t}}\label{sec:unwind_correctness}
    To establish correctness, we first need to show the base case for $t=2$, namely that \DALLR{2} (\Cref{alg:main_base}) is correct.  Then \Cref{thm:mainTech} provides the inductive step that will establish that \DALLR{t} is correct.

  \begin{restatable}[Base case: \DALLR{2} (\Cref{alg:main_base}) is correct]{lemma}{basecorrect}\label{lem:base_correct}  

        Assume the hypothesis of Theorem \ref{thm:mainLocal}.
  Then $C \otimes C$ is a $(Q_2, \rho_2, \ell, L_2)$-\DLLR\ with algorithm $\DALLR{2}$ (\cref{alg:main_base}).
    \end{restatable}

   The proof of the above lemma is similar to prior work~\cite{GGR11, HRW19,KRRSS20}; for completeness we include it in Appendix \ref{app:proofOfCorrectnessClaim_base}.

   The above lemma shows that $\DALLR{2}$ is correct, and \Cref{thm:mainTech} states that if $\DALLR{t-1}$ is correct, then $\DALLR{t}$ is correct.  In the next lemma, we put these together to conclude that $\DALLR{t}$ is correct by induction, and work out the parameters.
    \begin{lemma}\label{lem:unwind_correct}
    Assume the hypotheses of \Cref{thm:mainLocal}.  
    Then $\Cot$ is a $(Q, \eps, \rho, \ell, L)$-\ADLLR\ with algorithm $\DALLR{t}$ (\Cref{alg:main}), where $Q, \eps, \rho, L$ are as in \Cref{thm:mainLocal}.
    \end{lemma}

    \begin{proof}
        Let 
        \begin{equation}\label{eq:W0} W_0 := \frac{ 2d_0 }{ \delta(C) \rho_2^{(!)} },\end{equation}
        so that the definition of $\eps_i$ in \Cref{def:params} says that
        \[ \eps_i = \frac{ W_0 \eps_{i-1}}{\gamma_i}\]
        for all $i = 3, \ldots, t$.
        Thus, we have that
        \begin{equation}\label{eq:eps_unrolled} \eps_i = \frac{W_0^{i-2}}{\prod_{j=3}^i \gamma_j } \eps_2
        \end{equation}
        for all $i = 3, \ldots, t$.
    With this notation, our choice of $\eps_2$ in \cref{eq:eps_2}  
    is
    \begin{equation}\label{eq:eps_2_correct}
        \eps_2 = \frac{\prod_{j=3}^t \gamma_j}{W_0^{t-2}} \cdot \eps.
    \end{equation}
    Notice that this choice and \Cref{eq:eps_unrolled} implies that $\eps = \eps_t$.
    
        We have shown as our base case in \Cref{lem:base_correct} that \DALLR{2} is a correct 
        $(Q_2, \rho_2, \ell, L_2)$-\DLLR, and in particular this means that it is a  correct
        $(Q_2, \eps_2, \rho_2, \ell, L_2)$-\ADLLR\ for $\eps_2$ as in \Cref{eq:eps_2}.  Now choose any $i$ so that $2 < i \leq t$, and suppose inductively that $\DALLR{i-1}$ is a correct $(Q_{i-1}, \eps_{i-1}, \rho_{i-1}, \ell, L_{i-1})$-\ADLLR.

        We now wish to apply the first bullet point of~\Cref{thm:mainTech} (``Correctness''), to conclude that \DALLR{i} is a correct $(Q_i, \eps_i, \rho_{i}, \ell, L_{i})$-\ADLLR.  To do so, we need to check that the hypotheses of \Cref{thm:mainTech} hold.  In particular, we need to check that \eqref{eq:eps_t_requirement} holds, or equivalently that
        \begin{equation}\label{eq:newneedeps} 0 < \eps_{i-1} \leq 
        \min \left\{\frac{\rho_0}{4\kappa^{i^3}} \cdot \frac{\gamma_i}{ W_0}, \frac{(\delta(C))^i} 4 \right\}
        \end{equation}
        We also need to check the other hypothesis of \Cref{thm:mainTech} that $\rho_{i-1} \leq \rho_0$.
        
        We begin by establishing \Cref{eq:newneedeps}.  First, we write, using \Cref{eq:eps_unrolled} and \Cref{eq:eps_2_correct}, that
        \[ \eps_{i-1} = \frac{W_0^{i-3}}{\prod_{j=3}^{i-1} \gamma_j} \cdot \eps_2 = \frac{{W_0}^{i-3}}{\prod_{j=3}^{i-1} \gamma_j} \cdot \frac{ \prod_{j=3}^t \gamma_j }{ W_0^{t-2} } \cdot \eps = \frac{ \prod_{j=i}^t \gamma_j }{W_0^{t-i+1}} \cdot \eps.\]
        Thus, the first requirement in \Cref{eq:newneedeps} is equivalent to 
        \[ \frac{\prod_{j=i}^t \gamma_j }{W_0^{t-i+1}} \cdot \eps \leq \frac{\rho_0}{4\kappa^{i^3}} \cdot \frac{\gamma_i}{W_0},\]
        or equivalently
        \begin{equation}
            \label{eq:newnewneedeps}
            \eps \leq \frac{\rho_0}{4\kappa^{i^3}} \cdot \frac{ W_0^{t-i}}{ \prod_{j=i+1}^t \gamma_j } =: Z_i.
        \end{equation}
        (Above, in the definition of $Z_i$, notice that if $t=i$, then the product from $j=t+1$ to $t$ should be interpreted as $1$).
        Recall from the hypothesis \Cref{eq:eps_requirement} that 
        \[\eps \leq \frac{ \rho_0 }{4 \kappa^{t^3}} = Z_t.\]
        Next, observe that the $Z_i$ are decreasing: That is, we have $Z_i \geq Z_{i+1}$ for all $i < t$.  Indeed, we have
        \[ \frac{Z_i}{Z_{i+1}} = \frac{W_0^{t-i} \kappa^{(i+1)^3} \prod_{j={i+2}}^t \gamma_j }{W_0^{t-i-1} \kappa^{i^3} \prod_{j=i+1}^t \gamma_j} = \frac{W_0}{\gamma_{i+1}} \cdot  \kappa^{(i+1)^3 - i^3} > 1,\]
        using the fact that $\kappa , W_0 > 1$ and that $\gamma_i < 1$ for all $i$.  In particular, $Z_t \leq Z_i$ for all $i \leq t$, so our assumption that $\eps \leq Z_t$ implies \Cref{eq:newnewneedeps}.

       Next, we observe that $\eps_{i-1} \leq \eps \leq \frac{ (\delta(C))^t} {4}$. Indeed, the first inequality follows from the facts that $\eps_t = \eps$ and that the $\eps_i$ are increasing by definition (\Cref{def:params}), while the second inequality follows by \cref{eq:eps_requirement}.

       Finally, we establish that $\rho_{i-1} \leq \rho_0$.  It follows immediately from the definition of $\rho_2$ in Definition \ref{def:params} that $\rho_2 \leq \rho_0$; and it follows from the definition $\rho_i = \rho_{i-1}\eps_{i-1}$ in \Cref{def:params} that the $\rho_i$ are decreasing.  This implies that $\rho_{i-1} \leq \rho_0$ for all $i = 2, \ldots, t$.

       Thus, \Cref{thm:mainTech} applies and shows that $C^{\otimes i}$ is a $(Q_i, \eps_i, \rho_i, \ell, L_i)$-\ADLLR\ with algorithm $\DALLR{i}$, as desired.  By induction, we conclude that $\Cot$ is a $(Q_t,  \eps_t, \rho_t, \ell, L_t)$-\ADLLR\ with algorithm $\DALLR{t}$.  It remains to work out the parameters $Q_t, \eps_t, \rho_t, L_t$.  

       We have already established that $\eps_t = \eps$. We have $Q = Q_t = nM_t (n+t)$ and $L = L_t = L_0^{\kappa^{t^3} \log^t L_0}$ from \Cref{def:params2}.  
       The expression for $Q$ in \Cref{thm:mainLocal} follows from observing that
       \[ M_t = \prod_{j=2}^t m_j \leq \left(\frac{b_0 L_0}{\eps_2^3}\right)^t,\]
       for some constant $b_0$.  Indeed, 
      from \Cref{def:params2}, we have for any $j$ that
       \[ m_j = O\inparen{\frac{L_0}{\eps_{j-1} \nu_j^2 }} \leq O\inparen{\frac{L_0}{\eps_{j-1}\eps_2^2}} \leq O\inparen{\frac{L_0}{\eps_2^3}},\]
       using that $\nu_j = \min(\eps_{j-1}, \frac{\delta(C)} 4) \geq \eps_2$ for all $j$.
       
       For $\rho_t$, we have from \Cref{def:params} that
       \[ \rho_t = \rho_{t-1} \eps_{t-1},\]
       which implies that
       \begin{equation}\label{eq:rho_unroll}
       \rho_t = \rho_2 \cdot \prod_{i=2}^{t-1} \eps_i.
       \end{equation}
       As above, from \Cref{eq:eps_unrolled} and \Cref{eq:eps_2_correct}, we have
       \[ \eps_i = \frac{\prod_{j=i+1}^t \gamma_j}{W_0^{t-i}} \eps.\]
       Plugging this into \Cref{eq:rho_unroll}, we have
       \begin{align*}
           \rho_t &= \rho_2 \cdot \prod_{i=2}^{t-1} \left( \frac{\prod_{j=i+1}^t \gamma_j}{W_0^{t-i}} \eps \right) \\
           &= \rho_2 \cdot \eps^{t-2} \cdot W_0^{-(t-2)(t-1)/2} \cdot \prod_{j=3}^t \gamma_j^{j-2}.
       \end{align*}
       Recalling from \Cref{def:params} that 
       \[\gamma_j = \frac{\delta(C)^{2j}}{18^{\log_{1.5}(j)}} =  (\delta(C))^{2j} \cdot j^{-\tilde{b}_0}\] for some absolute constant $\tilde{b}_0 > 1$, and the definition of $W_0$ from \Cref{eq:W0}, we can write 
       \begin{align*} \rho_t &= \rho_2 \cdot \eps^{t-2} \cdot \left( \frac{ \rho_2^{(!)}}{2 d_0} \right)^{(t-1)(t-2)/2} \cdot (\delta(C))^{(t-1)(t-2)\left(\frac{1}{2} + \frac{ 2t + 3}{3} \right)} \cdot \left( \prod_{j=3}^t j^{j-2} \right)^{-\tilde{b}_0} \\
       &\geq \rho_2 \cdot \eps^{t-1} \cdot \left( \frac{\rho_2^{(!)}}{2 d_0} \right)^{t^2} \cdot (\delta(C))^{t^3} \cdot t^{-b_0 t^2}
       \end{align*}
       for some other absolute constant $b_0$.  
    
     Indeed, to see the equality above, we separate out the exponent of $\delta(C)$; the contribution from $W_0^{-(t-2)(t-1)/2}$ is $\delta(C)^{(t-2)(t-1)/2}$, and the contribution from the term $\prod_{j=3}^t \gamma_j^{j-2}$ is $\prod_{j=3}^t \delta(C)^{2j} = \delta(C)^{\sum_{j=3}^t 2j(j-2)} = \delta(C)^{(t-1)(t-2)(2t+3)/3}.$  Finally, the $(\rho_2^{(!)}/2d_0)^{(t-1)(t-2)/2}$ term is the remaining contribution from $W_0^{-(t-2)(t-1)/2}$, while the $\left(\prod_{j=3}^t j^{j-2}\right)^{-\tilde{b}_0}$ term is the remaining contribution from $\prod_{j=3}^t \gamma_j^{j-2}$.
       
       Plugging in the expression for $\rho_2$ from \Cref{def:params} establishes the bound on $\rho = \rho_t$. 
       This completes the proof of \Cref{lem:unwind_correct}.
    \end{proof}

    \subsubsection{Time and Space Bounds on \DALLR{t}}\label{sec:unwind_resources}

To establish the time and space bounds, we first establish the preprocessing time and space for the base case $t=2$.  
\begin{lemma}[Base case for $t=2$]\label{cl:A2simp}
Assume the hypotheses of \Cref{thm:mainLocal}. 
Then there is a constant $c_0 = \poly(1/\delta(C),  L_0,  1/\rho_2^{(!)}, 1/\rho_0)$ so that

     \begin{align*} \Tmpre_{\DALLR{2}} &= O\inparen{ n^3 \cdot L_0^{c_0} \inparen{ n\Tm_{\cA_0} + \Tm_{\mathrm{Dec}_{C \otimes C}}}}\\
     \Sppre_{\DALLR{2}} &= O\inparen{ L_0^{c_0} \log n  + c_0 n^2 \log|\Sigma|  + \Sp_{\cA_0} + \Spdectwo}.
     \end{align*}
\end{lemma}
\begin{proof}
First, we claim that 
the pre-processing time and space for $\DALLR{2}$ is given by
    \begin{align} \Tmpre_{\DALLR{2}} &= O \inparen{ 2^{r_2} L_0^{m_2}  L_2 n^2 \Tmeval_{\DALLR{2}} } \label{eq:tm2}\\
     \Sppre_{\DALLR{2}} &= O \inparen{L_2(r_2 + m_2 \log L_0) + \log n + \Speval_{\cP^{(2)}}}. \label{eq:sp2}
     \end{align}
     
     Indeed, recall that $\DALLR{2}$ is given in \Cref{alg:main_base}. 
To see the time bound \Cref{eq:tm2}, note that  first \Cref{alg:main_base} loops over $2^{r_2} L_0^{m_2}$ pairs $(\sigma^{(2)}, \tau)$. Within each loop, it must compute the variable $\mathrm{dist}$ ($O(n^2)$ evaluations of local algorithms $A^{(2)}$, so time $O(n^2) \Tmeval_{\DALLR{2}}$) and the variable $\mathrm{sim}$ (time $O(L_2 n^2)\cdot \Tmeval_{\DALLR{2}}$).  Then it must compare these variables to their thresholds (time $O(1)$) and possibly add something to the output list (time $O(1)$), establishing \Cref{eq:tm2}.

To see the space bound \Cref{eq:sp2}, notice that we must store the following elements:
\begin{itemize}
    \item $r_2  +m_2 \log L_0$ bits for the current representation $(\sigma^{(2)}, \tau)$, and then up to another $L_2( r_2 + m_2 \log L_0)$ bits for the list $\ctL^{(2)}$, for a total of $O( L_2(r_2 + m_2 \log L_0))$ bits.
    \item $O( \log n)$ bits to compute and store the variables $\mathrm{dist}$ and $\mathrm{sim}$.
    \item $\Speval_{\cP^{(2)}}$ to run the local algorithms $A^{(2)}$ sequentially.
\end{itemize}
This establishes \Cref{eq:sp2}.

Now we simplify \Cref{eq:sp2} and \Cref{eq:tm2} in terms of the constant $c_0$.
From \Cref{def:params2}, we see that
\[ 2^{r_2} = \frac{ 8n}{\delta(C)} \qquad m_2 = O\inparen{ \frac{L_0}{(\delta(C))^2 \rho_2^{(!)} }} \qquad \nu_2 = \frac{\delta(C)}{8} \qquad \eta_2 = \frac{\rho_2^{(!)}}{20 L_0} \qquad L_2 = L_0^{\kappa^8 \log^2 L_0}. \]
In particular, there is a constant $c_0$ as in the statement of the claim so that 
\[ 2^{r_2} \leq c_0 n \qquad m_2 \leq c_0 \qquad \nu_2  \geq 1/c_0 \qquad \eta_2 \geq 1/c_0 \qquad L_2 \leq L_0^{c_0}.\]
From \Cref{thm:sampler}, we have 
  \[ \Tm_{\Gamma_2} = \poly(\log n, 1/\eta_2, 1/\nu_2) \leq c_0 \cdot \poly\log n,\]
  adjusting $c_0$ appropriately.
  As in the proof of \Cref{prop:local_timespace} (\cref{eq:eval_time_base}), we have \begin{align*}
      \Tmeval_{\DALLR{2}} &= O(m_2(\Tm_{\Gamma_2} + n\Tm_{\cA_0} + L_0 + \Tm_{\mathrm{Dec}_{C \otimes C}})) \\
            &\leq O( c_0 (\polylog(n) + n \Tm_{\cA_0} + \Tm_{\mathrm{Dec}_{C \otimes C}}) )\\
      &\leq O( c_0 (n \Tm_{\cA_0} + \Tm_{\mathrm{Dec}_{C \otimes C}}) ),
  \end{align*}
  again adjusting $c_0$ as necessary.  Now from \Cref{eq:tm2}, we have
  \begin{align*} \Tmpre_{\DALLR{2}} &= O \inparen{ 2^{r_2} L_0^{m_2}  L_2 n^2 \Tmeval_{\DALLR{2}} }\\
  &= O\inparen{ n^3 \cdot c_0 L_0^{2c_0} \cdot \Tmeval_{\DALLR{2}}}\\
  &= O\inparen{ n^3 \cdot c_0^2 L_0^{2c_0} \inparen{ n\Tm_{\cA_0} + \Tm_{\mathrm{Dec}_{C \otimes C}}}}\\
  &= O\inparen{ n^3 \cdot L_0^{c_0} \inparen{ n\Tm_{\cA_0} + \Tm_{\mathrm{Dec}_{C \otimes C}}}},
  \end{align*}
  where in the last line we have absorbed the $c_0$'s into the exponent of $L_0^{c_0}$, adjusting $c_0$ appropriately.

  Next we do the same with the space bound.  We have from \Cref{eq:sp2} that
  \begin{align*}
      \Sppre_{\DALLR{2}} &= O \inparen{L_2(r_2 + m_2 \log L_0) + \log n + \Speval_{\cP^{(2)}}} \\
      &= O\inparen{ L_0^{c_0} ( \log(c_0 n) + c_0 \log c_0 ) + \log n+ \Speval_{\cP^{(2)}}} \\
      &= O\inparen{ L_0^{c_0} \log n + \Speval_{\cP^{(2)}}},
  \end{align*}
  again adjusting $c_0$ as necessary.  We further simplify this by recalling from the proof of \Cref{prop:local_timespace} (\cref{eq:eval_space_base}) that 
  \begin{align*}
  \Speval_{\cP^{(2)}} &= O\inparen{ m_2 \log (n |\Sigma|) + \log|\Sigma|(n (L_0 + n)) + \Tm_{\Gamma_2} + \Sp_{\cA_0} + \Spdectwo} \\
  &= O\inparen{c_0 \log (n|\Sigma|) + c_0 n^2 \log|\Sigma| + c_0 \polylog(n) + \Sp_{\cA_0} + \Spdectwo} \\
  &= O\inparen{ c_0 n^2 \log|\Sigma| + \Sp_{\cA_0} + \Spdectwo},
  \end{align*}

  updating $c_0$ as needed.  Thus, together we have
  \[ \Sppre_{\cP^{(2)}} = O\inparen{ L_0^{c_0} \log n + c_0 n^2 \log|\Sigma| + \Sp_{\cA_0} + \Spdectwo},\]
  as desired.
\end{proof}

    Now, we move on to the case for $t \geq 3$.  

\begin{lemma}\label{lem:unwind_resources}
Assume the hypotheses of \Cref{thm:mainLocal}. Then there is a constant
\[ c_0 = \poly\inparen{ \frac{1}{\delta(C)}, L_0, \frac{1}{\rho_2^{(!)}}, \frac{1}{\rho_0} }\]
so that the following holds.  The pre-processing time and space of $\cApre^{(t)}$  (\cref{alg:main}) is bounded by
\begin{align*}
\Tmpre_{\DALLR{t}} 
 &\leq n^{t+1} \frac{1}{\eps_2^{4t + O(1)}} \inparen{  n\Tm_{\cA_0} + \Tmdectwo}\inparen{ \frac{\Tmdetect}{n} + \frac{\Tmdectwo}{n^2}}  \cdot \exp\inparen{ c_0^{t^3} \eps_2^{-3} }.\\
\Sppre_{\DALLR{t}} &=
O\inparen{ t \Sp_{\cA_0} + \Spdetect + \Spdectwo + \frac{c_0 \cdot t \cdot n^2\polylog|\Sigma|}{\poly(\eps_2)} +  \eps_2^{-3t} \cdot \exp(c_0^{t^3} ) \cdot \log n }
\end{align*}
The evaluation time and space is bounded by
\begin{align*}
    \Tmeval_{\DALLR{t}} &\leq O\inparen{\frac{c_0^t}{\eps_2^{3t}} \cdot \poly\inparen{\frac{\log n}{\eps_2}} + n \cdot \Tm_{\cA_0} + \Tmdectwo}\\
    \Speval_{\DALLR{t}} &\leq 
c_0\cdot t\cdot \inparen{\poly\inparen{\frac{\log(n|\Sigma|)}{\eps_2}} + n^2\log|\Sigma| }+ O(t \cdot \Sp_{\cA_0} + \Spdectwo).
\end{align*}
Finally, the output length is bounded by 

\[ \Spop_{\DALLR{t}} \leq \frac{\exp(c_0^{t^3}) \cdot \log n}{\eps_2^{3t+O(1)}}.\]
\end{lemma}

\begin{proof}
    First, we simplify a few parameters in terms of  $c_0$.
    First, from \Cref{def:params2}, we see that there is some constant $c_0$ as in the statement of the claim 
    so that
    \[ \eta_t \geq \frac{\epsilon_{t-1}}{c_0} \qquad \nu_t \geq \frac{\eps_{t-1}}{c_0} \qquad m_t \leq \frac{c_0}{\eps_{t-1}^3} \qquad r_t \leq \log\left( \frac{c_0 n}{\eps_{t-1}} \right).\]
    Thus, using the fact that the $\eps_i$ are increasing with $i$, we have 
    \begin{align} \tilde{r}_t &= \sum_{j=2}^t r_j \leq t \log \left( \frac{ c_0 n}{\eps_2} \right) \label{eq:trt} 
    \end{align}
    Similarly, we have
    \begin{align} 
    M_t &= \prod_{j=2}^t m_j \leq \frac{ c_0^t }{\eps_2^{3t}} \label{eq:Mt}\\
    \tilde{M}_t &= \sum_{j=2}^t m_j \leq \frac{ tc_0}{\eps_2^3} \label{eq:Mtildet}.
    \end{align}
    From \Cref{thm:sampler}, we have
    \[ \Tm_{\Gamma} = \max_{2 \leq i \leq t } \poly( \log n , 1/\eta_i, 1/\nu_i ) = \poly\inparen{ \frac{ c_0 \log n }{\eps_2}}.\]

With these building blocks in place, we move on to working out the time and space bounds for $\DALLR{t}$.  We begin with the time bounds.

\paragraph{Time bounds.}
    We know from \Cref{prop:local_timespace} that
    \begin{align}
        \Tmeval_{\DALLR{t}} &= O\inparen{ M_t( t(\Tm_{\Gamma} + n \Tm_{\cA_0} + L_0) + \Tm_{\mathrm{Dec}_{C \otimes C}}) } \notag \\
        &= O\inparen{ \frac{c_0^t}{\eps_2^{3t}} \inparen{  t\cdot \poly( \eps_2^{-1} c_0 \log n ) +t  n \Tm_{\cA_0} + tL_0 + \Tm_{\mathrm{Dec}_{C \otimes C}}}} \notag \\
         &= O\inparen{ \frac{c_0^t}{\eps_2^{3t}} \inparen{  \poly( \eps_2^{-1} \log n ) + n \Tm_{\cA_0} + \Tm_{\mathrm{Dec}_{C \otimes C}}}} \label{eq:Tmeval}
    \end{align}
    where 

    we have absorbed the $L_0$ term into $c_0$, the factor of $t$ into $c_0^t$, and the $\poly(c_0)$ term into $c_0^t$, by adjusting $c_0$ as necessary. 
    Now, \Cref{eq:Tmeval} gives the claimed bound on $\Tmeval_{\DALLR{t}}$ in the lemma statement.  
    
    We move on to $\Tmpre_{\DALLR{t}}$.
    We first simplify $\TTestt.$
    We know from \Cref{cor:test_resources} that 
    \[ \TTestt = O\left(\Tmeval_{\DALLR{t}} (t^2 (n^{t-1} \Tm_{\textsc{Detect}_C} + n^{t-2} \Tm_{\mathrm{Dec}_{C \otimes C}}) + L_t n^t )\right).\]
    
    Together with \Cref{eq:Tmeval}, we see that
    \begin{equation}\label{eq:ttestt} \TTestt = O\inparen{\frac{c_0^t}{\eps_2^{3t}} \inparen{ \poly\left(\frac{ \log n }{\eps_2} \right) + n \Tm_{\cA_0} + \Tm_{\mathrm{Dec}_{C \otimes C}}} \cdot \inparen{n^{t-1} \Tm_{\textsc{Detect}_C} + n^{t-2}\Tm_{\mathrm{Dec}_{C \otimes C}} + L_t n^t }},\end{equation}
    where again we have absorbed the polynomial factor of $t$ into $c_0^t$, adjusting $c_0$ as necessary.

     We now unroll the recurrence relations.  
     
     Recall from \Cref{thm:mainTech} that
     \[ \Tmpre_{\DALLR{t}} =  2^{r_t}\inparen{\Tm_{\Gamma} + m_t \Tmpre_{\DALLR{t-1}} + L_{t-1}^{m_t} \Tm_{\Testt} + O(\mathrm{Len}_{A^{(t)}})}. \]
     Recalling
   from \Cref{def:params2} that
    \[ L_t = L_0^{\kappa^{t^3} \log^t L_0},\]
     we simplify this as
    \begin{align*}
        \Tmpre_{\DALLR{t}} &=  2^{r_t}\inparen{\Tm_{\Gamma} + m_t \Tmpre_{\DALLR{t-1}} + L_{t-1}^{m_t} \Tm_{\Testt}   + O(\mathrm{Len}_{A^{(t)}})} \\
        &\leq \frac{ c_0 n }{\eps_2} \inparen{ \poly \inparen{ \frac{c_0 \log n}{\eps_2}} + \frac{c_0}{\eps^3_2} \Tmpre_{\DALLR{t-1}} + L_0^{\kappa^{t^3} \log^t L_0 c_0/\eps_2^3}\cdot \TTestt + O(\mathrm{Len}_{A^{(t)}})} \\
        &\leq \frac{c_0 n}{\eps_2^4} \Tmpre_{\DALLR{t-1}} + c_0 n \inparen{ \frac{\polylog(n)}{\poly(\eps_2)} + \exp\inparen{\frac{c_0 \kappa^{t^3}}{\eps_2^3}\log^{t+1}L_0} \TTestt + O\left( \frac{\mathrm{Len}_{A^{(t)}}} {\epsilon_2}\right) }\\
        &\leq \frac{c_0 n}{\eps_2^4} \Tmpre_{\DALLR{t-1}} + c_0 n \inparen{ \frac{\polylog(n)}{\poly(\eps_2)} + \exp\inparen{\frac{c_0 \kappa^{t^3}}{\eps_2^3}\log^{t+1}L_0} \TTestt  }
    \end{align*}
    In the last line, we observe that the $O(\mathrm{Len}_{A^{(t)}})$ term is dominated by $T_{\Testt}$ term.  Indeed, as we saw above, $\TTestt \geq \Tmeval_{\DALLR{t}}$, the cost of evaluating $A^{(t)}$.  Since one must read the representation of $A^{(t)}$ to evaluate it, this takes time at least $\mathrm{Len}_{A^{(t)}}$.  Thus, we can drop this term at the cost of adjusting $c_0$ slightly in the coefficient in front of $\TTestt$ (assuming that $n^t$ is sufficiently large).

    Now we unroll the recurrence relation, to see that
    \begin{align*}
        \Tmpre_{\DALLR{t}} &= \inparen{ \frac{c_0 n}{\eps_2^4} }^{t-2} \Tmpre_{\DALLR{2}} + c_0 n \sum_{j=3}^t\inbrak{\inparen{\frac{c_0 n}{\eps_2^4}}^{t-j}\inparen{ \frac{  \polylog(n) }{\poly(\eps_2)} + \exp\inparen{\frac{ c_0 \kappa^{j^3}}{\eps_2^3} \log^{j+1} L_0} \Tm_{\textsc{Test}^{(j)}}}}
    \end{align*}
    We simplify each of the terms one at a time.  We begin with
    \begin{align}
        \inparen{ \frac{c_0 n}{\eps_2^4}}^{t-2} \Tmpre_{\DALLR{2}} 
        &= O\inparen{\inparen{ \frac{c_0 n}{\eps_2^4}}^{t-2} \cdot n^3 L_0^{c_0}(n\Tm_{\cA_0} + \Tmdectwo) } \notag \\
        &= O\inparen{ n^{t+1} L_0^{c_0} \inparen{\frac{c_0}{\eps_2^4}}^{t-2} \inparen{n\Tm_{\cA_0} + \Tmdectwo}} \label{eq:firstterm}
    \end{align}
    using \Cref{cl:A2simp} and adjusting $c_0$ as necessary to accommodate both the current $c_0$ and the $c_0$ from \Cref{cl:A2simp}.
Next we write
\begin{align}
    c_0n \sum_{j=3}^t \inparen{ \frac{c_0 n}{\eps_2^4} }^{t-j} \inparen{\frac{\polylog n}{\poly(\eps_2)}} 
    &= c_0 n \inparen{\frac{\polylog(n)}{\poly(\eps_2)}}\frac{ \inparen{ \frac{c_0 n }{\eps_2^4}}^{t-2}-1}{\frac{c_0 n}{\eps_2^4} - 1} \notag \\
    &\leq c_0 n \inparen{\frac{\polylog(n)}{\poly(\eps_2)}} \inparen{ \frac{c_0 n }{\eps_2^4}}^{t-2}  \notag \\
    &= n^{t-1} \polylog(n) \cdot \frac{c_0^{t-1}}{ \eps_2^{4t+O(1)}} \label{eq:secondterm}
\end{align}
The last term is 
\begin{equation}\label{eq:lastterm}
    c_0 n \sum_{j=3}^t  \inparen{ \frac{c_0 n}{\eps_2^4} }^{t-j} \exp\inparen{\frac{ c_0 \kappa^{j^3} \log^{j+1} L_0}{\eps_2^3}} \Tm_{\textsc{Test}^{(j)}} 
    \end{equation}
Recalling from \Cref{eq:ttestt} that
    \begin{align*}
    \Tm_{\textsc{Test}^{(j)}} &= O\inparen{\frac{c_0^j}{\eps_2^{3j}} \inparen{ \poly\left(\frac{ \log n }{\eps_2} \right) + n \Tm_{\cA_0} + \Tm_{\mathrm{Dec}_{C \otimes C}}} \cdot \inparen{n^{j-1} \Tm_{\textsc{Detect}_C} + n^{j-2}\Tm_{\mathrm{Dec}_{C \otimes C}} + L_j n^j }},
\end{align*}
we upper bound that $j$-th term of the sum by 
\begin{align*}
    &\qquad  n^t \inparen{ \frac{c_0}{\eps_2^4}}^t \inparen{ \poly\inparen{\frac{ \log n }{\eps_2}} + n\Tm_{\cA_0} + \Tmdectwo}\inparen{ \frac{\Tmdetect}{n} + \frac{\Tmdectwo}{n^2}+1}
\exp\inparen{ \inparen{\frac{ c_0}{\eps_2^3}} \kappa^{j^3} 
\log^{j+1} L_0 },
\end{align*}
where we have absorbed the $L_j = \exp(\kappa^{j^3} \log^{j+1}L_0 )$ into the existing exponential term, at the cost of increasing the constant $c_0$ slightly. Plugging this into \Cref{eq:lastterm}, we see that (again adjusting $c_0$ as necessary)
\begin{align*} \eqref{eq:lastterm} &\leq n^{t+1} \inparen{ \frac{c_0}{\eps_2^4}}^t \inparen{ \poly\inparen{\frac{ \log n }{\eps_2}} + n\Tm_{\cA_0} + \Tmdectwo}\inparen{ \frac{\Tmdetect}{n} + \frac{\Tmdectwo}{n^2}} \sum_{j=3}^t
\exp\inparen{ \inparen{\frac{ c_0}{\eps_2^3}} \kappa^{j^3} \log^{j+1} L_0 } \\
&\leq n^{t+1} \inparen{ \frac{c_0}{\eps_2^4}}^t \inparen{ \poly\inparen{\frac{ \log n }{\eps_2}} + n\Tm_{\cA_0} + \Tmdectwo}\inparen{ \frac{\Tmdetect}{n} + \frac{\Tmdectwo}{n^2}} t \cdot \exp\inparen{ \inparen{ \frac{c_0}{\eps_2^3}} \kappa^{t^3} \log^{t+1} L_0 }.
\end{align*}

Putting this together with \Cref{eq:firstterm} and \Cref{eq:secondterm}, we have
\begin{align*}
    \Tmpre_{\DALLR{t}} &\leq 
    O\inparen{ n^{t+1}L_0^{c_0} \inparen{ \frac{c_0}{\eps_2^4}}^{t-2}(n\Tm_{\cA_0} + \Tmdectwo) } \\
    &\qquad + n^{t-1} \polylog(n) \frac{ c_0^{t-1}}{\eps_2^{4t + O(1)}} \\
    &\qquad + n^{t+1} \inparen{ \frac{c_0}{\eps_2^4}}^t \inparen{ \poly\inparen{ \frac{\log n}{\eps_2}} + n\Tm_{\cA_0} + \Tmdectwo}\inparen{ \frac{\Tmdetect}{n} + \frac{\Tmdectwo}{n^2}} t \cdot e^{\inparen{ \inparen{ \frac{c_0}{\eps_2^3}} \kappa^{t^3} \log^{t+1} L_0 }} \\
    &\leq n^{t+1} \inparen{ \frac{c_0}{\eps_2^4}}^t \inparen{ \poly\inparen{ \frac{\log n}{\eps_2}} + n\Tm_{\cA_0} + \Tmdectwo}\inparen{ \frac{\Tmdetect}{n} + \frac{\Tmdectwo}{n^2} } t \cdot e^{\inparen{ \inparen{ \frac{c_0}{\eps_2^3}} \kappa^{t^3} \log^{t+1} L_0 }},
\end{align*}
where we observe that the first two terms are subsumed by the third, 

and adjusting $c_0$ appropriately.  Finally, by absorbing $\log L_0$ into $c_0$,  recalling that $\kappa = \poly(1/\delta(C), 1/\rho_0)$ can also be absorbed into $c_0$, resulting in an $\exp(c_0^{t^3})$ term, and observing that we can absorb the factors of $t$ and $c_0^t$ into that term by adjusting $c_0$, we conclude that
\[ \Tmpre_{\DALLR{t}} \leq n^{t+1} \frac{1}{\eps_2^{4t + O(1)}} \inparen{n\Tm_{\cA_0} + \Tmdectwo}\inparen{ \frac{\Tmdetect}{n} + \frac{\Tmdectwo}{n^2} }  \cdot \exp\inparen{ c_0^{t^3} \eps_2^{-3} }.\]
This establishes the claimed time bound.

\paragraph{Space bounds.}
     Now, we turn to the space.  We begin with $\Speval_{\DALLR{t}}$.  
     We recall from \Cref{prop:local_timespace} that

\[\Speval_{\DALLR{t}} = O\inparen{\widetilde{M}_t\log (n|\Sigma|) + t\inparen{\Sp_{\cA_0} + \Tm_{\Gamma} + L_0 n^2 \log|\Sigma|} + \Sp_{\mathrm{Dec}_{C \otimes C}}}.\]
Now we can plug in $\tilde{M}_t$ from \Cref{eq:Mtildet}; and recall that $\Tm_{\Gamma}= \poly( \eps_2^{-1} c_0 \log n )$.  We see that, adjusting $c_0$ as necessary, the terms that dominate above are 
\begin{equation} \label{eq:speval} \Speval_{\DALLR{t}} \leq 
c_0\cdot t\cdot \inparen{\poly\inparen{\frac{\log(n|\Sigma|)}{\eps_2}} + n^2\log|\Sigma| }+ O(t \cdot \Sp_{\cA_0} + \Spdectwo).\end{equation}

Now \Cref{eq:speval} gives the claimed bound on $\Speval_{\DALLR{t}}$.
     
     Now we move on to $\Sppre_{\DALLR{t}}$.
     Recall from \Cref{thm:mainTech} that
      \[ \Sppre_{\DALLR{t}} = \Sppre_{\DALLR{t-1}} + O\inparen{m_t \log n + \Tm_{\Gamma} +  m_t \cdot L_{t-1} \cdot \mathrm{Len}_{A^{(t-1)}} + L_t \cdot \mathrm{Len}_{A^{(t)}} + \Sp_{\Testt}}.  \]
      We first simplify this, by recalling from \Cref{prop:local_timespace} that 
      \[ \mathrm{Len}_{A^{(t)}} = O(M_t( \tilde{r}_t + \log L_0)).\]
      Thus, we have
      \[ m_t L_{t-1} \mathrm{Len}_{A^{(t-1)}} =O(m_t L_{t-1} M_t (\tilde{r}_{t-1} + \log L_0 ) )\leq O(m_t L_t M_t (\tilde{r}_t + \log L_0 ))\]
      and
      \[ L_t \mathrm{Len}_{A^{(t)}} = O(L_t M_t ( \tilde{r}_t + \log L_0 )),\]
      so we can collapse those two terms inside the big-Oh notation, and obtain 
      \[ \Sppre_{\DALLR{t}} = \Sppre_{\DALLR{t-1}} + O\inparen{ m_t \log n + \Tm_{\Gamma} + m_tL_tM_t(\tilde{r}_t + \log L_0) + \Sp_{\Testt}}.  \]
      Unrolling the recursion, we have
      \begin{equation*}\label{eq:unroll_space}
      \Sppre_{\DALLR{t}} = \Sppre_{\DALLR{2}} + O\inparen{  \tilde{M}_t \log n + t \cdot \Tm_{\Gamma} + \sum_{j=3}^t m_j L_j M_j (\tilde{r}_j + \log L_0) + \sum_{j=3}^t \Sp_{\textsc{Test}^{(j)}}},
      \end{equation*}
      recalling that $\tilde{r}_t = \sum_{j=2}^t r_j$ and that $\tilde{M}_t = \sum_{j=2}^t m_j$.  Since $L_j$, $M_j$, and $\tilde{r}_j$ are all increasing in $j$, and $m_j$ is decreasing in $j$, 
      we may bound
      \[ \sum_{j=3}^t m_j L_j M_j (\tilde{r}_j + \log L_0) \leq t \cdot m_2 L_t M_t (\tilde{r}_t + \log L_0).\]
      Similarly, using \Cref{cor:test_resources} and \Cref{prop:local_timespace}, we can see that $\Sp_{\textsc{Test}^{(j)}}$ is increasing in $j$, so we can also bound 
      \[\sum_{j=3}^t \Sp_{\textsc{Test}^{(j)}} \leq t \cdot \Sp_{\Testt}. \]
      Thus,
      \begin{equation}\label{eq:space_ugly} \Sppre_{\DALLR{t}} = \Sppre_{\DALLR{2}} + O\inparen{ \tilde{M}_t \log n + t \inparen{ \Tm_{\Gamma} + m_2L_t M_t( \tilde{r}_t + \log L_0) + \Sp_{\Testt}}}.\end{equation}
      From \Cref{cl:A2simp}, we have
      \[ \Sppre_{\DALLR{2}} = O(L_0^{c_0} \log n + c_0 n^2 \log|\Sigma| + \Sp_{\cA_0} + \Spdectwo).\]
     We then focus on the remaining terms,
\begin{equation}\label{eq:space_unsimp} 
 O\inparen{  \tilde{M}_t \log n + t \inparen{ \Tm_{\Gamma} + m_2 L_t M_t( \tilde{r}_t + \log L_0) + \Sp_{\Testt}}}.
\end{equation}
To simplify this, we first simplify $\Sp_{\Testt}$.  Recall from \Cref{cor:test_resources} that
\[ \Sp_{\Testt} = O\inparen{ \Speval_{\DALLR{t}} + \Spdetect + \Spdectwo + t \log n + n^2 }.\]
Now we plug in \Cref{eq:speval} to conclude that
\begin{equation}\label{eq:sptest_simp}\Sp_{\Testt} = O\inparen{c_0 t \inparen{ \poly\inparen{ \frac{\log(n|\Sigma|)}{\eps_2} } + n^2 \log|\Sigma|} + t \Sp_{\cA_0} + \Spdetect + \Spdectwo } \end{equation}

Now we work out the terms of \Cref{eq:space_unsimp} other than $t\cdot \Sp_{\Testt}$.  Again plugging in the expressions for $\tilde{r}_t, \tilde{M}_t, M_t$, $m_2$, and $\Tm_{\Gamma}$ from above, along with $L_t = L_0^{\kappa^{t^3} \log^t L_0}$, we see that these terms are bounded by
\begin{align}
     & \tilde{M}_t \log n + t \inparen{ \Tm_{\Gamma} + m_2 L_t M_t( \tilde{r}_t + \log L_0)} \notag \\
     &\qquad \leq \frac{ tc_0}{\eps_2^3} \log n + t \inparen{ \poly\inparen{ \frac{c_0 \log n }{\eps_2}} + \frac{ c_0^{t+1}}{\eps_2^{3t}} \cdot L_0^{\kappa^{t^3} \log^t L_0} \inparen{ t \log\inparen{\frac{c_0 n}{\eps_2}} + \log L_0} } \notag \\
     &\qquad \leq t \cdot c_0 \cdot \poly\left(\frac {\log(n)} {\epsilon_2}\right) + \inparen{\frac{c_0}{\eps_2^3}}^t \cdot \exp(\kappa^{t^3} \log^{t+1}L_0) \cdot \log n, \label{eq:extraterms}
\end{align}
where in the final line we have adjusted $c_0$ appropriately to absorb lower-order terms into $(c_0 / \eps_2^3)^t$.

Now, \Cref{eq:space_ugly} implies that $\Sppre_{\DALLR{t}}$ is bounded by the sum of \Cref{eq:extraterms}, $\Sppre_{\DALLR{2}}$ and $t \cdot \Sp_{\Testt}$.  Plugging in \Cref{eq:sptest_simp} for $\Sp_{\Testt}$ and  as in \Cref{cl:A2simp}, we have
\begin{align*}
    \Sppre_{\DALLR{t}} &= O(L_0^{c_0} \log n + c_0 n^2 \log|\Sigma| + \Sp_{\cA_0} + \Spdectwo) \\
    &\qquad + O\inparen{ t\cdot c_0 \cdot  \poly\left(\frac {\log(n)} {\epsilon_2}\right) + \inparen{ \frac{c_0}{\eps_2^3}}^t \cdot \exp(\kappa^{t^3} \log^{t+1} L_0 ) \cdot \log n } \\
    &\qquad + O\inparen{c_0 t \inparen{ \poly\inparen{ \frac{\log(n|\Sigma|)}{\eps_2} } + n^2 \log|\Sigma|} + t \Sp_{\cA_0} + \Spdetect + \Spdectwo } \\
    &= O\inparen{ t \Sp_{\cA_0} + \Spdetect + \Spdectwo + \frac{c_0 \cdot t \cdot n^2\polylog|\Sigma|}{\poly(\eps_2)} +  \eps_2^{-3t} \cdot \exp(c_0^{t^3} ) \cdot \log n }.
\end{align*}
In the final line, we have used the fact that $\polylog(n |\Sigma|) = O( n^2 \polylog|\Sigma|)$ and we have absorbed both the $c_0^t$ term and the $\exp( \kappa^{t^3} \log^{t+1} L_0 )$ terms into $\exp( c_0^{t^3} )$, adjusting the constant $c_0$ as necessary and recalling that $\kappa = \poly(1/\delta(C), 1/\rho_0)$ and so $\kappa$ can be absorbed into $c_0$.  We have also absorbed $L_0^{c_0} \log n$ into $\exp(c_0^{t^3}) \log n$.
    This gives us our final bound for $\Sppre_{\DALLR{t}}$.

    Finally, we record the output length, which is
    \begin{align*}
        \Spop_{\DALLR{t}} &= L_t \cdot \mathrm{Len}_{A^{(t)}} \\
        &= L_0^{\kappa^{t^3} \log^t L_0 } \cdot M_t (\tilde{r}_t + \log L_0 )\\
        &\leq \exp(c_0^{t^3}) \cdot \inparen{ \frac{ c_0^t }{\eps_2^{3t}} } \inparen{t \log\inparen{ \frac{c_0 n }{\eps_2}} + \log L_0}\\
        &\leq \frac{\exp(c_0^{t^3})\cdot \log n}{\eps_2^{3t+O(1)}}.
    \end{align*}

    Above, in the second line we have used \Cref{prop:local_timespace} for $\mathrm{Len}_{A^{(t)}}$ and \Cref{def:params2} for $L_t$; in the third line we have used \Cref{eq:Mt} and \Cref{eq:trt} to plug in for $M_t$ and $\tilde{r}_t$.  Finally, adjusting $c_0$ as necessary, we have absorbed most of the terms into the $\exp(c_0^{t^3})$ term.
    \end{proof}

    \subsubsection{Proof of \Cref{thm:mainLocal}}\label{sec:pfmainLocal}
    Finally, we observe that we have proved \Cref{thm:mainLocal}.
    \begin{proof}[Proof of \Cref{thm:mainLocal}]
        The proof follows immediately from \Cref{lem:unwind_correct} (which proves correctness) and \Cref{lem:unwind_resources} (which establishes the time and space).
    \end{proof}

\section{Deterministic Locally List-Recoverable Codes}\label{sec:DLLR}
So far, \Cref{thm:mainLocal} tells us how to construct a \ADLLR\ using tensor codes.  However, for our final construction, we will want a deterministic local list-recovery algorithm (\DLLR), without the ``A''; aka, the special case of \Cref{def:DALLR} when $\eps = 0$.  To obtain this, we intersect our \ADLLR\  from \Cref{thm:mainLocal} with a Deterministic Locally Correctable Code (\DLCC, \Cref{def:DLC}) using  \Cref{lem:mainDLLR}.  Then we apply the AEL distance amplification to boost the result to a capacity-achieving DLLRC using \Cref{lem:AEL}. 

In this section, we put the pieces together to carry out this plan.  In \Cref{subsec:instant_DALLR} we instantiate our tensor-based DALLRC with an appropriate base code to obtain a high-rate \ADLLR. Then, in Section \ref{subsec:instant_DLCC} we construct a high-rate DLCC based on the methods of \cite{CM25}.   Finally, in Section \ref{subsec:instant_DLLR} we obtain our high-rate DLLRC by intersecting these two codes, and we then apply distance amplification to obtain our capacity-achieving DLLRC, which as a corollary also gives a capacity-achieving time- and space-efficient list recoverable (and list-decodable) code.

\subsection{High-Rate \ADLLR s}\label{subsec:instant_DALLR}
In this section we instantiate \Cref{thm:mainLocal} with an appropriate base code $C$ to obtain a high-rate \ADLLR.  Our main result is the following.

\begin{theorem}[High-rate \ADLLR s]
    \label{thm:instantiate_DALLR}
    For any constant $\xi > 0$ and constant $\ell \geq 1$, there are constants  $\eps_0 > 0$ and $q_0>1$ (which may depend on $\xi, \ell$) so that the following holds. 
    For any $\eps < \eps_0$ and any prime power $q \geq q_0$, and for infinitely many integers $N$, there is an explicit linear code $C \subseteq \F_q^N$ with rate at least $1 - \xi$ and distance $\delta =\Omega(1)$, that is a $(Q,\eps, \rho, \ell, L)$-\ADLLR\  with
    \[ Q = \frac{N^\xi}{\poly(\eps)} \qquad \rho = \poly(\eps) \qquad L = O(1).\]
    The preprocessing time and space are
    \[ \Tmpre = {N^{1 + \xi}}\cdot \poly(q) \cdot \exp(\eps^{-3})  \qquad \Sppre = N^\xi \cdot \left(\frac{\polylog(q)}{\poly(\eps)}+\poly(q)\right). \]
    The evaluation time and space are
    \[ \Tmeval = N^\xi \cdot \poly(q)+ \frac{\polylog(N)}{\poly(\eps)} \qquad \Speval = N^\xi \cdot \poly(q) + \frac{\polylog(qN)}{\poly(\eps)},\]
    and the output length is 
    $\Spop = O\left(\frac{\log N} {\poly(1/\epsilon)} \right).$

Further, for any sufficiently large integer $N'$, there is a legitimate block length $N$ for the above code so that $N'(1 - o(1)) \leq N \leq N'.$
    
    Above, all of the asymptotic notation is as $N \to \infty$ and $\eps \to 0$, assuming that $\ell, \xi$ are constant.
   
\end{theorem}

In order to prove \Cref{thm:instantiate_DALLR}, 
we will use the following base code, derived from the AEL construction of \cite{ST25}.
\begin{restatable}{theorem}{basecode}\emph{[List-recoverable base code]}\label{thm:basecode}
    Fix $\zeta > 0$ and a positive integer $\ell$.  
    There is a choice of 
    \[ q_0 = \ell^{O(1/\zeta^2)}, \qquad L_0 = \exp\inparen{\exp\inparen{ \ell^{O\inparen{(2\ell)^{1/\zeta}/\zeta^2}}}}\]
    and a choice of 
    \[ d_0 = \ell^{O\inparen{\frac{(2\ell)^{1/\zeta}}{\zeta^2}}}\]
    so that the following holds.
    For sufficiently large $n \in \mathbb{N}$ so that $n$ is a multiple of $d_0$, 
     and for any prime power $q \geq q_0$, there is a linear code $C \subseteq \F_q^n$ so that $C$ is $(\rho_0, \ell, L_0)$-globally list-recoverable deterministically in time $$O_{\ell,\zeta}\inparen{n^{3.5} + n q^{O_{\ell,\zeta}(1)}},$$ with $\rho_0 = \zeta^2$.  Further, $C$ has rate $1 - O(\zeta)$ and distance $\delta(C) = \Omega(\zeta^2)$, and a description\footnote{In more detail, the code is a concatenated code, with an explicit outer code, and a short inner code that is found by brute force; the time reported is the time to do the brute force search.} of it can be constructed in time $q^{O_{\zeta,\ell}(1)}$.
\end{restatable}

In order to obtain the code $C$ as in \Cref{thm:basecode}, we begin with a code from \cite{ST25}; that work gives a capacity-achieving list-recoverable code with fast algorithms.  However, the code is $\F_q$-linear and the alphabet for this code is $\F_q^d$.  For our purposes, we need a code that is linear over its alphabet, not over a smaller field; this is so that we can apply properties of tensor codes.  Thus, we concatenate the code of \cite{ST25} with a small linear code in order to bring the alphabet down to $\F_q$.  Since this transformation is standard, we defer the proof of \Cref{thm:basecode} to the appendix; see \Cref{app:baseCode}.

As we recall, we require the base code $C$ to not just be list-recoverable, but we also need to be able to detect errors, and we need to be able to uniquely decode $C \otimes C$ efficiently.  In the following corollary, we show that we can do that for the $C$ in \Cref{thm:basecode}.

\begin{remark}
    In \Cref{cor:basecode} below, we use \Cref{thm:basecode} as a black-box to obtain polynomial-time algorithms for detecting errors in $C$ and uniquely decoding $C \otimes C$.  In fact, since the code in \Cref{thm:basecode} is an AEL code concatenated with a small linear code, there are near-linear-time algorithms for both of these tasks.  However, as our eventual block length will be $N = n^t$ for large $t$, the difference between time $\tilde{O}(n)$ and $\poly(n)$ will not matter much quantitatively for our purposes, and we prefer a simpler argument that does not require opening up \Cref{thm:basecode}.
\end{remark}

\begin{corollary}\label{cor:basecode}
There is an absolute constant $b_0$ so that the following holds.
    Let $C \subseteq \F_q^n$ be the code from \Cref{thm:basecode}.  Then there is an algorithm $\textsc{Detect}_{C}$ that can detect errors in $C$ in time and space $O(n^{b_0}\cdot \polylog(q))$, and there is an algorithm $\mathrm{Dec}_{C \otimes C}$ that can uniquely decode $C \otimes C$ up to radius $\rho_2^{(!)} = \Omega(\zeta^4)$ in time $O_\zeta(n^{b_0})$ and space $O_{\zeta}(n^{b_0} \log(q))$ bits. 
\end{corollary}
\begin{proof}
    To implement $\textsc{Detect}_C$, we observe that since $C$ is linear, we may simply store the parity-check matrix $H$, and on input $x$, check that $H x = 0$. This takes time $O(n^2\polylog(q))$ and space $O(n^2 \log q )$ bits.

    To implement $\mathrm{Dec}_{C \otimes C}$, we first use the list-recovery algorithm guaranteed in \Cref{thm:basecode} with $\ell = 1$ as a unique-decoding algorithm for $C$, up to radius $\rho_1 = \min(\rho_0, \frac{\delta(C)} 2) = \Omega(\zeta^2)$.  Call this unique decoding algorithm $\mathrm{Dec}_C$. 
    Note that the running time and hence space usage of $\mathrm{Dec}_C$ is $O_{\zeta}(n^{3.5})$.
    
    Choose $\rho_2^{(!)} = \rho_1^2 = \Omega(\zeta^4)$.  Now we define the algorithm $\mathrm{Dec}_{C \otimes C}$ as follows, on an input $x \in \F_q^{n\times n}$.
    \begin{itemize}
        \item Run $\mathrm{Dec}_C$ on all of the rows of $x$ to obtain $x'$.  If the algorithm does not find a close-by codeword, replace that row with $\bot^n$.
        \item Run $\mathrm{Dec}_C$  on all of the columns of $x'$  to obtain $x''$.
        \item Return $x''$.
    \end{itemize}
    To see that this algorithm is correct, suppose that $\delta(x, C\otimes C) \leq \rho_2^{(!)}$.   Notice that since $\rho_2^{(!)} = \rho_1^2 \leq (\frac {\delta(C)} 2)^2 \leq \frac{(\delta(C))^2} 2$, there is a unique codeword $\Phi(x) \in C \otimes C$ that is within radius $\rho_2^{(!)}$ of $x$, and
    \[ \delta(x, \Phi(x)) \leq \rho_2^{(!)}.\]
    
    By Markov's inequality, at least a $1 - \rho_1$ fraction of the rows $i \in [n]$ have 
    $$\delta(x|_{\{i\} \times[n]}, \Phi(x)|_{\{i\} \times [n]}) \leq \frac{\rho_2^{(!)}}{\rho_1} = \rho_1.$$
    Thus, at least a $1 - \rho_1$ fraction of the rows are decoded correctly in the first step.  Thus, we are guaranteed that all of the columns $j \in [n]$ satisfy
    $$\delta(x'|_{[n]\times \{j\}}, \Phi(x)|_{[n]\times \{j\}}) \leq \rho_1,$$
    which means that in the second step, the decoder $\mathrm{Dec}_C$ is successful on all columns.  This implies that $x'' = \Phi(x)$, and the decoding is successful.

    The running time is dominated by the cost of running $\mathrm{Dec}_C$ $O(n)$ times, for a total of  $O_\zeta(n^{4.5})$. The space is naively bounded by the same quantity times $\log(q)$. 
\end{proof}

Finally, we can plug \Cref{thm:basecode} and \Cref{cor:basecode} into \Cref{thm:mainLocal} to obtain our \ADLLR.
We begin with a version in terms of a parameter $\zeta$ and the tensor parameter $t$. Then  \Cref{thm:instantiate_DALLR} will follow by choosing $\zeta$ and $t$. 

\begin{lemma}[High-Rate {\ADLLR}s in terms of $\zeta$ and $t$]
    \label{lem:instantiate_DALLR}
    Fix a parameter $\zeta > 0$, and fix a positive integer $\ell$.  Then there is an absolute constant $b_0$, and a constant $c_0 = c_0(\zeta, \ell)$ that depends only on $\zeta$ and $\ell$, a constant $q_0 = \ell^{O(1/\zeta^2)}$, and a constant $d_0 = \ell^{O( (2\ell)^{1/\zeta} /\zeta^2 )}$ so that for sufficiently large $n$ that is a multiple of $d_0$, the following holds.

    Fix any prime power $q \geq q_0$, and any $t \geq 3$.  Fix $\eps > 0$ so that 
    $\eps \leq \zeta^{b_0\cdot t^3}$.  Then
    there is an explicit linear code $\Cot \subseteq \F_q^N$ for $N = n^t$, with rate $1 - O(t \cdot \zeta)$ and distance $\Omega(\zeta^{2t})$, that is a $(Q, \eps, \rho, \ell, L)$-\ADLLR, for
    \[ \rho = \eps^{t-1}\zeta^{b_0 t^3} ,\]
    \[ L = \exp\inparen{ \zeta^{-b_0 t^3} c_0^t},\]
    and
    \[ Q = n(n+t) \cdot \frac{ c_0^t }{\zeta^{b_0t^3}} \cdot \eps^{-3t}.\]
    Further, $\Cot$ has a DALLR algorithm $\DALLR{t}$ with preprocessing time and space
    \[ \Tmpre_{\DALLR{t}} \leq n^{t + b_0} \cdot \frac{1}{\eps^{4t + b_0}}\exp\inparen{  \frac{ c_0^{t^3}}{\eps^{3}}} \cdot q^{c_0} \]
    \[ \Sppre_{\DALLR{t}} \leq n^{b_0}q^{c_0}+\frac{c_0 \cdot \polylog(q) \cdot \zeta^{-b_0\cdot {t^2}}}{\poly(\eps)} \cdot n^{b_0}+ \frac{\exp(c_0^{t^3}) }{\eps^{3t}} \cdot \log n\]
    and evaluation time and space
    \[ \Tmeval_{\DALLR{t}} \leq q^{c_0} n^{b_0} + \frac{ c_0^t }{\eps^{3t+O(1)} \zeta^{b_0 t^3}} \cdot \polylog(n)\]
    \[ \Speval_{\DALLR{t}} \leq tn^{b_0}q^{c_0}+ \frac{c_0 \cdot \poly(\log(q)  ,\log(n))}{\poly(\eps) \zeta^{b_0 t^2} }.\]
    Finally, the output length is given by 
    \[\Spop_{\DALLR{t}}\leq \frac{ \exp(c_0^{t^3}) \log n}{\eps^{3t+O(1)} \zeta^{b_0 t^3}}.\]
\end{lemma}
\begin{proof}
    Let $C \subseteq \F_q^n$ be as in the statement of \Cref{thm:basecode} and \Cref{cor:basecode}, with this choice of $\zeta$ and $\ell$.  
    Then $C$ is $(\rho_0, \ell, L_0)$-list-recoverable, has distance $\delta(C)$, and $C \otimes C$ is uniquely decodable up to radius $\rho_2^{(!)}$, for
    \[ \delta(C) = \Omega(\zeta^2), \qquad \rho_0 = \zeta^2, \qquad \rho_2^{(!)} = \Omega(\zeta^4).\]
    Further, as in the statement of \Cref{thm:basecode}, we may take $n$ to be any sufficiently large multiple of $d_0$, where $d_0$ is as in the theorem statement.
    Let $L_0 = L_0(\ell, \zeta)$ be as in \Cref{thm:basecode}.
    Let $\kappa = \kappa(1/\delta(C), 1/\rho_0) = \poly(1/\zeta)$ be the constant from \Cref{cor:HRW}.  

    Now consider $\Cot$.      First, we observe that the rate of $C$ is $R = (1 - O(\zeta))^t \geq 1 - O(t\cdot \zeta)$ and the distance is $(\delta(C))^t = \Omega(\zeta^{2t})$.
    
    Next we apply \Cref{thm:mainLocal} with the base code $C$ to show that $\Cot$ is indeed a \ADLLR.
    First, we observe that the requirement on $\eps$ is
    \[ \eps \leq \min\inset{\frac{\rho_0}{4 \kappa^{t^3}}, \frac{ (\delta(C))^t }{4}},\]
    which is met as long as $\eps \leq \zeta^{b_0 t^3}$ for sufficiently large $t$.

    Next, we recall that in the statement of \Cref{thm:mainLocal} we have set
    \begin{equation}\label{eq:eps2} \eps_2 = \inparen{ \prod_{j=3}^t \gamma_j } \inparen{ \frac{ \delta(C) \rho_2^{(!)}}{2d_0} }^{t-2} \cdot \eps \geq { \zeta^{b_0 t^2}} \cdot \eps, \end{equation}
    adjusting $b_0$ as necessary.  Above, we have used the definition of $\gamma_j$ in \Cref{def:params} to see that
    \[ \gamma_j = \frac{(\delta(C))^{2j}}{18^{\log_{1.5}(j)}} = \frac{ (\delta(C))^{2j}}{j^{\log_{1.5}(18)}} \]
    and so $\prod_{j=3}^t \gamma_j \geq {(\delta(C))^{b_0t^2}}$ for some constant $b_0$, and then we have used the fact that $\delta(C)$ and $\rho_2^{(!)}$ are all $\poly(\zeta)$, and that $d_0$ is an absolute constant. 

    Having worked out these parameters, we go through the two parts of \Cref{thm:mainLocal}.

    \begin{itemize}
        \item \textbf{Correctness.}
        \Cref{thm:mainLocal} says that $\Cot$ is a $(Q, \eps, \rho, \ell, L)$-\ADLLR; we simplify each of the parameters $\rho, L, Q$ given in that theorem below, in terms of $\zeta$.

        First, we have, for some absolute constant $b_0$ (which we may take to be the same as the constant $b_0$ above by increasing whichever is smaller),
        \begin{align*}
            \rho &\geq \eps^{t-1} \inparen{ \frac{ \rho_2^{(!)}}{2 d_0 t^{b_0}}}^{t^2} (\delta(C))^{t^3} \rho_0 \min \inset{ \frac{\delta(C)}{8}, \frac{\rho_2^{(!)}}{2}, \frac{1}{\kappa^8}} \\
            &= \eps^{t-1}\zeta^{b_0 t^3} t^{-b_0 t^2}\\
            &\geq \eps^{t-1} \zeta^{b_0 t^3},
        \end{align*} 

        adjusting the constant $b_0$ as necessary. 

        Then we have
        \begin{align*}
            L = L_0^{\kappa^{t^3} \log^t L_0} 
            = \exp\inparen{ \zeta^{-b_0 t^3} c_0^t},
        \end{align*}
        where again $b_0$ is an absolute constant (adjusted as necessary), and $c_0$ is a constant that depends only on $\zeta$ and $\ell$.  

        Finally, we have
        \begin{align*}
        Q &\leq n(n+t)\inparen{\frac{b_0 L_0}{\eps_2^3}}^t \\
        &\leq n(n+t) c_0^t \inparen{ \frac{ 1}{\zeta^{b_0 t^2} \eps}}^{3t}\\
        &= n(n+t) \cdot \frac{ c_0^t }{\zeta^{b_0t^3}} \cdot \eps^{-3t},
        \end{align*}
        again adjusting $c_0$ as necessary.
        \item \textbf{Resources.}
        Note that the constant $c_0$ from \Cref{thm:mainLocal} depends only on $\delta(C), \rho_2^{(!)}, \rho_0$ and $L_0$, and so depends only on $\ell$ and $\zeta$.  Thus, we can take the constant $c_0$ from that theorem to be the same $c_0$ as above, by increasing whichever is smaller.  With that in mind, \Cref{thm:mainLocal} gives the following bounds on the time and space of \Cref{alg:main} for $\Cot$.  

        The preprocessing time is bounded by 
        \begin{align*}
            \Tmpre_{\DALLR{t}} 
            &\leq n^{t+1} \frac{1}{\eps_2^{4t + O(1)}} \inparen{ n\Tm_{\cA_0} + \Tmdectwo}\inparen{ \frac{\Tmdetect}{n} + \frac{\Tmdectwo}{n^2} }  \cdot \exp\inparen{ c_0^{t^3} \eps_2^{-3} }.\\
            &\leq n^{t + b_0} \cdot \frac{1}{\eps_2^{4t + b_0}}\exp( c_0^{t^3} \eps_2^{-3}) \cdot q^{c_0}\\
            &\leq n^{t + b_0} \frac{1}{\zeta^{b_0 t^3} \eps^{4t + b_0} }\cdot \exp\inparen{  \frac{ c_0^{t^3}}{\zeta^{b_0 t^2} \eps^{3}}}\cdot q^{c_0}\\
            &\leq n^{t + b_0} \cdot \frac{1}{\eps^{4t + b_0}}\cdot \exp\inparen{  \frac{ c_0^{t^3}}{\eps^{3}}}\cdot q^{c_0},
        \end{align*}
        where in the last line we have absorbed the terms with $\zeta$ into $c_0$, and adjusting $c_0$ as necessary.
        The preprocessing space is bounded by
        \begin{align*}
            \Sppre_{\DALLR{t}} &=
O\inparen{ t \cdot \Sp_{\cA_0} + \Spdetect + \Spdectwo + \frac{c_0 \cdot t \cdot n^2\polylog(q)}{\poly(\eps_2)} +  \eps_2^{-3t} \cdot \exp(c_0^{t^3} ) \cdot \log n } \\
&\leq n^{b_0}q^{c_0}+ t\cdot c_0 \cdot n^{b_0} \cdot \poly\inparen{\frac{ \log(q)}{\eps}} \zeta^{-b_0t^2} + \frac{\zeta^{-b_0t^3}}{\eps^{3t}} \exp(c_0^{t^3}) \cdot \log n  \\
&\leq   n^{b_0}q^{c_0}+ \frac{c_0 \zeta^{-b_0t^2} \polylog(q)}{\poly(\eps)} \cdot n^{b_0}+ \frac{\exp(c_0^{t^3}) }{\eps^{3t}} \cdot \log n   \\
        \end{align*}
        Above, in the second line we have used the fact that $\Sp_{\cA_0}, \Spdetect, \Spdectwo$ are all bounded by $O_{\zeta, \ell}(n^{b_0} q^{O_{\zeta,\ell}(1))}) = O_{\zeta, \ell}(n^{b_0}q^{c_0})$ by adjusting $c_0$; and we have used our expression for $\eps_2$ from \Cref{eq:eps2}.
        In the third line 
        we have absorbed the factor of $t$ into $\zeta^{b_0^{t^2}}$, adjusting $b_0$ appropriately, and we have absorbed the factor of $\zeta^{-b_0t^3}$  into $\exp(c_0^{t^3})$.

Next, we consider the evaluation time.  We have
        \begin{align*}
    \Tmeval_{\DALLR{t}} &\leq O\inparen{\frac{c_0^t}{\eps_2^{3t}} \cdot \poly\inparen{\frac{\log n}{\eps_2}} + n \cdot \Tm_{\cA_0} + \Tmdectwo}\\
&\leq q^{c_0} n^{b_0} + \frac{ c_0^t }{\eps^{3t+O(1)} \zeta^{b_0 t^3}} \cdot \polylog(n)
    \end{align*}
    for a sufficiently large $c_0$ that depends only on $\zeta$ and $\ell$.
    Next, we consider the evaluation space.  We have
    \begin{align*}
    \Speval_{\DALLR{t}} &\leq 
c_0\cdot t\cdot \inparen{\poly\inparen{\frac{\log(nq)}{\eps_2}} + n^2\log(q) }+ t \cdot \Sp_{\cA_0} + \Spdectwo \\
&\leq t \cdot n^{b_0}q^{c_0}+c_0 \cdot  \polylog(q) \inparen{ \frac{\polylog(n)}{\poly(\eps) \zeta^{b_0 t^2}} },
\end{align*}
again using that $\Sp_{\cA_0}, \Spdectwo = O_{\zeta, \ell}(n^{b_0} q^{c_0})$ for some constant $b_0$ and constant $c_0$ that depends only on $\zeta$ and $\ell$. 
Finally, we bound the output length by 
\begin{align*} \Spop_{\DALLR{t}} &\leq \frac{\exp(c_0^{t^3}) \log n}{\eps_2^{3t+O(1)}} \\
&\leq \frac{ \exp(c_0^{t^3}) \log n}{\eps^{3t+O(1)} \zeta^{b_0 t^3}}.
\end{align*}
    \end{itemize}
This proves the lemma.
\end{proof}
\Cref{lem:instantiate_DALLR} immediately leads to \Cref{thm:instantiate_DALLR}: 
\begin{proof}[Proof of \Cref{thm:instantiate_DALLR}]
    We apply \Cref{lem:instantiate_DALLR}. We choose an arbitrarily small constant $\xi$ and set $\zeta = O(\xi^3)$ and $t = b_0/\xi$, so that the rate of $C$ is $$1 - O(t \cdot \zeta) \geq 1 - O(t \cdot \xi^3) = 1 - O( b_0 \xi^2 ) \geq 1 - \xi$$ for sufficiently small $\xi$. Above, $b_0$ is the constant from \Cref{lem:instantiate_DALLR}.  Plugging these into \Cref{lem:instantiate_DALLR} immediately gives the desired bounds.

Next we verify the claim about the allowable block lengths.  
Fix a target block length $N'$; we wish to show that we can take $N$ so that $N'(1 - o(1)) \leq N \leq N'$.
The block length is $N = n^t$, where $t = b_0/\xi = O(1)$, and where $n$ can be any sufficiently large multiple of some constant $d_0 = d_0(\xi, \zeta) = O(1)$.  So in particular we can choose $n$ so that $(N')^{1/t} \geq n \geq (N')^{1/t} - d_0$.  As $t$ and $d_0$ are constant, this implies that
\[ N' \geq n^t \geq N'(1 - o(1)),\]
as desired.

Finally, note that since the base code is explicit, the tensor code is also explicit by known properties of tensor codes (See e.g., \cite[Section 2.4]{KRRSS20}).
\end{proof}

\subsection{High-Rate \DLCC s}\label{subsec:instant_DLCC}

In this section, we establish the existence of high-rate DLCCs.  

\begin{theorem}[High-rate DLCCs]\label{thm:DLCC_main}
    Let $\eta>0$ be a constant, and let $q$ be a fixed power of $2$.  Then for infinitely many integers $N$, there is  an explicit linear code $C \subseteq \F_q^N$ with
    rate at least $1 - \eta$ and distance $\delta =\Omega(1)$, that is a $(Q, \rho)$-\DLCC\ with
$Q = N^{o(1)}$ and  $\rho = \Omega(1)$. The pre-processing time is $N^{1 + o(1)}$, the pre-processing space, evaluation time, and evaluation space are $N^{o(1)}$, and the output length is  $\tilde{O}(\log N)$.
\end{theorem}

The proof of the above theorem relies on the following theorem, which gives DLCCs as above, albeit with a sub-constant distance and decoding radius and super-constant alphabet size.

\begin{restatable}{theorem}{DLCCthm}\emph{[Similar to \cite{CM25}, Lemma 4.22]}
    \label{thm:DLCC}
    Let $\eta>0$ be a constant, and let $q$ be a fixed power of $2$.  Then for infinitely many integers $N$, there is  an explicit $\F_q$-linear code $C \subseteq (\F_q^b)^N$ with
    alphabet size $q^b = N^{o(1)}$, rate at least $1 - \eta$, and distance $\delta = N^{-o(1)}$, that is a $(Q, \rho)$-\DLCC\ with
$Q = N^{o(1)}$ and  $\rho = N^{-o(1)}$. The pre-processing time is $N^{1 + o(1)}$, the pre-processing space, evaluation time, and evaluation space are $N^{o(1)}$, and the output length is  $\tilde{O}(\log N)$.
\end{restatable}

\begin{remark}
    In fact, the codes $C$ in \Cref{thm:DLCC} are linear over their alphabet (treated as $\Sigma = \F_{q^b}$), not just linear over a subfield $\F_q$.  However, since our next step will be to apply the AEL transformation, we state the result of \Cref{thm:DLCC} about $\F_q$-linear codes over $\Sigma = \F_q^b$.
\end{remark}

The proof of \Cref{thm:DLCC} is an adaptation of the proof of Lemma 4.22 in \cite{CM25}.  That work gives a low-space unique-decoding algorithm for lifted Reed-Solomon codes.  As mentioned in the introduction, \cite{CM25} implicitly defines and uses a DLCC, and in fact they show that lifted Reed-Solomon codes are DLCCs.  However,  their Lemma 4.22 does not explicitly establish the high-rate guarantee that we need, only that the codes are asymptotically good.  Thus, for completeness, we include the proof of \Cref{thm:DLCC} in \Cref{app:DLCC}.

\begin{proof}[Proof of \Cref{thm:DLCC_main}]
We first apply the AEL transformation for DLLRCs (\Cref{lem:AEL}, specialized to the DLCC case) to amplify the distance and decoding radius of the codes given by Theorem \ref{thm:DLCC} to a constant,  and then apply the concatenation lemma for DLLRCs (Lemma \ref{lem:DLLR_concat}, specialized to the DLCC case) on the resulting codes to reduce the alphabet size to a constant. Details follow.

In what follows, fix a constant $\eta>0$ and a $q$ which is a power of $2$.  Let $C \subseteq (\F_q^b)^N$ be the explicit $\F_q$-linear code with alphabet size  $q^b=N^{o(1)}$, rate $1 - \frac \eta 3$, and distance $\delta=N^{-o(1)}$, given by Theorem \ref{thm:DLCC}. 
    Then $C$ is  a $(Q, \rho)$-\DLCC\ with
$Q = N^{o(1)}$ and  $\rho = N^{-o(1)}$, and with pre-processing time $N^{1 + o(1)}$, pre-processing space, evaluation time, and evaluation space $N^{o(1)}$, and output length $\tilde{O}(\log N)$.

\paragraph{AEL transformation:} 

We first apply the AEL transformation of \Cref{lem:AEL} on $C$ to increase its  distance and decoding radius to a constant. 

Let $d= d(\delta, \rho,  \eta^4)=N^{o(1)}$, where $d$ is as given in the statement of \Cref{lem:AEL}. Let $C_{\inn} \subseteq \F_q^d$ be an explicit linear code of 
rate $1 - \frac{\eta} 3$ and distance $\eta^3=\Omega(1)$, that is (globally) uniquely decodable from $\eta^3=\Omega(1)$ errors in time (and space) $\poly(d) = N^{o(1)}$.  For example, concatenated Reed-Solomon codes will work (see, e.g., \cite[Theorem 14.3.3]{ECT}).  

We would like to apply the AEL transformation of \Cref{lem:AEL} with the outer code $C$ and the inner code $C_{\inn}$, and to this end we first increase the alphabet size of $C$ to $ |\Sigma_{\inn}|^{R_{\inn} \cdot N_{\inn}}=q^{(1- \eta/3) d}$. This is done as follows. Let $\tilde b= (1- \frac \eta 3) \cdot d = N^{o(1)}$. 
Let $\tilde C$ be the code obtained from $C$ by dividing the entries of $C$ into groups of size $\tilde b /b = N^{o(1)}$ (recalling that $q^b =N^{o(1)}$, and so $b=O(\log N)$) of symbols of $\F_q^b$, and viewing each group as a single symbol of $\F_q^{\tilde b}$. Then $\tilde C \subseteq (\F_q^{\tilde b})^{\tilde N}$ is an $\F_q$-linear code of block length $\tilde N = N \cdot b /\tilde b \in ( N^{1-o(1)}, N)$. It can also be verified that the rate of $\tilde C$ is the same as that of $C$, and the distance of $\tilde C$ is at least the distance of $C$.

The DLC algorithm $\tilde \cApre$ for $\tilde C$ emulates the corresponding DLC algorithm  $\cApre$ for $C$, except that for each query that $\cApre$ 
makes, $\cApre$ 
queries the corresponding symbol of $\F_q^{\tilde b}$ and retrieves from it the requested symbol of $\F_q^b$. 
For each local algorithm $A$ output by $\cApre$, the DLC algorithm $\tilde \cApre$ outputs a local algorithm $\tilde A$. Given a query to a block of size $\tilde b /b = N^{o(1)}$ of symbols in $\F_q^b$ the local algorithm $\tilde A$ executes $A$ on each of the $\tilde b /b$ entries, and simulates answers to  $A$'s queries as above. Note that this transformation does not change the decoding radius and space bounds of the DLCC, and increases the query complexity and the running time by a multiplicative factor of $\tilde b /b = N^{o(1)}$. 

 We conclude that $\tilde C \subseteq (\F_q^{\tilde b})^{\tilde N}$ is an explicit $\F_q$-linear code with alphabet size $q^{\tilde b} = q^{(1-\frac \eta 3)d}$, 
 rate $1 - \frac \eta 3$, and distance $\tilde \delta = \delta ={\tilde N}^{-o(1)}$,  that is a $(\tilde Q, \tilde \rho)$-\DLCC\ with
$\tilde Q = {\tilde N}^{o(1)}$ and  $\tilde \rho = {\tilde N}^{-o(1)}$, and with pre-processing time ${\tilde N}^{1 + o(1)}$, pre-processing space, evaluation time, and evaluation space ${\tilde N}^{o(1)}$, and output length $\tilde{O}(\log {\tilde N})$. 

Now apply the distance amplification procedure given by \Cref{lem:AEL} (specialized to the DLCC setting of $\ell_{\out}=L_{\out}=\ell_{\inn}=L_{\inn}=1$) with the outer code $\tilde C$ and the inner code $C_{\inn}$. Then the resulting code $C' \subseteq (\F_q^d)^{N'}$ 
is an explicit $\F_q$-linear code of block length $N' = \tilde N  \in (N^{1-o(1)},N)$, alphabet size $q^d = \exp((N')^{o(1)})$, 
rate at least $1 - \frac {2\eta} 3$, and distance at least $\eta^3 - 2\eta^4 = \Omega(1)$, that is a $(Q',\rho')$-\DLCC\ with 
$Q' = \tilde Q \cdot d^2 = (N')^{o(1)}$ and $\rho'\geq \eta^3 - \eta^4 = \Omega(1)$. Furthermore, it can be verified that the corresponding DLC algorithm has pre-processing time $(N')^{1 + o(1)}$, pre-processing space, evaluation time, and evaluation space $(N')^{o(1)}$, and output length $\tilde{O}(\log (N'))$.

\paragraph{Concatenation:} Next we apply the concatenation transformation of \Cref{lem:DLLR_concat} on the resulting code $C'$ to decrease its alphabet size to a constant.

Let $C'_{\inn} \subseteq \F_q^{d/(1- \eta/3)}$ be an explicit linear code of 
rate $1 - \frac{\eta} 3$ and distance $\eta^3 = \Omega(1)$, that is (globally) uniquely decodable from $\eta^3=\Omega(1)$ errors in time (and space) $\poly(d)=N^{o(1)}$. Note that the alphabet size of the code $C'$ constructed above is $q^d = |\Sigma_{\inn}|^{R_{\inn} \cdot N_{\inn}}$. Again, a concatenated RS code will work (see, e.g., \cite[Theorem 14.3.3]{ECT}). 

Now apply the concatenation procedure given by \Cref{lem:DLLR_concat} (specialized to the DLCC setting of $\ell_{\out}=L_{\out}=\ell_{\inn}=L_{\inn}=1$) with the outer code $C'$ and the inner code $C'_{\inn}$. Then the resulting code $C'' \subseteq \F_q^{N''}$ is an explicit linear  
code of block length $N'' = N' \cdot \frac {d} {1- \eta/3}  \in (N', (N')^{1+o(1)})$, 
rate at least $1 - \eta$, and distance $\Omega(1)$,  that is a $(Q'',\rho'')$-\DLCC\ with 
$Q'' = Q' \cdot O(d) = (N'')^{o(1)}$  and $\rho''=\Omega(1)$. Furthermore, it can be verified that the corresponding DLC algorithm has pre-processing time $(N'')^{1 + o(1)}$, pre-processing space, evaluation time, and evaluation space $(N'')^{o(1)}$, and output length $\tilde{O}(\log N'')$.

\end{proof}

\subsection{High-rate and capacity-achieving \DLLR s}\label{subsec:instant_DLLR}

In this section we first construct our high-rate DLLRC by intersecting the high-rate DALLRC and high-rate DLCC constructed in the previous sections, and then obtain our capacity-achieving DLLRC by applying the AEL transformation to the resulting code. 
We begin with the high-rate DLLRC construction.

\begin{theorem}[High-Rate {\DLLR}s]\label{thm:final}
For any constant $\xi > 0$ and constant $\ell \geq 1$, there is a constant $q_0>1$ (which may depend on $\xi, \ell$) so that the following holds. Let  $q \geq q_0$ be a fixed power of $2$.  Then for infinitely many $N$, there is an explicit linear code $C \subseteq \F_q^N$ with rate at least $1 - \xi$ and distance $\delta = \Omega(1)$, 
that is a $(Q,\rho, \ell, L)$-\DLLR\  with $Q=N^\xi$,  $\rho=\Omega(1)$, and $L=O(1)$. The pre-processing time is $N^{1 + \xi}$, the pre-processing space, evaluation time, and evaluation space are $N^{\xi}$, and the output length is $\tilde{O}(\log N)$.
\end{theorem}

\begin{proof}
 Fix a constant $\xi > 0$  and a constant $\ell \geq 1$, and let   $q_0=q_0(\frac \xi 2,\ell)$ and $\epsilon_0= \epsilon_0(\frac \xi 2,\ell)$ be the constants guaranteed by  \Cref{thm:instantiate_DALLR}. 
 Let $q \geq q_0$ be a fixed power of $2$.

 Then by \Cref{thm:DLCC_main}, for infinitely many values of $N'$, there is  an explicit linear code $C' \subseteq \F_q^{N'}$ of    rate $1 - \frac \xi 2$ and distance $\delta' =\Omega(1)$, that is a $(Q', \rho')$-DLCC with
$Q' = (N')^{o(1)}$ and  $\rho' = \Omega(1)$, and with pre-processing time $(N')^{1 + o(1)}$, pre-processing space, evaluation time, and evaluation space $(N')^{o(1)}$, and output length $\tilde{O}(\log N')$.

    By \Cref{thm:instantiate_DALLR}, 
    there is an explicit linear code $C \subseteq \F_q^N$, so that 
    \[ N'(1 - o(1)) \leq N \leq N',\]
    of rate $1 - \frac \xi 2$ and distance $\Omega(1)$, 
    that is a $(Q,\eps, \rho, \ell, L)$-\ADLLR\ with $\epsilon = \min\{\epsilon_0, \frac {\rho'} 2, \frac{\delta'}{4} \}= \Omega(1)$, $Q=O(N^{\xi/2})$,  
    $\rho=\Omega(1)$, and $L=O(1)$, and 
with pre-processing time $O(N^{1 + \xi/2})$, pre-processing space, evaluation time, and evaluation space $O(N^{\xi/2})$, and output length $O(\log N)$.

Now, let $C'' \subseteq \F_q^N$ be the projection of $C'$ to the first $N$ coordinates, and let $\pi:C' \to C''$ be the projection map.  Note that $\pi$ is injective, since $C'$ has constant distance, while the number of coordinates that were dropped is
\[ N' - N = o(1)\cdot N',\]
by the above. Then $C''$ has rate at least $1- \frac \xi 2$ and distance at least $\frac {\delta'} {2}$. 
We claim that $C''$ is a $(Q', \rho'/2)$-DLCC, with the same space and time requirements.  Indeed, suppose that $y \in \F_q^N$ satisfies $\delta(y, c) \leq \rho'/2$ for some $c \in C''$.  Then $\delta( (y,0), \pi^{-1}(c) ) \leq \rho'/2 + o(1) \leq \rho'$, where $(y,0) \in \F_q^{N'}$ is the vector $y$ padded with $N' - N$ zeros.  Thus, a DLC algorithm for $C''$ is as follows:
\begin{itemize}
    \item Run the DLC pre-processing algorithm for $C'$ with access to input $(y,0) \in \F_q^{N'}$, to obtain a description of a local algorithm $A$, where $A(i) = (\pi^{-1}(c))_i$ for all $i \in [N']$ (in particular, $A(i)=c_i$ for all $i \in [N]$).
    \item Return this same description of $A$, with the understanding that we will now only evaluate $A$ only on $i \in [N]$, rather than $i \in [N']$.
\end{itemize}
Further, since $N'(1 -o(1)) \leq N \leq N' $, we have $Q' = N^{o(1)},$ pre-processing time $N^{1 + o(1)}$; pre-processing space, evaluation space, and evaluation time $N^{o(1)}$; and output length $\tilde{O}(\log N)$.

Now $C$ and $C''$ have the same alphabet size and same length, so we may intersect them.
Let $C^* = C \cap C''$. Then by \Cref{lem:mainDLLR}, $C^* \subseteq \F_q^N$ is a linear code of rate at least $1 - \xi$ and distance $\Omega(1)$, that is a $(Q \cdot Q', \rho, \ell,L)$-\DLLR\, where $Q \cdot Q' =O(N^{\xi/2}) \cdot N^{o(1)} \leq N^{\xi}$ and $\rho = \Omega(1)$. Furthermore, it can be verified that the corresponding DLLR algorithm has pre-processing time $N^{1 + \xi}$, pre-processing space, evaluation time, and evaluation space $N^{\xi}$, and output length $\tilde{O}(\log N)$.

\end{proof}

We now proceed to our final construction of capacity-achieving DLLRCs.

\begin{theorem}[Capacity-Achieving DLLRCs] \label{thm:main}
Let $\gamma > 0$, $\ell \geq 1$, and $R \in (0,1)$ be constants.
Then for infinitely many integers $N$, there is an explicit code $C \subseteq \Sigma^N$ of alphabet size $|\Sigma|=O(1)$, 
rate at least $R$, and distance at least $1-R-\gamma$, 
that is a $(Q,\rho, \ell, L)$-\DLLR\  with $Q=N^\gamma$,  $\rho=1- R- \gamma$, and $L=O(1)$. The pre-processing time is $N^{1 + \gamma}$, the pre-processing space, evaluation time, and evaluation space are $N^{\gamma}$, and the output length is $\tilde{O}(\log N)$.
\end{theorem}

\begin{proof}
Fix constants $\gamma >0$, $\ell \geq 1$, and $R \in (0,1)$. 
Let $L_0= L_0(\frac \gamma 4,\ell)$ be the constant guaranteed by Theorem \ref{thm:ST25}, and let $q_1=q_1(\frac \gamma 4, \ell)$ be the minimum field size guaranteed by this theorem.
Let $q_0=q_0(\frac \gamma 4, L_0)$ be the constant guaranteed by Theorem \ref{thm:final}. Let $q \geq \max\{q_0,q_1\}$ be a fixed power of $2$.

Let $C \subseteq \F_q^N$ be the code of rate $1 - \frac \gamma 4$ and distance $\delta=\Omega(1)$, given by Theorem \ref{thm:final}.
Then $C$ is a $(Q,\rho, L_0, L)$-\DLLR\  with $Q=N^{\gamma/4}$,  $\rho=\Omega(1)$, and $L=O(1)$, and 
with pre-processing time $N^{1 + \gamma/4}$, pre-processing space, evaluation time, and evaluation space $N^{\gamma/4}$, and output length $\tilde{O}(\log N)$.

Let $d=d(\delta, \rho, \frac \gamma 4)=O(1)$, where $d$ is as given in the statement of \Cref{lem:AEL}. 
Let $C_{\inn} \subseteq \F_q^{N_{\inn}}$ be the code of constant
block length $N_{\inn} \geq d$,
rate $R + \frac \gamma 4$,  and distance at least $1-R- \frac \gamma 2$,  that is 
$(1-R - \frac \gamma 2,\ell,L_0)$-list recoverable, given by Theorem \ref{thm:ST25}.\footnote{Since the code $C_{\inn}$ is of constant size we could have alternatively also used any non-explicit capacity-achieving list-recoverable code (e.g., a random code).} Note that the code $C_{\inn}$ is of constant size, and therefore is encodable and list recoverable in constant time, and the encoding map can also be inverted in constant time.

Similarly to the proof of Theorem \ref{thm:DLCC_main}, we may assume that $C$ has alphabet size $|\Sigma_{\inn}|^{R_{\inn} \cdot N_{\inn}} = q^{(R+\gamma/4)N_{\inn}}$ by grouping together subsets of entries of $C$ of size $(R+ \frac \gamma 4)N_{\inn}=O(1)$, without meaningfully affecting any of the properties of $C$. 

Now apply the distance amplification procedure given by \Cref{lem:AEL}  with the outer code $C$ and the inner code $C_{\inn}$. Then the resulting code $C' \subseteq (\F_q^{N_{\inn}})^{N'}$ is a code of block length $N'= \Theta(N)$, alphabet size $q^{N_{\inn}}=O(1)$, rate at least $R$, and distance at least $1-R-\gamma$,  that is a $(Q',\rho',\ell,L)$-\DLLR\ with 
$Q' = Q \cdot (N_{\inn})^2 = O((N')^{\gamma/4}) \leq (N')^\gamma$, $\rho'=1-R-\gamma$, and $L=O(1)$. Furthermore, it can be verified that the corresponding DLLR algorithm has pre-processing time $O((N')^{1 + \gamma/4}) \leq (N')^{1+\gamma}$, pre-processing space, evaluation time, and evaluation space $O((N')^{\gamma/4})\leq (N')^\gamma$, and output length $\tilde{O}(\log (N'))$. 
\end{proof}

Finally, as an immediate corollary, we obtain our deterministic time- and space-efficient, capacity-achieving, (globally) list-recoverable codes.

\begin{corollary}[Capacity-Achieving Deterministic List-Recoverable Codes]\label{cor:main}
Let $\gamma > 0$, $\ell \geq 1$, and $R \in (0,1)$ be constants.
Then for infinitely many integers $N$, there is an explicit code $C \subseteq \Sigma^N$ of alphabet size $|\Sigma|=O(1)$, 
rate $R$, and distance $1-R-\gamma$, 
that is $(\rho, \ell, L)$-(globally) list recoverable deterministically with $\rho=1-R-\gamma$ and $L=O(1)$, 
and in time $N^{1+\gamma}$ and space $N^{\gamma}$.
\end{corollary}

\begin{proof}
The statement follows as an immediate corollary of Theorem \ref{thm:main}, using Lemma \ref{lem:dllr_to_global}.
\end{proof}

\bibliographystyle{alpha}
\bibliography{ref.bib}

\appendix

\section{Distance amplification for DLLRC (Proof of Lemma \ref{lem:AEL})}
\label{sec:AEL}

In this section we prove Lemma \ref{lem:AEL}, restated below, which shows how to transform a \emph{high-rate} \DLLR\ into a \emph{capacity-achieving} \DLLR.

\AEL*

The proof is based on the AEL transformation \cite{AEL95}, and is very similar to \cite[Lemma 5.4]{GKORS18} which showed a similar transformation for the setting of (randomized) local list recovery, and we provide here a full proof for completeness.

Following \cite{AEL95,GKORS18}, 
a main ingredient in the construction is a family of $d$-regular bipartite expanders 
which have the property of being good {\it balanced samplers}, defined as follows. 
For a graph $G$, a vertex $s$ and a set of vertices $T$, let $E(s,T)$
denote the set of edges that go from~$s$ into~$T$. 

\begin{definition}[Balanced samplers]\label{def:expander_sampler} Let $G=\left(U\cup V,E\right)$ be a bipartite
$d$-regular graph with $\left|U\right|=\left|V\right|=n$. We say
that $G$ is an \emph{$(\eta,\nu)$-balanced sampler} if the following
holds for every $T\subseteq V$: For at least $1-\eta$~fraction
of the vertices $s\in U$ it holds that 
\[
\frac{\left|E(s,T)\right|}{d}-\frac{\left|T\right|}{n}\leq\nu.
\]
\end{definition}

\begin{remark}
The above definition of a balanced sampler corresponds to a special case of the sampler given by Definition \ref{def:samplerGamma} where the randomness is $\log(n)$ and the function $f: [n] \to \{0,1\}$ is Boolean (the randomized algorithm picks a random left vertex, and outputs its neighborhood).  
\end{remark}

\begin{lemma}[\cite{KMRS2017}, Lemma 2.12; \cite{GKORS18}, Lemma 7.2]
\label{lem:expander-samplers} For every $\eta,\nu>0$, there exists 
$\hat d = \poly(\frac{1}{\eta \nu})$ such that
for every sufficiently large $n $ and for every $d> \hat d$ there exists a bipartite 
$d$-regular graph $G_{n,d,\eta,\nu}=\left(U\cup V,E\right)$ with 
$\left|U\right|=\left|V\right|=n$ such that $G_{n, d,\eta,\nu}$
is an $\left(\eta,\nu\right)$-balanced sampler. Furthermore, there exists
an algorithm that takes as inputs $n$, $d$, $\eta$, $\nu$ and
a vertex $w$ of $G_{n,d,\eta,\nu}$, and computes the list of
the neighbors of $w$ in $G_{n,d,\eta,\nu}$ in time 
$\poly(d, \log n)$.\end{lemma}

We now turn to the proof of \Cref{lem:AEL}.

\begin{proof}[Proof of \Cref{lem:AEL}]

First, we describe the construction of the code $C$ using the balanced samplers above. The construction is identical to that of \cite[Lemma 5.4]{GKORS18}, and we repeat the details of the construction below for the sake of completeness. 

\paragraph{Construction of code $C$.}
We construct $C$ by giving a bijection from $C_{\out}$ to $C$. Let 
$\Sigma_{\out}, \Sigma_{\inn}$ denote the alphabets of $C_{\out}, C_{\inn}$ respectively.
Given a codeword $c_{\out}\in C_{\out}$, one obtains the corresponding codeword
$c\in C$ as follows: 
\begin{itemize}
\item View each codeword symbol in $\Sigma_{\out}$ as a vector of length $R_{\inn}\cdot N_{\inn}$ over $\Sigma_{\inn}$ and encode it via the code $C_{\inn}$. Each codeword symbol gets mapped to a string in  $\Sigma_{\inn}^{N_{\inn}}$. We denote the resulting string by $c' \in \Sigma_{\inn}^{N_{\inn}\cdot N_{\out}}$ and the various resulting codewords of $C_{\inn}$ by $B_1, B_2, \ldots B_{N_{\out}} \in \Sigma_{\inn}^{N_{\inn}}$.  
 
\item Next, we apply a ``pseudorandom'' permutation to the coordinates
of~$c'$ as follows: Let $G_{N_{\out}}$ be a graph from the infinite
family of $N_{\inn}$-regular $(\min\{\rho_{\out}, \frac {\delta_{\out}} 2\}, \gamma)$-balanced samplers above and let $U=\left\{ u_{1},\ldots,u_{N_{\out}}\right\} $ and
$V=\left\{ v_{1},\ldots,v_{N_{\out}}\right\} $ be the left and right vertices
of~$G_{N_{\out}}$ respectively. For each $i\in\left[N_{\out}\right]$ and $j\in\left[N_{\inn}\right]$,
we write the $j$-th symbol of $B_{i}$ on the $j$-th edge of $u_{i}$.
Then, we construct new blocks $D_{1},\ldots,D_{N_{\out}}\in \Sigma_{\inn}^{N_{\inn}}$, by setting
the $j$-th symbol of $D_{i}$ to be the symbol written on the $j$-th
edge of $v_{i}$. 
We reinterpret each of these blocks to be a symbol of the new alphabet $\Sigma:=\Sigma_{\inn}^{N_{\inn}}$.

\item Finally, we define the codeword $c$ of~$C\subseteq\Sigma^{N_{\out}}$ as
follows: the $i$-th coordinate $c_{i}$ is the block $D_{i}$, reinterpreted
as a symbol of the alphabet $\Sigma$. We choose $c$
to be the codeword in $C$ that corresponds to the codeword $c_{\out}$ in
$C_{\out}$. 
\end{itemize}

This completes the definition of the bijection. It follows that $C$ is an $\F$-linear code of blocklength $N_{\out}$ and alphabet size $\Sigma_{\inn}^{N_{\inn}}$. Rate, distance, and encoding time follow similarly to the proof of \cite[Lemma 5.4]{GKORS18}. We now turn to describe and analyze the DLLR algorithm for the code $C$. The algorithm and its analysis are very similar to that of \cite[Lemma 5.4]{GKORS18}, and we refer to \cite{GKORS18} for some of the technical claims.

\paragraph{DLLR algorithm for $C$.}
We will now describe the $(Q,\rho_{\inn}-\gamma,\ell_{\inn},L_{\out})$-DLLR algorithm $\cApre$ for the code $C$. This is based on the following deterministic algorithm $\tilde \cApre$ which locally list recovers coordinates of $C_{\out}$ (instead of coordinates of $C$, as required of $\cApre$).  

\begin{lemma}\label{lem:AEL_inter}
There exists an algorithm $\tilde \cApre$ which satisfies the following:
\begin{itemize}
 \item $\tilde \cApre$ is a deterministic algorithm which receives as input $S \in {\Sigma \choose {\leq \ell_{\inn}}}^{N_{\out}}$.
    \item   $\tilde \cApre$ outputs a list of  local algorithms $(\tilde A_1, \ldots, \tilde A_{L_{\out}})$  
in time  
$$ T_{\out}^{(Pre)} \cdot \left(T_{\inn} + L_{\inn} \cdot \Tm_{\inn}^{(unenc)} + \poly(N_{\inn}, \log(N_{\out})) \right)  $$ and space 
$$ \Sp_{\out}^{(Pre)} + \Sp_{\inn} +  \poly(N_{\inn}, \log(N_{\out})) + N_{\inn} \cdot (\ell_{\inn}+L_{\inn}) \cdot \log(|\Sigma_{\inn}|) + \Sp_{\inn}^{(unenc)},$$ where the output length is at most $O(\Spop_{\out})$. 
  
    \item Each local algorithm $\tilde A_j$ is a deterministic algorithm with query access to $\mathcal{S}$ and receives as input a coordinate $i \in [N_{\out}]$.  On input $i \in [N_{\out}]$, $\tilde A_j$ makes at most $Q_{\out} \cdot N_{\inn}$ queries to $\cS$, and runs in time  
$$ \Tm_{\out}^{(Eval)} \cdot \left(\Tm_{\inn} + L_{\inn} \cdot \Tm_{\inn}^{(unenc)}+ \poly(N_{\inn}, \log(N_{\out})) \right)  $$ and space 
$$ \Sp_{\out}^{(Eval)} + \Sp_{\inn} +  \poly(N_{\inn}, \log(N_{\out})) + N_{\inn} \cdot (\ell_{\inn}+ L_{\inn}) \cdot \log(|\Sigma_{\inn}|) + \Sp_{\inn}^{(unenc)}.$$
    \item  For each $c_{\out} \in C_{\out}$ such that the corresponding codeword $c$ of $C$ (as given by the bijection above) satisfies $\dist(c,\cS) \leq \rho_{\inn} - \gamma$, 
    there is some $j \in [L_{\out}]$ so that $\tilde A_j(i)=(c_{\out})_i$ for any 
    $i \in [N_{\out}]$.
\end{itemize}
\end{lemma}

\begin{proof}
Let $\bar \cApre$ be the DLLR algorithm for $C_{\out}$. $\bar \cApre$ is a deterministic algorithm that given
a tuple $\overline \cS = (\overline S_1, \ldots , \overline S_{N_{\out}}) \in {\Sigma_{\out} \choose {\leq \ell_{\out}}}^{N_{\out}}$
 outputs a list of $L_{\out}$ deterministic local algorithms $\bar A_1, \bar A_2, \ldots, \bar A_{L_{\out}}$. 

We now describe $\tilde \cApre$. 
Suppose the algorithm $\tilde \cApre$ is invoked on a tuple $\cS = (S_1, \ldots, S_{N_{\out}})\in {\Sigma \choose {\leq \ell_{\inn}}}^{N_{\out}}$, the algorithm $\tilde \cApre$ invokes the algorithm $\bar \cApre$ and emulates $\bar \cApre$ in the natural way. 
Recall that $\bar \cApre$ expects to be given access to a tuple $\bar S \in  {\Sigma_{\out} \choose {\leq \ell_{\out}}}^{N_{\out}}$. 
For any $k \in [N_{\out}]$, whenever  $\bar \cApre$ (or any of the local algorithms $\bar A_j$) queries the $k$th element of the sequence $\bar S_1, \ldots, \bar S_{N_{\out}} \in  {\Sigma_{\out} \choose {\leq \ell_{\out}}}$, the algorithm $\tilde \cApre$ (and any of the local algorithms $\tilde A_j$) performs the following steps.
\begin{enumerate}

\item In the first step, for each coordinate $r \in [N_{\inn}]$ of $B_k$, $\tilde \cApre$ will find a list $S^{(k,r)} \in  {\Sigma_{\inn} \choose {\leq \ell_{\inn}}}$ and associate that list with the $r$th coordinate of $B_k$. The list 
 $S^{(k,r)}$ is defined as follows: Suppose that $v_{k_{r}}$ is the $r$th neighbor of the vertex $u_k$ in $G_{N_{\out}}$. Suppose that $u_k$ is the $\hat r$th neighbor of the vertex $v_{k_{r}}$. Then in the construction of the codeword $c$ from $c_{\out}$, the value of the $r$th coordinate of $B_k$ is stored in the $\hat r$th coordinate of $D_{k_r}$. Now $S_{k_r} \in {\Sigma \choose {\leq \ell_{\inn}}} = {\Sigma_{\inn}^{N_{\inn}} \choose {\leq \ell_{\inn}}}$ is the input list associated with the $k_r$th coordinate. Note that each element $s\in S_{k_r}$ can be viewed as an $N_{\inn}$-tuple of elements from $\Sigma_{\inn}$. Let the $\hat r$th element of this tuple be $s^{(\hat r)}$. Then $S^{(k,r)}$ is defined to be the set in ${\Sigma_{\inn} \choose {\leq \ell_{\inn}}}$ obtained by taking the $\hat r$th element of each member of the set $S_{k_r}$.
$\tilde \cApre$ can find this set by making a single query to the $k_r$th element of $S$ to obtain $S_{k_r}$, and from it find $S^{(k,r)}$. 

\item $\tilde \cApre$ then invokes the global list-recovery algorithm for $C_{\inn}$ with the lists $S^{(k,r)}$ for each $r \in [N_{\inn}]$. The output of this algorithm is a list of size at most $L_{\inn} \leq \ell_{\out}$ with elements from $\Sigma_{\inn}^{ N_{\inn}}$. We denote by $\bar S_k$ the set of messages in $\Sigma_{\inn}^{r_{\inn} N_{\inn}} = \Sigma_{\out}$ corresponding to the codewords in this list.
This is what $\tilde \cApre$ feeds to $\bar \cApre$ (and each of the local algorithms $\tilde A_j$ feeds to
$\bar A_j$). 

\end{enumerate}

Clearly, the query complexity of each algorithm~$\tilde A_j$ is at most $N_{\inn}$ times the query complexity of~$\bar A_j$, and hence it is at most $Q_{\out} \cdot N_{\inn}$.
Next assume that $c_{\out} \in C_{\out}$ is such that the corresponding 
codeword $c$ of $C$ (as given by the bijection above) satisfies 
$\dist(c,S)\leq \rho_{\inn} - \gamma$. Claim 7.4 in \cite{GKORS18} states that the tuple  $\bar S :=(\bar S_1, \bar S_2, \ldots, \bar S_{N_{\out}})$ as defined 
above satisfies $\dist(c_{\out}, \bar S) \leq \rho_{\out}$.
Consequently, by the correctness of the algorithm $\bar \cApre$, we have that $\tilde A_j(i)=(c_{\out})_i$ for any   $i \in [N_{\out}]$.

It remains to analyze the running time and space of the algorithm $\tilde \cApre$ and each of the local algorithms $\tilde A_j$. To this end, note that the algorithm $\tilde \cApre$ ($\tilde A_j$, respectively) emulates the algorithm $\bar \cApre$ ($\bar A_j$, respectively), where in each time step of  $\bar \cApre$ ($\bar A_j$, respectively), the algorithm $\tilde \cApre$ ($\tilde A_j$, respectively) needs to find the set of neighbors of $u_k$ in $G_{N_{\out}}$, as well as the set of neighbors of $v_{k_r}$ for any $r \in [N_{\inn}]$, store the values of $S^{(k,r)}$, 
then invoke the global list recovery algorithm for $C_{\inn}$, and find and store the preimages of the codewords in the output list under the encoding map of $C_{\inn}$. By Lemma \ref{lem:expander-samplers} and our choice of $N_{\inn} \geq d$, finding the neighbor sets in $G_{N_{\out}}$ takes time and space at most $\poly( N_{\inn}, \log(N_{\out})) $, while storing the values takes  space at most  $N_{\inn} \cdot \ell_{\inn} \cdot \log(|\Sigma_{\inn}|)$. 

Consequently,  $\tilde \cApre$ has running time at most $$ T_{\out}^{(Pre)} \cdot \left(T_{\inn} + L_{\inn} \cdot \Tm_{\inn}^{(unenc)} + \poly(N_{\inn}, \log(N_{\out})) \right)  $$
and space at most
$$ \Sp_{\out}^{(Pre)} + \Sp_{\inn} +  \poly(N_{\inn}, \log(N_{\out})) + N_{\inn} \cdot (\ell_{\inn}+L_{\inn}) \cdot \log(|\Sigma_{\inn}|) + \Sp_{\inn}^{(unenc)}.$$ 
Similarly, each local algorithm $\tilde A_j$ has running time at most 
$$ \Tm_{\out}^{(Eval)} \cdot \left(\Tm_{\inn} + L_{\inn} \cdot \Tm_{\inn}^{(unenc)}+ \poly(N_{\inn}, \log(N_{\out})) \right)  $$ and space at most $$ \Sp_{\out}^{(Eval)} + \Sp_{\inn} +  \poly(N_{\inn}, \log(N_{\out})) + N_{\inn} \cdot (\ell_{\inn}+ L_{\inn}) \cdot \log(|\Sigma_{\inn}|) + \Sp_{\inn}^{(unenc)}.$$
\end{proof}

Given such an algorithm $\tilde \cApre$ guaranteed by the above Lemma~\ref{lem:AEL_inter}, we show how to construct the required DLLR algorithm $\cApre$ for the code $C$. The algorithm 
$\cApre$ is given access to an $\cS \in {\Sigma \choose {\leq \ell_{\inn}}}^{N_{\out}}$,  and outputs a list of $L_{\out}$ local algorithms $A_1$, ..., $A_{L_{\out}}$ which list recover all codewords $c\in C$ that ``disagree" with $\cS$ in at most $\rho_{\inn} - \gamma$ fraction of coordinates. 
Recall that the algorithm $\tilde \cApre$ from the above Lemma \ref{lem:AEL_inter}
is given access to the same $\cS$ and outputs a list of $L_{\out}$ local algorithms $\tilde A_1$, ..., $\tilde A_{L_{\out}}$ which list recover all codewords of $C_{\out}$ such that the corresponding codeword $c$ of $C$ satisfies $\dist(c,\cS) \leq \rho_{\inn} - \gamma$.

The algorithm $\cApre$ emulates $\tilde \cApre$, and for any algorithm $\tilde A_j$ output by $\tilde \cApre$, it outputs an algorithm $A_j$ defined as follows. 
Each $A_j$ takes as input a coordinate $i \in [N_{\out}]$ and also gets oracle access to the tuple $\cS$.
Let $B_1,\ldots, B_{N_{\out}}$ and $D_1,\ldots, D_{N_{\out}}$ be the corresponding blocks that arise in the construction of $c$ from $c_{\out}$. 
In order for $A_j(i)$ to be able to decode the value of $c_i$, it should be able to correctly decode all the symbols in the block $D_i$. Let $u_{i_1}, \ldots, u_{i_{N_{\inn}}}$ be the neighbors of $v_i$ in the graph $G_{N_{\out}}$. Each symbol of $D_i$ belongs to one of the blocks $B_{i_1}, \ldots, B_{i_{N_{\inn}}}$, and therefore it suffices to retrieve these blocks.  Each of these blocks $B_{i_r}$ is the encoding of $c_{\out_{i_r}}$
 (the $i_r$th symbol of $c_{\out}$) via the code $C_{\inn}$. Thus to recover   $B_{i_1}, \ldots, B_{i_{N_{\inn}}}$, it suffices to recover  $c_{\out_{i_1}}, \ldots, c_{\out_{i_{N_{\inn}}}}$. The algorithm $A_j$ invokes the algorithm $\tilde A_j$ to recover each of $c_{\out_{i_1}}, \ldots, c_{\out_{i_{N_{\inn}}}}$, and it then retrieves the blocks $B_{i_1}, \ldots, B_{i_{N_{\inn}}}$ by encoding each of the symbols in $\Sigma_{\inn}^{r_{\inn}\cdot N_{\inn}}$ using $C_{\inn}$, 
 and hence also $D_i$ and hence $c_i$. 

Clearly the query complexity of $A_j$ is $N_{\inn}$ times the query complexity of $\tilde A_j$, and is hence at most $Q_{\out} \cdot N_{\inn}^2$. Correctness also follows immediately by the correctness of $\tilde \cApre$. 
The running time and space of $\cApre$ is also clearly the same as that of $\tilde \cApre$, and so it remains to analyze the running time and space of each of the local algorithms $A_j$. To this end, note that on input $i \in [N_{\out}]$, $A_j$ first needs to find the set of neighbors of $v_i$ in $G_{N_{\out}}$ with takes time and space at most $\poly( N_{\inn}, \log(N_{\out})) $. Then for each $r=1, \ldots, N_{\inn}$, the algorithm $A_j$ needs to invoke the algorithm
$\tilde A_j$  to obtain $c_{\out_{i_r}}$, then encode the resulting symbol using $C_{\inn}$, and finally store a single symbol in $\Sigma_{\inn}$ out of the encoding. Thus if $\tilde A_j$ has running time $\tilde T$ and space $\tilde \Sp$, then  $A_j$ has running time at most
$ N_{\inn} \cdot (\tilde T + T_{\inn}^{(enc)}) +\poly( N_{\inn}, \log(N_{\out})),$ and space at most $\tilde S_p + \Sp_{\inn}^{(enc)} +
N_{\inn} \cdot \log(|\Sigma_{\inn}|)+\poly( N_{\inn}, \log(N_{\out}))$. 
Finally, plugging the running time and space of the local algorithm $\tilde A_j$ given in Lemma \ref{lem:AEL_inter} gives the claimed running time and space of the local algorithms $A_j$.

The output length is also clearly $\Spop_{\out}+ O(1)$. 
\end{proof}

\section{Correctness of DALLR algorithm - missing proofs}\label{app:proofOfCorrectnessClaim}

  In this appendix, we complete the proofs of \Cref{cl:goodAt} and \Cref{lem:base_correct}.  As noted earlier, the proofs are similar to the proofs of the analogous statements in \cite{GGR11, HRW19, KRRSS20}.

  \subsection{Proof of \Cref{cl:goodAt}}\label{app:proofOfCorrectnessClaim_t}

  In this section, we prove \Cref{cl:goodAt}, which we restate for the reader's convenience.

  \goodAt*

The rest of this section is devoted to the proof of \Cref{cl:goodAt}.
     Fix a seed $\sigma^{(t)}$ and let $H_t \subseteq [n]$ be the set of $m_t$ indices chosen by the sampler $\Gamma$ using the seed $\sigma^{(t)}$  in \Cref{alg:main}.
    Recall that in \Cref{alg:main}, for each such seed $\sigma^{(t)}$, we get lists $\cL_1^{(t-1)}, \ldots, \cL_{m_t}^{(t-1)}$ of representations of local algorithms for $C^{\otimes (t-1)}$, where the list $\cL_a^{(t-1)}$ corresponds to the $h_a$'th column.

Given $\sigma^{(t)}$ and $\tilde{c}$, define local algorithms $A_a^{(t-1)}$ for $a \in [m_t]$ and a vector of local algorithms $\vec{A}$ by
\begin{equation}\label{eq:repA}
\begin{gathered}
\rep(A_a^{(t-1)}) = \operatorname*\argmin_{\rep(A^{(t-1)}) \in \mathcal{L}_a^{(t-1)}} \delta( \tilde{c}|_{[n]^{t-1} \times \{h_a\}}, x(A^{(t-1)}) ) , \\
\text{and} \\
\vec{A} = \vec{A}(\sigma^{(t)}, \tilde{c}) = (A_1^{(t-1)}, \ldots, A_{m_t}^{(t-1)}),
\end{gathered}
\end{equation}
    where we recall from \Cref{def:alg_to_string} that $x(A_a^{(t-1)})$ is the string corresponding to $A_a^{(t-1)}$.  
    Let $A^{(t)}$ be the local algorithm with representation 
    \[ \rep(A^{(t)}) = (\sigma^{(t)}, \rep(A_1^{(t-1)}), \ldots, \rep(A_{m_t}^{(t-1)})) = (\sigma^{(t)}, \vec{A}),\]
    as in \Cref{alg:local-alg}; note that $\rep(A^{(t)})$ is one of the candidate local algorithms formed in \Cref{line:localalg} in \Cref{alg:main}.
    We would like to show that, with probability at least $\frac 1 2$ over the choice of the seed $\sigma^{(t)}$, this choice of $\vec{A}$ is good, meaning that
    the string $x(A^{(t)})$ is close to $\tilde{c}$.  
    
    To do this, following \cite{HRW19}, we will define certain rows $i' \in [n]^{t-1}$ to be \emph{good}, and we will show that $A^{(t)}$ will correctly decode all the coordinates of $\tilde{c}$ in such rows.  Then we will show that with high probability over the choice of $\sigma^{(t)}$, many rows are good.  

    \begin{definition}[Good row]\label{def:good-row}
    Let $\tilde{c} \in \Cot$, and let $\vec{A} = (A_1^{(t-1)}, \ldots, A_{m_t}^{(t-1)})$ be as in \Cref{eq:repA} above.  Let $i' \in [n]^{t-1}$.  We say that row $i'$ is a \emph{good row} (with respect to $\tilde{c}$, $\sigma^{(t)}$, and $\vec{A}$) if the following are satisfied.  Below, $H_t \subseteq [n]$ is the set of indices chosen by the sampler $\Gamma^{(t)}$ with seed $\sigma^{(t)}$.
    \begin{enumerate}
        \item $\delta( \tilde{c}|_{\{i'\} \times [n]}, \cS|_{\{i'\} \times[n]} ) \leq \rho_{0}$.
        \item Let $\mathcal{K}_{i'} = \inset{ c \in C \,:\, \delta(c, \cS|_{\{i'\} \times [n]} ) \leq \rho_{0}}$. 
        Then for any $c'' \in \mathcal{K}_{i'}$ so that $c'' \neq \tilde{c}|_{\{i'\} \times [n]}$,
 $$\frac{1}{|H_t|} \sum_{h \in H_t}\mathbf{1}[\tilde{c}_{(i',h)} \neq c''_h] > \frac{3 \delta(C)} 4.$$
        \item Let $v \in \Sigma^{m_t}$ be given by $v_a = A_a^{(t-1)}(i')$ for any $a \in [m_t]$.   Then $\delta(\tilde{c}|_{\{i'\} \times H_t}, v) \leq \frac{\delta(C)} 4.$
    \end{enumerate}
    \end{definition}
     Next, we show that $A^{(t)}$ (as defined above) correctly decodes on a good row. 
     \begin{claim}\label{cl:decodegood}
         Suppose that the hypotheses of \Cref{thm:mainTech} are satisfied, and suppose that $\tilde{c} \in \Cot$ has $\delta(\tilde c, \cS) \leq \rho_t$.  Suppose that $i' \in [n]^{t-1}$ is a good row, with respect to $\tilde{c}$, $\sigma^{(t)}$, and $\vec{A}$.  Let $A^{(t)}$ be the local algorithm represented by $\rep(A^{(t)}) = ( \sigma^{(t)}, \vec{A})$.  Then $A^{(t)}(i) = \tilde{c}_i$ for all $i \in \{i'\} \times [n]$.
     \end{claim}
     \begin{proof}
         Suppose that $i' \in [n]^{t-1}$ is a good row.  Let $c' = \tilde{c}|_{\{i'\} \times [n]} \in C$.  By Item (1) in \Cref{def:good-row}, 
         \[ \delta(c', \cS|_{\{i'\} \times [n]}) \leq \rho_0,\]
         which implies that $c' \in \mathcal{K}_{i'}$, where $\mathcal{K}_{i'} = \inset{ c \in C \,:\, \delta(c, \cS|_{\{i'\} \times [n]} ) \leq \rho_{0}}$.  
         Further, Condition (3) of \Cref{def:good-row} says that \begin{equation}\label{eq:vclose}
         \delta(c'|_{H_t}, v) \leq \frac{\delta(C)} 4,
         \end{equation}
         where $v \in \Sigma^{m_t}$ is the vector so that $v_a = A_a^{(t-1)}(i')$ for any $a \in [m_t]$.
         
         Let $c'' \in \mathcal{K}_{i'} \setminus \{c'\}$ be any other element of $\mathcal{K}_{i'}$.  Then the triangle inequality implies that
         \[ \delta(c''|_{H_t},v) \geq \delta(c'|_{H_t}, c''|_{H_t}) - \delta(c'|_{H_t}, v) \geq \delta(c'|_{H_t}, c''|_{H_t}) - \frac{\delta(C)}{4},\]
         using \Cref{eq:vclose} in the final inequality.  Then by Condition (2) of \Cref{def:good-row}, we have $\delta(c'|_{H_t}, c''|_{H_t}) \geq \frac{3\delta(C)} 4,$ so
         \begin{equation}\label{eq:uniquely_close}
         \delta(c''|_{H_t},v) \geq \frac{3 \delta(C)}{4} - \frac{\delta(C)}{4} >  \frac{\delta(C)}{4} \geq \delta(c'|_{H_t}, v),
         \end{equation}
         where we have again used \Cref{eq:vclose}.

        Now, we recall from \Cref{alg:local-alg} that the way that $A^{(t)}$ works is that, on input $i=(i',i_t) \in [n]^{t-1} \times [n]$, it obtains the list $\mathcal{K}_{i'} \subseteq C$ as above, and then chooses some $c^* \in \mathcal{K}_i$ to return $c^*_{i_t}$.  The $c^*$ it chooses is
        \[ c^* = \argmin_{c \in \mathcal{K}_{i'}} |\{a \in [m_t] : c_{h_a} \neq v_a \}| = \argmin_{c \in \mathcal{K}_{i'}} \delta(c|_{H_t}, v).\]
        Thus, \Cref{eq:uniquely_close} implies that the algorithm $A^{(t)}$ will choose $c^* = c'$, and so on input $i=(i',a) \in [n]^{t-1} \times [n]$, it will return $c'_{i_t} = \tilde{c}_{(i',i_t)} = \tilde c_i$, which proves the claim.
     \end{proof}

     The next step is to show that for many seeds $\sigma^{(t)}$, most rows are good with respect to $\tilde{c}$, $\sigma^{(t)}$, and $\vec{A}$, where $\vec{A}$ is as in \Cref{eq:repA}.  Before we make this precise, we introduce the definition of a \emph{good} column, and show that for most seeds $\sigma^{(t)}$, most columns in $H_t$ are good.

     \begin{definition}[Good column]\label{def:good-col}
     Let $\tilde{c} \in \Cot$.  A column $h \in [n]$ is \emph{good} (with respect to $\tilde{c}$) if $$\delta(\tilde{c}|_{[n]^{t-1} \times \{h\}}, \cS|_{[n]^{t-1} \times\{h\}}) \leq \rho_{t-1}.$$
     \end{definition}

     \begin{claim}[Most columns are good]\label{cl:good-col}
     Suppose that $\tilde{c} \in \Cot$ so that $\delta(\tilde{c}, \cS) \leq \rho_t$.
     Let $H_t \subseteq [n]$ be the set of indices sampled by $\Gamma_t$ with seed $\sigma^{(t)}$.  Let
     \[ \widehat{H}_t = \{ h \in {H}_t \,:\, \text{ column $h$ is good with respect to $\tilde{c}$}\}.\]
     Then with probability at least $1 - \eta_t$ over the choice of $\sigma^{(t)}$ (and hence over the choice of $H_t$),
     \[ |\widehat{H}_t| \geq (1 - 2\eps_{t-1})|H_t|.\]
     \end{claim}
     \begin{proof}
         First, we note that
         \begin{align*} \mathbb{E}_{h \in [n]} \delta(\tilde{c}|_{[n]^{t-1} \times \{h\}}, \cS|_{[n]^{t-1} \times \{h\}})
         &= \delta(\tilde{c}, \cS) \\
         &\leq \rho_t \qquad\qquad \text{by assumption} \\
         &= \eps_{t-1}\rho_{t-1} \qquad \ \ \text{by the definition of $\rho_t$ (\Cref{def:params})}.
         \end{align*}
         Thus, Markov's inequality implies that
         \[ \Pr_{h \in [n]}[\delta(\tilde{c}|_{[n]^{t-1} \times \{h\}}, \cS|_{[n]^{t-1} \times \{h\}}) > \rho_{t-1}] \leq \eps_{t-1}. \]
         Let $f(h) = \mathbf{1}[\delta(\tilde{c}|_{[n]^{t-1} \times \{h\}}, \cS|_{[n]^{t-1} \times \{h\}}) > \rho_{t-1}]$.  By the properties of the $(n, \eta_t, \nu_t)$-sampler $\Gamma_t$ (\Cref{thm:sampler}, along with the observation that the parameters in \Cref{def:params2} are chosen to satisfy the requirements of \Cref{thm:sampler}), with probability at least $1 - \eta_t$ over the choice of $\sigma^{(t)}$, 
         \[ |\EE_{h \in [n]} f(h) - \EE_{h \in H_t} f(h) | \leq \nu_t \leq \eps_{t-1},\]
         where the final inequality is because $\nu_t = \min\{\eps_{t-1}, \frac{\delta(C)} 4\}$ as in \Cref{def:params2}.  Together with the triangle inequality, the above implies that
         \[ \Pr_{h \in H_t}[\delta(\tilde{c}|_{[n]^{t-1} \times \{h\}}, \cS|_{[n]^{t-1} \times \{h\}}) > \rho_{t-1}] \leq 2 \eps_{t-1}. \]
         In other words, in the favorable case, at least a $1 - 2\eps_{t-1}$ fraction of the $h \in H_t$ are good, which is what we wanted to show.
     \end{proof}
     Now, we can show that most rows are good.
     \begin{claim}\label{cl:goodrows}
     Assume the hypotheses of \Cref{thm:mainTech}.  Fix $\tilde{c} \in \Cot$ so that $\delta( \tilde{c}, \cS) \leq \rho_t$, and let $\vec{A} = \vec{A}(\sigma^{(t)}, \tilde{c})$ be as in \Cref{eq:repA}.
     
        Then
         with probability at least $1/2$ over the choice of $\sigma^{(t)}$, 
         \[ \Pr_{i' \in [n]^{t-1}}[ \text{Row $i'$ is good with respect to $(\tilde{c}, \sigma^{(t)}, \vec{A})$}] \geq 1 - \frac{d_0 \eps_{t-1}}{\delta(C)},\]
         where the probability is only over the choice of a uniform random row  $i' \in [n]^{t-1}$.
     \end{claim}

     \begin{proof}
         We will establish that each of the three conditions of \Cref{def:good-row} hold for most rows, with high probability over the choice of $\sigma^{(t)}$.  
         \begin{enumerate}
             \item[1.] We will show that $\delta( \tilde{c}|_{\{i'\} \times [n]}, \cS|_{\{i'\} \times [n]}) \leq \rho_0$ for at least a $1 - \eps_{t-1}$ fraction of the rows $i' \in [n]^{t-1}$.  Note that this property holds deterministically for all seeds $\sigma^{(t)}$.
             
             Similar to the derivation in the proof of \Cref{cl:good-col}, we have
             \begin{align*}
                 \mathbb{E}_{i' \in [n]^{t-1}} \delta( \tilde{c}|_{\{i'\} \times[n]}, \cS|_{\{i'\} \times [n]}) &= \delta(\tilde{c}, \cS) \\
                 &\leq \rho_t \\
                 &= \eps_{t-1}\rho_{t-1},
             \end{align*}
             and so by Markov's inequality
             \[ \Pr_{i' \in [n]^{t-1}}[\delta( \tilde{c}|_{\{i'\} \times[n]}, \cS|_{\{i'\} \times [n]}) > \rho_{t-1} ] \leq \eps_{t-1}.\]

             Since $\rho_0 \geq \rho_{t-1}$ by assumption, we conclude that
             at least a $1 - \eps_{t-1}$ fraction of the rows $i' \in [n]^{t-1}$ 
             have
             \[ \delta(\tilde{c}|_{\{i'\} \times [n]}, \cS|_{\{i'\} \times [n]} ) \leq \rho_{t-1} \leq \rho_0,\]
             and hence
             satisfy Condition (1) of \Cref{def:good-row}.
             \item[2.] We will show that, with probability at least $0.9$ over the choice of $\sigma^{(t)}$, for at least a $1 - \eps_{t-1}$ fraction of the rows $i' \in [n]^{t-1}$,  we have
\begin{equation}\label{eq:good_dist_Ki}
\min_{c'' \in \mathcal{K}_{i'} \setminus \{\tilde{c}|_{\{i'\} \times [n]}\}} \delta(c''|_{H_t}, \tilde{c}|_{\{i'\} \times H_t}) > \frac{3\delta(C)} 4.
         \end{equation}
          
         To establish this, first let $i' \in [n]^{t-1}$ be any row, and let $c' =\tilde{c}|_{\{i'\} \times[n]}$.  Fix $c'' \in \mathcal{K}_{i'} \setminus \{c'\}$.  Let $f:[n] \to \{0,1\}$ be given by $f(h) = \mathbf{1}[c''_h \neq c'_h]$.  Since $\Gamma_t$ is a $(n, \eta_t, \nu_t)$-sampler, \Cref{thm:sampler} (along with the observation that the parameters in \Cref{def:params2} are chosen to satisfy the requirements of \Cref{thm:sampler}) implies that, with probability at least $1 - \eta_t$ over the choice of $\sigma^{(t)}$,
         \[| \EE_{h \in [n]} f(h) - \EE_{h \in H_t} f(h) | \leq \nu_t \leq \frac{\delta(C)} 4,\]
         where in the final inequality we have used the definition of $\nu_t = \min\{\eps_{t-1}, \frac{\delta(C)} 4\}$ in \Cref{def:params2}.  Since $\EE_{h \in [n]} f(h) = \delta(c',c'') \geq \delta(C)$ by the definition of distance, the triangle inequality implies that 
         \[ \delta(c'|_{H_t}, c''|_{H_t}) = \EE_{h \in H_t} f(h) \geq \frac{3 \delta(C)}{4}.\]
         Now we union bound over all at most $L_0$ codewords in $\mathcal{K}_{i'}$ to establish that, for any fixed row $i'$,  \eqref{eq:good_dist_Ki} holds with probability at least $1 - L_0 \eta_t$ over the choice of $\sigma^{(t)}$. 

         Now, taking expectations over all of the rows $i' \in [n]^{t-1}$, we have that
         \[ \mathbb{E}_{i' \in [n]^{t-1}} \Pr_{\sigma^{(t)}}[ \eqref{eq:good_dist_Ki} \text{ does not hold } ] \leq L_0 \eta_t,\]
         or equivalently
         \[ \EE_{\sigma^{(t)}} \Pr_{i' \in [n]^{t-1}}[\eqref{eq:good_dist_Ki} \text{ does not hold }] \leq L_0 \eta_t \leq \frac{\eps_{t-1}}{10},\]
         using the definition of $\eta_t = \frac{\eps_{t-1}}{10 L_0}$ from \Cref{def:params2}. 
         Thus, by Markov's inequality, we have
         \[ \Pr_{\sigma^{(t)}}\left[ \Pr_{i' \in [n]^{t-1}}[ \eqref{eq:good_dist_Ki} \text{ does not hold} ] > \eps_{t-1} \right] \leq \frac{1}{10}.\]
         In other words, we conclude that with probability at least $0.9$ over the choice of $\sigma^{(t)}$, \eqref{eq:good_dist_Ki} holds for at least a $1 - \eps_{t-1}$ fraction of the rows $i' \in [n]^{t-1}$.

            \item[3.] Now, we will show that with probability at least $1 - \eta_t$ over the choice of $\sigma^{(t)}$, for a $1 - \frac{12 \eps_{t-1}}{\delta(C)}$ fraction of rows $i' \in [n]^{t-1}$, we have $\delta(\tilde{c}|_{\{i'\} \times H_t}, v) \leq \frac{\delta(C)} 4$, where $v \in \Sigma^{m_t}$ is given by $v_a = A_a^{(t-1)}(i')$ as in \Cref{def:good-row}, and where $\vec{A}$ is as defined in \Cref{eq:repA} .

           Let $\widehat{H}_t \subseteq H_t$ be as in \Cref{cl:good-col}, so $\widehat{H}_t$ is the set of $h \in H_t$ so that column $h$ is good with respect to $\tilde{c}$.  
           That is, for any $h \in \widehat{H}_t$, 
            \[ \delta( \tilde{c}^{(h)}, \cS|_{[n]^{t-1} \times\{h\}} ) \leq \rho_{t-1},\]
            where
            \[ \tilde{c}^{(h)} = \tilde{c}|_{[n]^{t-1} \times \{h\}}.\]
            As we are assuming by induction that $\DALLR{t-1}$ is a correct $(Q_{t-1}, \eps_{t-1}, \rho_{t-1}, \ell, L_{t-1})$-DALLR algorithm, this means that, for any $h_a \in \widehat{H}_t$, when $\DALLR{t-1}$ is run on the column $\cS|_{[n]^{t-1} \times \{h_a\}}$ in \Cref{alg:main} to produce $\cL_a^{(t-1)}$, there is some $\tilde{A}_a^{(t-1)} \in \cL_a^{(t-1)}$ so that 
            \begin{equation}\label{eq:Agood2} \Pr_{i' \in [n]^{t-1}}[ \tilde{A}_a^{(t-1)}(i') = \tilde{c}^{(h_a)}_{i'} ] \geq 1- \eps_{t-1}.\end{equation}
            Since we have chosen $A_a^{(t-1)}$ in \Cref{eq:repA} precisely to maximize
            \[ \Pr_{i' \in [n]^{t-1}}[ {A}_a^{(t-1)}(i') = \tilde{c}^{(h_a)}_{i'} ],\]
            we conclude that in fact \Cref{eq:Agood2} applies to the local algorithms $A_a^{(t-1)}$ for all $a \in \widehat{H}_t$.

            Now, from \Cref{cl:good-col}, we have that $|\widehat{H}_t| \geq (1 - 2\eps_{t-1})|H_t|$ with probability at least $ 1- \eta_t$ over the choice of $\sigma^{(t)}$.  Thus, averaging over both $h \in H_t$ and over $i' \in [n]^{t-1}$, we see that
            \[ \EE_{i' \in [n]^{t-1}} \Pr_{h_a \in H_t}[ A_a^{(t-1)}(i') = \tilde{c}_{(i', h_a)} ] \geq (1 - \eps_{t-1})(1 - 2\eps_{t-1}) \geq 1 - 3\eps_{t-1}.\]
           Thus, by Markov's inequality,
           \begin{align*}
               \Pr_{i' \in [n]^{t-1}} \left[ \Pr_{h_a \in H_t}[A_a^{(t-1)}(i') \neq \tilde{c}_{(i',h_a)} ] > \frac{\delta(C)}{4} \right] \leq \frac{ 12 \eps_{t-1}}{\delta(C)}.
           \end{align*}
           Since $\Pr_{h_a \in H_t}[ A_a^{(t-1)}(i') \neq \tilde{c}_{(i',h_a)}] = \delta( \tilde{c}|_{\{i'\} \times {H_t}}, v)$, this establishes the claim.
         \end{enumerate}
         Thus, we have shown that:
         \begin{enumerate}
             \item Condition (1) of \Cref{def:good-row} holds for at least a $ 1- \eps_{t-1}$ fraction of the rows $i' \in [n]^{t-1}$.
             \item With probability at least $0.9$ over the choice of $\sigma^{(t)}$, Condition (2) of \Cref{def:good-row} holds for at least a $1 - \eps_{t-1}$ fraction of the rows $i' \in [n]^{t-1}$.
             \item With probability at least $1 - \eta_t$ over the choice of $\sigma^{(t)}$, Condition (3) of \Cref{def:good-row} holds for a $1 - \frac{12 \eps_{t-1}}{\delta(C)}$ fraction of the rows $i'\in [n]^{t-1}$.
         \end{enumerate}
         
         Together, these imply that with probability at least $0.9 - \eta_t \geq 1/2$ over the choice of $\sigma^{(t)}$, all three conditions hold for at least  a
         \[ 1 - \frac{12 \eps_{t-1}}{\delta(C)} - \eps_{t-1} - \eps_{t-1} \geq 1 - \frac{ 14 \eps_{t-1}}{\delta(C)} \]
         fraction of rows, using the fact that $\delta(C) \leq 1$.  This proves the claim, using the assumption in \Cref{def:params} that $d_0 \geq 14$.
     \end{proof}

     Now, we put together \Cref{cl:decodegood} and \Cref{cl:goodrows} to establish \Cref{cl:goodAt}.
     
     \begin{proof}[Proof of \Cref{cl:goodAt}]
     \Cref{cl:goodrows} implies that with probability at least $1/2$ over the choice of $\sigma^{(t)}$, there is some $\vec{A} \in \cL_1^{(t-1)} \times \cdots \times \cL_{m_t}^{(t-1)}$ (namely, the one given in \Cref{eq:repA}) so that at least a $1 - \frac{d_0 \eps_{t-1}}{\delta(C)}$ fraction of the rows $i' \in [n]^{t-1}$ are good with respect to $(\tilde{c}, \sigma^{(t)}, \vec{A})$.   Then \Cref{cl:decodegood} implies that for any good row $i' \in [n]^{t-1}$, 
     \[ A^{(t)}(i) = \tilde{c}_i \, \; \forall i \in \{i'\} \times [n].\]
     Together, these imply that with probability at least $1/2$ over the choice of $\sigma^{(t)}$, there is some $\vec{A}$ so that
     \[\Pr_{i' \in [n]^{t-1}}[ A^{(t)}(i) = \tilde{c}_i \, \; \forall i \in \{i'\}\times[n] ] \geq 1 - \frac{ d_0 \eps_{t-1}}{\delta(C)},\]
     which implies that
     \begin{equation*}\Pr_{i \in [n]^t}[ A^{(t)}(i) = \tilde{c}_i ] \geq 1- \frac{d_0 \eps_{t-1}}{\delta(C)}.\end{equation*}
    This proves the claim.
     \end{proof}

      \subsection{Proof of \Cref{lem:base_correct}}\label{app:proofOfCorrectnessClaim_base}

       In this section, we prove \Cref{lem:base_correct}, which we restate for the reader's convenience.

  \basecorrect*

The rest of this section is devoted to the proof of \Cref{lem:base_correct}, the proof is similar to that of \Cref{cl:goodAt}.
    Let $\cS \subseteq {\Sigma \choose \leq \ell}^{n \times n}$ be an input list for $\DALLR{2}$.
    Let $\cL^* = \inset{ c \in C \otimes C : \delta(c, \cS) \leq \rho_2}$ be the ``true'' list of close-by codewords, and let $\ctL^{(2)}$ be the list returned by \DALLR{2}.  We want to show that for any $\tilde{c} \in \cL^*$, there is some local algorithm $A^{(2)} \in \ctL^{(2)}$ so that 
    \begin{equation}\label{eq:perfect_decoding}
    \Pr_{i \in [n]\times [n]}[ A^{(2)}(i) \neq  \tilde{c}_i ] = 0.
    \end{equation}

    To establish \eqref{eq:perfect_decoding}, we begin with a definition.  Fix $\tilde{c} \in \cL^*$.  Say that $h \in [n]$ (which we think of as the index of a column of an element of $\Sigma^{n \times n}$) is \emph{good} (with respect to $\tilde c$) if
    \begin{equation}\label{eq:good_col_two}
    \delta( \tilde{c}|_{[n] \times \{h\}} , \cS|_{[n] \times \{h\}})\leq \rho_0.
    \end{equation}
    Notice that since $\delta(\tilde{c}, \cS) \leq \rho_2$ by assumption, the fraction of columns $h \in [n]$ that are not good is at most $\rho_2/\rho_0$.  Let $H_2 = \Gamma_2(\sigma^{(2)}) \subseteq [n]$ be the set chosen by the sampler $\Gamma_2$ on the seed $\sigma^{(2)} \in \{0,1\}^{r_2}$.  By \Cref{thm:sampler}, along with the fact that we have chosen the parameters for the sampler in Definition \ref{def:params2} 
    to satisfy the hypotheses of that theorem, we conclude that with probability at least $1 - \eta_2$ over the choice of $\sigma^{(2)}$,\footnote{We note that, similarly to the proof of \Cref{cl:goodAt}, the algorithm \DALLR{2} enumerates over all seeds $\sigma^{(2)}$, so there is no actual randomness in the algorithm.  Rather, we imagine choosing a random $\sigma^{(2)}$ for the analysis only.}
    \[ |\Pr_{h \in[n]}[ h \text{ is not good } ] -\Pr_{h \in H_2}[ h \text{ is not good } ]| \leq \nu_2. \]
    Thus, with probability at least $1 - \eta_2$ over the choice of $\sigma^{(2)}$, 
    \begin{equation}\label{eq:E} \Pr_{h \in H_2}[ h \text{ is not good } ] \leq \frac{\rho_2}{\rho_0} + \nu_2 \leq \frac{\delta(C)}{4},\end{equation}
    where in the final inequality we have used the definitions of $\rho_2$ in \Cref{def:params} 
    and $\nu_2 = \frac{\delta(C)}{8}$ in Definition \ref{def:params2}.  Let $\mathcal{E}$ denote the favorable event that \Cref{eq:E} holds.

    Now, consider running the list-recovery algorithm $\cA_0$ on each column $\cS|_{[n] \times \{h_a\}}$ for $h_a \in H_2$, to obtain a list $\cL_a \subseteq C$ for each $a \in [m_2]$.  For each $a \in [m_2]$, define
    \begin{equation}\label{eq:choice_ca} c^{(a)} = \argmin_{c \in \cL_a} \delta(c, \tilde{c}|_{[n] \times \{h_a\}}).\end{equation}
    Define $v = v(\tilde{c}) \in \Sigma^{[n] \times H_2}$ to be the word whose $a$'th column is $c^{(a)}$.  Notice that if $h_a$ is good, $\tilde{c}|_{[n] \times \{h_a\}} \in \cL_a$, and so 
    \begin{equation}
        \label{eq:good_agree}c^{(a)} = v|_{[n] \times \{a\}} = \tilde{c}|_{[n] \times \{h_a\}}.
    \end{equation}
    Now we turn our attention to the rows.  For each $b \in [n]$, let $\mathcal{K}_b \subseteq C$ be the result of running the list-recovery algorithm $\cA_0$ on the row $\cS|_{\{b\} \times [n]}$.  
    \begin{definition}\label{def:good_row_two}
        We say that a row $b \in [n]$ is \emph{good} (with respect to $\tilde{c}$ and $\sigma^{(2)}$) if the following two things occur.
        \begin{enumerate}
            \item $\delta( \tilde{c}|_{\{b\} \times [n]}, \cS|_{\{b\} \times [n]} ) \leq \rho_0$.  Notice that if this occurs then $\tilde{c}|_{\{b\} \times [n]} \in \mathcal{K}_b.$
            \item For all $c \in \mathcal{K}_b \setminus \{\tilde{c}|_{\{b\} \times [n]} \}$, we have $\delta(c|_{H_2}, \tilde{c}|_{\{b\} \times H_2} ) \geq \frac{3\delta(C)}{4}.$
        \end{enumerate}
    \end{definition}
    We next show that, assuming $\cE$ holds, then all good rows are decoded correctly.  Formally, we have the following claim.
    \begin{claim}\label{cl:good_is_good}
    Suppose that the event $\mathcal{E}$ holds.
    Suppose that $b \in [n]$ is good with respect to $\tilde{c}$ and $\sigma^{(2)}$.  Let
    \[ \tilde{c}^{(b)} := \operatorname*\argmin_{c'' \in \mathcal{K}_b} | \{a \in [m_2] : c_{h_a}'' \neq c_b^{(a)} \}| = \operatorname*\argmin_{c'' \in \mathcal{K}_b} \delta( c''|_{H_2}, v|_{\{b\}\times H_2}),\]
    where $c^{(a)}$ and $v$ are as defined above.
    Then $\tilde{c}^{(b)} = \tilde{c}|_{\{b\} \times [n]}$.  That is, if the $c^{(a)}$ we have chosen above are the same as in \Cref{alg:local-alg-two} and $b$ is good, then the row $\tilde{c}^{(b)}$ computed in \Cref{alg:local-alg-two} matches $\tilde{c}$.
    \end{claim}
    \begin{proof}
        Suppose that $b$ is good.  Let $c' = \tilde{c}|_{\{b\} \times [n]}$, so item (1) of \Cref{def:good_row_two} implies that $c' \in \mathcal{K}_b$.  Notice that by \eqref{eq:good_agree}, $c'$ and $v$ agree on all good columns $h_a \in H_2$. Since $\mathcal{E}$ holds, the fraction of good columns is at least $1 - \frac {\delta(C)} 4$, 

and it follows that
        \begin{equation}\label{eq:correct_is_close_two} \delta(c'|_{H_2}, v|_{\{b\} \times H_2}) \leq \frac{\delta(C)}{4}.
        \end{equation}
        On the other hand, let $c'' \in \mathcal{K}_b \setminus \{c'\}$.  Then we have
        \begin{align*}
            \delta(c''|_{H_2}, v|_{\{b\} \times H_2}) &\geq \delta( c''|_{H_2},  c'|_{H_2} ) -  \delta(c'|_{H_2}, v|_{\{b\} \times H_2})\\
            &\geq \frac{3 \delta(C)}{4} - \frac{\delta(C)}{4} \\
            &= \frac{\delta(C)}{2} \\
            &> \delta(c'|_{H_2}, v|_{\{b\} \times H_2}),
        \end{align*}
        where the first line is the triangle inequality and the second line follows from item (2) in \Cref{def:good_row_two} and from \Cref{eq:correct_is_close_two}.  In particular, the argmin in the claim is attained by $c'$, and we conclude that $\tilde{c}^{(b)} = c' = \tilde{c}|_{\{b\} \times [n]}$, which proves the claim.
    \end{proof}
    Now that we know that good rows are decoded correctly (assuming that the $c^{(a)}$ are chosen as in \Cref{eq:choice_ca} and that event $
    \cE$ holds), we will show that most rows are likely to be good.
    \begin{claim}\label{cl:most_rows_good}
        With probability at least $0.9$ over the choice of $\sigma^{(2)}$, at least a $1- \rho_2^{(!)}$ fraction of rows $b \in [n]$ are good with respect to $\tilde{c}$ and $\sigma^{(2)}$.
    \end{claim}
    \begin{proof}
        First, we observe that for any choice of $\sigma^{(2)}$, at least a $1 -\frac {\rho_2 }{\rho_0 }\geq 1-\frac{ \rho_2^{(!)} }{2}$ fraction of the rows satisfies the requirement (1) in \Cref{def:good_row_two}.  Indeed, we have $\delta(\tilde{c}, \cS) \leq \rho_2$ by assumption, so by Markov's inequality 
        \[\Pr_{b \in [n]} [\delta(\tilde{c}|_{\{b\} \times [n]}, \cS|_{\{b\} \times [n]}) > \rho_0 ] \leq \frac{\rho_2}{\rho_0} \leq \frac  {\rho_2^{(!)}} {2},\]
        where the last inequality is by definition of $\rho_2$ in \Cref{def:params}.
        For the requirement (2) in \Cref{def:good_row_two}, fix $b$ and let $c' = \tilde{c}|_{\{b\}\times [n]}$.  Fix any $c'' \in \mathcal{K}_b \setminus \{c'\}$.  From the distance of the code $C$, we have $\delta(c', c'') \geq \delta(C)$.  Thus, by \Cref{thm:sampler} about the properties of the sampler $\Gamma_2$, we have
        \[ \Pr_{\sigma^{(2)}}[ \delta( c'|_{H_2}, c''|_{H_2} ) < {\delta(C)} - \nu_2 ] \leq \eta_2.\]
        Taking a union bound over all $L_0$ elements of $\mathcal{K}_b$ and recalling that $\delta(C) - \nu_2 \geq 3\delta(C)/4$, we conclude that
        \[ \Pr_{\sigma^{(2)}}\left[ \exists c'' \in \mathcal{K}_b \text{ s.t. } \delta(c'|_{H_2}, c''|_{H_2}) < \frac{3 \delta(C)}{4} \right] \leq L_0 \eta_2. \]
        Applying an expectation over $b \in [n]$, we see that
        \begin{align*}
            \EE_{b \in [n] }\Pr_{\sigma^{(2)}}\left[ \exists c''\in \mathcal{K}_b \text{ s.t. } \delta(c'|_{H_2} , c''|_{H_2} ) < \frac{3 \delta(C)}{4} \right] &\leq L_0 \eta_2
            \end{align*}
            and hence, switching the order of the randomness,
            \begin{align*}
            \EE_{\sigma^{(2)}}\Pr_{b \in [n] }\left[ \exists c'' \in \mathcal{K}_b \text{ s.t. } \delta(c'|_{H_2} , c''|_{H_2} ) < \frac{3 \delta(C)}{4} \right] &\leq L_0 \eta_2.
        \end{align*}
        Then by Markov's inequality, we see that with probability at least $1 - \frac{2 L_0 \eta_2 }{\rho_2^{(!)}} \geq 0.9$ over the choice of $\sigma^{(2)}$, we have
        \[ \Pr_{b \in [n] }\left[ \exists c'' \in \mathcal{K}_b \text{ s.t. } \delta(c'|_{H_2} , c''|_{H_2} ) < \frac{3 \delta(C)}{4} \right] \leq \frac{ \rho_2^{(!)}}{2}.\]
        Above, we have used the definition of $\eta_2 = \frac{\rho_2^{(!)}}{20 L_0} $ given in \Cref{def:params2} to bound the failure probability by $0.1$.
        Thus, in the favorable case, item (2) holds for at least a $1 - \frac{ \rho_2^{(!)}} 2$ fraction of the rows $b \in [n]$.  We have already established that item (1) holds for at least a $1 - \frac{\rho_2^{(!)}} 2$ fraction of $b \in [n]$, so we conclude that both hold for at least a $1 - \rho_2^{(!)}$ fraction of the rows $b \in [n]$, proving the claim.
    \end{proof}

    Now using \Cref{cl:good_is_good} and \Cref{cl:most_rows_good}, along with the fact that the event $\mathcal{E}$ happens with probability at least $1 - \eta_2$, we see that with probability at least $0.9 - \eta_2 > 1/2$ over the choice of $\sigma^{(2)}$, if we have a local algorithm $A^{(2)}$ (\Cref{alg:local-alg-two}) that chooses $c^{(a)}$ as in \Cref{eq:choice_ca}, then $\delta( w, \tilde{c}) \leq \rho_2^{(!)}$, where $w \in \Sigma^{n \times n}$ has the  $\tilde{c}^{(b)}$ as rows.  Since \Cref{alg:local-alg-two} then calls $\mathrm{Dec}_{C \otimes C}$ on $w$, and by assumption $\mathrm{Dec}_{C \otimes C}$ can uniquely decode up to radius $\rho_2^{(!)}$, we conclude that, in this case, we would have $x(A^{(2)}) = \tilde{c}$, aka, that \Cref{eq:perfect_decoding} holds.  
    
    Next, we observe that \Cref{alg:main_base} iterates over all seeds $\sigma^{(2)} \in \{0,1\}^{r_2}$ and all advice vectors $\tau$ to choose the columns $c^{(a)} \in \mathcal{L}_a$.  So we conclude that there is some $A^{(2)}$ considered so that \eqref{eq:perfect_decoding} holds.  Finally, we observe that the two tests at the end of \Cref{alg:main_base} will not filter out the representation of such an algorithm $A^{(2)}$ (unless an equivalent local algorithm is already in $\ctL^{(2)}$), so $\rep(A^{(2)})$ appears in $\ctL^{(2)}$ at the end of \Cref{alg:main_base}.

    To finish the proof of \Cref{lem:base_correct}, we need to establish the query complexity and the list size.  The fact that the query complexity is at most $Q_2 = n^2$ is immediate, as there are only $n^2$ entries to query.  The fact that $|\ctL^{(2)}| \leq L_2$ follows from \Cref{cor:HRW}.  In more detail,
    because of the pruning step at the end of \Cref{alg:main_base} and the unique decoding step at the end of \Cref{alg:local-alg-two}, we know that $\ctL^{(2)}$ contains no duplicates, and that for each $\rep(A^{(2)}) \in \ctL^{(2)}$, there is some $c \in \cL^*$ so that $x(A^{(2)}) = c$.   (Here, $\cL^* = \inset{c \in C \otimes C\,:\, \delta(c,\cS) \leq \rho_2}$ is the ``true'' list).  This implies that $|\ctL^{(2)}| \leq |\cL^*|$.  
Since $\rho_2 \leq \rho_0 \kappa^{-8}$ by \Cref{def:params}, \Cref{cor:HRW} implies that $|\cL^*| \leq L^{\kappa^{8} \log^2 L} = L_2,$ so this implies that $|\ctL^{(2)}|\leq L_2$ as well.
    
    This proves \Cref{lem:base_correct}.

\section{Construction of a high-rate DLCC (Proof of \Cref{thm:DLCC})}\label{app:DLCC}
In this section we prove \Cref{thm:DLCC}.  The proof follows closely the approach in \cite{CM25}.  We begin by first stating some ingredients from that paper that we need.  We note that we have adapted these statements to our notation.

We first define the codes that we will use, which are \emph{lifted Reed-Solomon (RS) Codes}~\cite{GKS13}.
\begin{definition}[Lifted RS Codes, \cite{GKS13}]
    Fix a positive integer $m$ and $d$, and a prime power $q$ so that $d \leq q$.  Define
    \[ \mathcal{F}_{m,d,q} = \inset{ f \in \F_q[X_1, \ldots, X_m] : \deg( f(\ell(T)) ) \leq d  \text{ for all lines $\ell: \F_q \to \F_q^m$}}.\]
    Then the lifted RS code with parameters $m,d,q$ is given by
    \[ C = \mathrm{Lift}_m(RS_q(d)) = \inset{ \langle f(\alpha) \rangle_{\alpha \in \F_q^m } \,:\, f \in \mathcal{F}_{m,d,q} }.\]
\end{definition}
That is, a lifted RS codeword is the evaluation vector of a multivariate polynomial $f$ that has the property that the restriction of $f$ to any line is a univariate polynomial of degree at most $d$.  Surprisingly, when the characteristic of $\F_q$ is small, the rate of the lifted RS code can be much higher than the rate of the corresponding $m$-variate Reed-Muller code, as was shown in \cite{GKS13}.  Formally, we have the following theorem.
\begin{theorem}[Rate of Lifted RS Codes, \cite{GKS13}]\label{thm:GKS}
    Suppose that $q = 2^s$ for some integer $s$.  Fix $\eta \in (0,1)$.  Let $C = \mathrm{Lift}_m(RS_q(d))$, and define
    \[ c = (1 + \lceil \log m \rceil ) \cdot 2^{ ( 1 + \lceil \log m \rceil )m} \log(1/\eta).\]
    Suppose that $d = (1 - 2^{-c})q.$  Then the rate of $C$ is at least $1 - \eta$.
\end{theorem}

\begin{remark}[Explicitness]\label{rem:explicit_LRS}
    As stated, it is not clear that $\mathrm{Lift}_m(RS_q(d))$ is explicit, as it is not clear that the set $\mathcal{F}_{m,d,q}$ has a nice description.  However, as part of the proof of \Cref{thm:GKS}, \cite{GKS13} shows that in fact this set does have a nice description, and gives an explicit set of monomials that span a $(1-\eta)q^m$-dimensional subspace of $\mathcal{F}_{m,d,q}$.  So if we take the span of those monomials as a message space, rather than $\mathcal{F}_{m,d,q}$, we obtain a subcode of $\mathrm{Lift}_m(RS_q(d))$ of rate at least $1-\eta$ with all the same properties and an explicit generator matrix (and hence efficient deterministic encoding maps).
\end{remark}

Next, we introduce some machinery from \cite{CM25}.  In particular, they introduce a notion called \emph{improving sets}, and show how to use them obtain deterministic local decoding algorithms.  We note that \cite{CM25} does not explicitly use the terminology ``deterministic local decoding,'' as per \Cref{def:DLC}, but it is implicit in their work.

\begin{definition}[Improving Sets, \cite{CM25}]
    \label{def:improving_set}
    Fix $\rho_1, \rho_2 \in (0,1)$.
    Let $C \subseteq \Sigma^N$ be a code with $\delta(C) \geq 2\rho_1$.  Let $\mathcal{I}$ be a collection of functions $I:\Sigma^N \to \Sigma^N$.  We say that $\mathcal{I}$ is a \emph{$\rho_1$-to-$\rho_2$-improving set} for $C$ if for all $w \in \Sigma^N$ so that there is some $c \in C$ with $\delta(w,c) \leq \rho_1$, we have
    \[ \EE_{I \in \cI} [ \delta(I(w), c) ] \leq \rho_2.\]
    We say that $\cI$ is a \emph{below-$\rho$, factor-$\gamma$ improving set} for $C$ if it is a $\rho'$-to-$\gamma \cdot \rho'$ improving set for $C$ for all $\rho' \leq \rho$.

    If each function $I \in \mathcal{I}$ makes at most $Q$ queries to $w \in \Sigma^N$, we say that $\cI$ has query complexity $Q$.  If every $I \in \mathcal{I}$ runs in time at most $\Tm$ and space at most $\Sp$, we say that $\cI$ has running time $\Tm$ and space use $\Sp$.
\end{definition}

The reason that improving sets are useful is that they can give us deterministic decoding algorithms, and in particular DLCCs.  More precisely, they prove the following lemma.

\begin{remark}[Uniform vs. Non-Uniform]
    The paper \cite{CM25} includes results for both uniform and non-uniform decoding algorithms.  In this paper, we focus on uniform algorithms.  Thus, when quoting results from \cite{CM25}, we include only the uniform statements, and the assumption is that all algorithms are uniform.
\end{remark}

\begin{lemma}[Lemma 3.3 in \cite{CM25}]
\label{lem:imp_means_DLCC}
    Let $C \subseteq \Sigma^N$.  Suppose there is a set $\cI$ that is a below-$\rho$, factor-$\gamma$ improving set for $C$ with query complexity $Q$, running time $\Tm_{\cI}$ and space $\Sp_{\cI}$.
    Suppose that $\gamma < (1 - \gamma)^2$, and define
    \[ b = \left\lceil \frac{ \log( \rho N / 3 + 1 )}{\log( (1-\gamma)^2/\gamma )} \right\rceil.\]

    Then there is a deterministic algorithm $\cP$ that takes as input $w \in \Sigma^N$, runs in time
    \[ \Tmpre_{\cP} = O(N\cdot b \cdot |\cI|^2 Q^{b+1} \Tm_{\cI})\]
    and space
    \[ \Sppre_{\cP} = O( b\cdot \min\{ \Sp_{\cI}, Q \log|\Sigma| \} + b \log|\cI| + \log N + \Sp_{\cI} )\]
    and outputs a $\Spop_{\cP} = O(b \log|\cI|)$-bit description of an algorithm $A$ so that:
    \begin{itemize}
        \item $A$ is a deterministic algorithm which makes at most $Q^b$ queries to $w \in \Sigma^{N}$ and takes as input $i \in [N]$
        \item If there is some $c \in C$ so that $\delta(w,c) \leq \rho$, then $A(i) = c_i$.
    \end{itemize}
    Further, the algorithm $A$ runs in time 
    \[ \Tmeval_{\cP} = O(\Tm_{\cI} \cdot Q^b )\]
    and space
    \[ \Speval_{\cP} = O( b \cdot \min\{ \Sp_{\cI}, Q \log|\Sigma| \} + \Sp_{\cI}).\]
\end{lemma}
In particular, we note that the algorithm $\cP$ guaranteed in \Cref{lem:imp_means_DLCC} is precisely a $(Q, \rho)$-DLCC as in \Cref{def:DLC}, with the reported running times and space complexities.

Next, we quote another result from \cite{CM25} that states that Lifted RS codes have (uniform) improving sets.
\begin{lemma}[Improving Sets for Lifted RS Codes, \cite{CM25}, Lemma 4.21] 
    \label{lem:imp_LRS}
    Let $\F_q$ be a finite field with $q = 2^s \geq 101$.  Fix $\theta > 0$.  Fix integers $h,a$ and let $m = h\cdot a$.  Let $C = \mathrm{Lift}_m(RS_q(d))$, for $d = (1 - 2\theta)q$.  Then there is a below-$(\theta/12)$, factor-$\gamma$ improving set $\cI$ for $C$, where 
    \[ \gamma = O\inparen{ \frac{a^{2a + 1}400^a}{q^a \theta^2} }.\]
    Further, we have $$|\cI| = q^{a(2a + O(1))}\cdot \poly(h),$$ and $\cI$ has query complexity $Q \leq q(q-1)$, and running time and space that satisfiy
    \[ \Sp_{\cI}, \Tm_{\cI} \leq \poly(q^a \cdot h).\]
\end{lemma}

Next, we prove the following theorem; \Cref{thm:DLCC} will follow afterwards as a corollary.

\begin{theorem}[Similar to \cite{CM25}, Lemma 4.22]

\label{thm:DLCC_raw}
Let $\eta \in (0,1)$ be a constant.  Then for any sufficiently large $s$, there is an explicit linear code $C \subseteq \F_q^N$ of rate at least $1 - \eta$ and distance $N^{-o(1)}$, for $q = 2^s$ and $N = q^{\log\log\log(q) \cdot \log\log\log\log(q)}$, so that $C$ is a $(Q, \rho)$-DLCC with $Q = N^{o(1)}$ and $\rho = N^{-o(1)}$.  Further, $C$ has a DLC algorithm with pre-processing time $N^{1 + o(1)}$, pre-processing space, evaluation time, and evaluation space $N^{o(1)}$, and output length $\tilde{O}(\log N)$.
\end{theorem}

\begin{proof}
    Let $\eta \in (0,1)$ be as in the proof statement, and fix an arbitrarily small constant $\xi$.
    Fix $q = 2^s$ as in the theorem statement.
Define
    \[ h = \log\log\log(q) \qquad a = \log(h) = \log\log\log\log(q).\]
    Let $m = a\cdot h$.  Define
    \[ c = (1 + \lceil \log m \rceil ) 2^{(1 + \lceil \log m \rceil ) m } \cdot \log(1/\eta),\]
    and set $\theta = 2^{-c-1}$.  Define $d = (1 - 2\theta)q$.

    Let $C = \mathrm{Lift}_m(RS_q(d)) \subseteq \F_q^N$, where $N = q^m$.  Notice that, as $d = (1 - 2^{-c})q$ for the same choice of $C$ as in \Cref{thm:GKS}, that theorem implies that the rate of $C$ is at least $1 - \eta$.

    Now, by \Cref{lem:imp_LRS}, there is an below-$\rho$, factor-$\gamma$ improving set $\cI$ 
    for $C$, with 
    \[ \rho = \theta/12, \]
    \[ \gamma = O\inparen{\frac{a^{2a + 1} 400^a}{q^a \theta^2}},\]
    size $|\cI| = q^{a(2a + O(1))} \cdot \poly(h)$, query complexity
    $Q \leq q^2$, and running time and space $\Tm_{\cI}, \Sp_{\cI} \leq \poly(q^a \cdot h)$.

    Before plugging this improving set into \Cref{lem:imp_means_DLCC}, we simplify these parameters.  
      We first write $q$ in terms of $N$.  We have
    \[ N = q^m = q^{\log\log\log(q) \cdot \log\log\log\log(q)},\]
    which justifies the value of $N$ in the theorem statement.  From there, we see that 
    \[ \log\log(N) = \log\log(q)(1 + o(1)) \qquad \Rightarrow \qquad \log\log(q) = \log\log(N)(1 - o(1)).\]
    This implies that $q = N^{1/m} = N^{o(1)}$.

   Next, we consider $c$ above.  We have (using the fact that $q$ and hence $m$ is sufficiently large, and $\eta$ is a constant),
    \begin{align}
        c &= ( 1 + \lceil \log m \rceil ) 2^{ (1 + \lceil \log m \rceil ) m } \cdot \log(1/\eta) \notag \\
        &\leq 2^{2m \log m} \cdot \log(1/\eta) \notag\\
        &= 2^{ 2 \log\log\log(q) \log\log\log\log(q) ( \log\log\log\log(q) + \log\log\log\log\log(q) )} \cdot \log(1/\eta) \notag\\
        &\leq 2^{3 \log\log\log(q) ( \log\log\log\log(q))^2 } \cdot \log(1/\eta)\notag \\
        &= (\log\log(q))^{3(\log\log\log\log(q))^2}  \cdot \log(1/\eta) \notag\\
        &= o(\log q) \label{eq:cstarq}\\
        &= o(\log N)\label{eq:cstarN}.
    \end{align}
    In particular, we have
    \begin{equation}\label{eq:theta} \theta = 2^{-c-1} = N^{-o(1)}.
    \end{equation}
    Note that this implies that $\rho = N^{-o(1)}$ as well.

    Now we are ready to apply \Cref{lem:imp_means_DLCC}.  We first observe that for sufficiently large $q$, the value of $\gamma$ above is sufficiently small so that $\gamma < (1 - \gamma)^2$.  Then we conclude that there is a DLC algorithm $\cP$ as in \Cref{lem:imp_means_DLCC}.  To work out the parameters, we first observe that we may bound the parameter $b$ by
    \begin{align*} b &= \left\lceil \frac{ \log(\rho N/3 + 1) }{\log((1-\gamma)^2/\gamma)} \right\rceil \\
    &\leq \frac{2\log(N)}{ \log(1/\gamma)}  \qquad \text{ since $\gamma$ is sufficiently small }\\
    &\leq \frac{ 2ah \log (q) }{ { a \log q - 2\log(1/\theta) - (2a + 1) \log(400 a) - O(1) }} \qquad \text{Definition of $\gamma$, and $N = q^{ah}$} \\
    &\leq \frac{ 2ah \log q }{ a \log q - 2(c+1) - (2a + 1)\log(400 a) - O(1)} \qquad \text{Definition of $\theta$}\\
    &\leq \frac{2 ah \log q }{a \log q - 2( o(\log q) + 1 ) - (2a + 1) \log(400a) - O(1)} \qquad \text{\Cref{eq:cstarq}} \\
    &\leq \frac{ 2ah }{a - o(1) } \\
    &= 2h(1 + o(1)),
    \end{align*}
    where in the final  inequality we have used that $a = \log\log\log\log q$, and hence $a \log a = o(\log q)$.
  
    Finally, we simplify the expressions for $|\cI|, \Tm_{\cI}, \Sp_{\cI}$.  We have
    \begin{align}
        |\cI| &= q^{a(2a + O(1))} \cdot \poly(h) \notag \\
        &= N^{ (2a + O(1))/h } \cdot \poly(h) \label{eq:Isize} \\
        &= N^{o(1)}. \notag
    \end{align}
    We also have
    \begin{align*}
        \Tm_{\cI}, \Sp_{\cI} &\leq \poly( q^a \cdot h ) \\
        &= \poly( N^{1/h} \cdot h ) = N^{o(1)}.
    \end{align*}
    
    Now we work out the parameters from \Cref{lem:imp_means_DLCC}.  We see that $\cP$ has query complexity
    \[ Q^b \leq q^{2b} = N^{o(1)}.\]
    We have already seen that the radius (and so also the distance) is
    \[ \rho = \frac{\theta}{12} \geq N^{-o(1)}\]
    by \Cref{eq:theta}.
    Plugging in $b \leq 2h(1 + o(1))$, we see that the pre-processing time is
    \begin{align*} \Tmpre_{\cP} &= O(N \cdot b \cdot |\cI|^2 Q^{b+1} \cT_{\cI} ) \\
    &=  N \cdot 2h(1 + o(1)) \cdot N^{o(1)} \cdot q^{2(h(1 + o(1)) + 1)} \cdot N^{o(1)} \\
     &=  N^{1 + o(1)}  \cdot q^{2(h(1 + o(1)) + 1)}  \\
    &=  N^{1 + o(1)} \cdot N^{(2h(1+o(1)) + 1 )/m} \\
    &=  N^{1 + o(1)} \cdot N^{ (2 + o(1))/a } \\
    &= N^{1 + o(1)}
    \end{align*}
    Similarly, the pre-processing space is 
    \begin{align*}
        \Sppre_{\cP} &= O\inparen{b \cdot \min\inset{ \Sp_{\cI}, Q \log q } + b \log|\cI| + \log N + \Sp_{\cI}} \\
        &\leq 2h(1 + o(1)) \cdot N^{o(1)} + 2h(1 + o(1)) \cdot o(1) \log N + \log N + N^{o(1)} \\
        &= N^{o(1)}.
    \end{align*}
    The evaluation time is
    \begin{align*}
        \Tmeval_{\cP} &= O\inparen{\Tm_{\cI} \cdot Q^b }\\
        &\leq N^{o(1)} \cdot q^{4h(1 + o(1))} \\
        &= N^{o(1)} \cdot N^{ 4h(1 + o(1))/m } \\
        &= N^{o(1)}
    \end{align*}
    and the evaluation space is
    \begin{align*}
        \Speval_{\cP} &= O\inparen{b \cdot \min\inset{ \Sp_{\cI}, Q\log q} + \Sp_{\cI}} = N^{o(1)},
    \end{align*}
    where the final equality follows exactly as in the computation of $\Sppre_{\cP}$.
    Finally, we observe that the output length is
    \begin{align*} \Spop_{\cP} &= O(b \log|\cI|) \\
    &= O\inparen{h \cdot \inparen{\inparen{ \frac{2a + O(1)}{h} } \log(N) + \log(h)} } \qquad \text{by \Cref{eq:Isize} and the fact that $b = O(h)$} \\
    &= O(\log(N) \log\log\log\log(N) ) = \tilde{O}(\log N).
    \end{align*}
    This completes the proof of the theorem.
\end{proof}

Finally, we are ready to prove \Cref{thm:DLCC}, which follows as a corollary from \Cref{thm:DLCC_raw}.  We restate the theorem for the reader's convenience.
\DLCCthm*

\begin{proof}
    Fix $q = 2^a$ as in the theorem statement.  Now, for each sufficiently large integer $b$, we apply \Cref{thm:DLCC_raw} with $s =a\cdot b$, resulting in a linear code $C$ with alphabet size $q' = 2^{ab} = q^b$, and with block length $N = (q')^{\log\log\log(q') \cdot \log\log\log\log(q')}$, which is a $(Q,\rho)$-DLCC for $Q = N^{o(1)}, \rho = N^{-o(1)}$; and that has a DLC algorithm $\cP$ with $\Tmpre_{\cP} = N^{1 + o(1)}$ and $\Tmeval_{\cP}, \Sppre_{\cP}, \Speval_{\cP} = N^{o(1)}$ and $\Spop = \tilde{O}(\log N)$.  We observe that this choice of $N$ satisfies $q' = N^{o(1)}$.  Finally, we observe that if $C$ if $\F_{q'}$-linear, then it is also $\F_q$-linear, as $\F_q \leq \F_{q'}$, since $a | s$ by construction.

    Since there are an infinite number of choices of $b$, we get an infinite number of choices of $N$ this way, and this proves the theorem.
\end{proof}

\section{Construction of a good base code (Proof of \Cref{thm:basecode})}\label{app:baseCode}
In this section we build the base code $C$ to instantiate \Cref{thm:mainLocal} and prove \Cref{thm:basecode}.

The starting point for our code will be the construction of \cite{ST25}.  They show that a code resulting from the Alon-Edmunds-Luby construction (AEL, \cite{AEL95}) can be list-recovered to capacity in near-linear randomized time, or polynomial deterministic time.  As we are interested in deterministic algorithms, we will use the polynomial-time version.  However, we cannot simply use this code as a black box, because the code is $\F_q$ linear with an alphabet of $\F_q^d$, for some $d > 1$.  In order for our tensor machinery to work, the code $C$ must be linear over its alphabet.  Thus, we concatenate the code from \cite{ST25} with a small list-recoverable inner code, so that the alphabet becomes $\F_q$.  Since $d$, the length of this inner code, is a constant, we can find such a code and decode it by brute force.  Thus, for this inner code, we use a recent result about the list-recoverability of random linear codes from \cite{GG26}.

To begin, we state the building blocks that we will be using.  We start with the AEL code from \cite{ST25}.\footnote{We note that we have changed the statement slightly from Corollary 4.16 in \cite{ST25}  That corollary is stated as a randomized, near-linear time algorithm.  However, as noted in \cite{ST25}, the only randomized component in it is the algorithm for finding regular factors of an expander graph, given in \cite[Lemma 3.9]{ST25}.  That lemma states that this can be done in randomized time $\tilde{O}(|E(G)|)$, or deterministic time $O(|E(G)|^{3.5})$.  Thus, we replace the near-linear contribution to the running time in \Cref{thm:ST25} below with a polynomial contribution, to obtain a deterministic algorithm.  See the second bullet point in ``Remarks'' on Page 9 of \cite{ST25}.}

\begin{theorem}[\cite{ST25}; see \cite{ST25_arxiv}, Corollary 4.16]\label{thm:ST25}
    For all $R, \zeta \in (0,1)$, and for all sufficiently large $n$,
    there is an explicit $\F_q$-linear code $C_{AEL} \subseteq (\F_q^d)^{n}$ so that the following hold.
    \begin{enumerate}
        \item The rate of $C_{AEL}$ is at least $R$.
        \item The distance of $C_{AEL}$ satisfies $\delta(C_{AEL}) \geq 1 - R - \zeta$
        \item For all $\cS \subseteq {\F_q^d \choose \leq \ell }^{n}$, $C_{AEL}$ is $(1 - R - \zeta, \ell, L_0)$-list-recoverable with 
        \[ L_0 \leq \exp\inparen{\exp\inparen{O\inparen{\frac{\ell}{\zeta} \log\inparen{\frac{\ell}{\zeta}}}}}\]
        \item $C_{AEL}$ has a deterministic list-recovery algorithm to perform the above that runs in time $O_{\ell, \zeta}\left((d\cdot n)^{3.5}\right)$.
        \item We may take $d = \inparen{\frac{\ell}{\zeta}}^{O(\ell/\zeta)}$ and $q$ to be any prime power at least $\ell^{O(1/\zeta)}$.
    \end{enumerate}
\end{theorem}

\begin{remark}[Legitimate values of $n$]

    In \cite{ST25}, \Cref{thm:ST25} is stated for infinitely many $n$, rather than for any sufficiently large $n$.  However, an inspection of the proof shows that the requirement on $n$ is that there exists an explicit bipartite expander graph $G = (L,R,E)$ with $|L| = |R| = n$,  with normalized expansion parameter\footnote{For a bipartite graph $G = (L,R,E)$ with $|L| = |R| = n$, let $A_G$ be the normalized adjacency matrix of $G$, so the rows are indexed by $L$ and the columns by $R$, and each nonzero entry is $1/d$.  Then the normalized expansion parameter $\lambda$ is defined as the second-largest singular value of $A_G$.} $\lambda \leq 2^{-b_0 \cdot \ell/\zeta \log(\ell/\zeta)}$ for some absolute constant $b_0$, and of degree $d = 2^{O(\ell/\zeta \log(\ell/\zeta)))}$.   In fact, we can obtain such an expander for any sufficiently large $n$.  In more detail, \cite[Theorem 1.2]{Alo21} states that for any prime $p \equiv 1$ mod $4$, and for \emph{all} large enough $n$, there is a strongly explicit construction of an expander graph $G'$ on $n$ vertices with degree $d = p+2$ and normalized expansion parameter $\lambda$ satisfying $\lambda < (1+ \sqrt{2})/\sqrt{d} + o(1)$.\footnote{For a non-bipartite graph $G' = (V,E)$, let $A_{G'}$ be the normalized adjacency matrix, which is a symmetric matrix with all nonzero entries equal to $1/d$.  The expansion parameter $\lambda$ is defined to be the second largest eigenvalue in magnitude of $A_{G'}$.}  
    
    Thus, in the proof of \Cref{thm:ST25}, we first fix $\lambda \leq 2^{-b_0 \cdot (\ell/\zeta) \log(\ell/\zeta)}$.  Then we choose $d = 2^{O((\ell/\zeta) \log(\ell/\zeta)}$ so that (a) $d = p+2$ for some prime $p \equiv 1$ mod $4$, and (b) so that $d \geq 2\inparen{\frac{1 + \sqrt{2}}{\lambda}}^2$.  To see why we can do this, first set $\phi = (\ell/\zeta) \log(\ell/\zeta)$.  Then for (a), we are guaranteed that there is at least one prime $p$ between $2^{10 \cdot b_0 \phi}$ and $2^{10 \cdot b_0 \phi + 1}$ that is equal to $1$ mod $4$ (this follows from, e.g., \cite{CH12}); so fix this $p$ and choose $d = p+2$, so $d = 2^{O(\phi)}$, as desired.  For (b), we then observe that with this choice, we have
    \[ d \geq 2^{10 b_0 \phi} = (1/\lambda)^{10} \geq 2\cdot \inparen{\frac{1 + \sqrt{2}}{\lambda}}^2.\]
    Now we invoke the result of \cite{Alo21}, and conclude that there is some graph $G'$ on exactly $n$ vertices of degree $d$ with expansion parameter 
    \[ \lambda' < \frac{1 + \sqrt{2}}{\sqrt{d}} + o(1) < \lambda,\]
    where in the last line we have used that $d \geq 2\inparen{\frac{1 + \sqrt{2}}{\lambda}}^2$ implies that $\lambda \geq \sqrt{2}\inparen{\frac{1 + \sqrt{2}}{\sqrt d}}$.  Thus, we conclude that $G'$ is also an expander with parameter $\lambda$.  The final step is to let $G$ be the double-cover of $G'$\footnote{Given a graph $G' = (V,E')$, the double-cover $G = (L,R,E)$ is a biparate graph where both $L$ and $R$ are a copy of $V$; and so that for all $\{u,v\} \in E'$, there are two edges $\{u_L, v_R\}, \{v_L,u_R\} \in E$.  Observe that the adjacency matrix for $G$ is the same as for $G'$, and in particular they have the same expansion parameter.} to obtain a bipartite graph with expansion factor $\lambda \leq 2^{-b_0 \phi}$, degree $d = 2^{O(\phi)}$, and with exactly $n$ vertices on each side, as desired.
\end{remark}
We will also use the following result from \cite{GG26}.  (To obtain the version below, we have substituted $R \gets 1 - 2\zeta$ and $\eps, \eps_0 \gets \zeta$ into the version from \cite{GG26}).
\begin{theorem}[\cite{GG26}, follows from Theorem 4.7]\label{thm:GG26}
    Let $\zeta > 0$, and fix $\ell \geq 1$.  Let $q$ be a prime power and $d'$ be a positive integer.  Let $C_{RLC} \subseteq \F_q^{d'}$ be a random linear code of rate $1 - 3\zeta$.  
    Suppose that 
    \[ L' \leq \inparen{ \frac{\ell}{1 - \zeta} }^{(1 - \zeta)/\zeta}\]
    and that
    \[ q \geq (2L')^{4/\zeta} \qquad d' \geq \frac{(L')^2}{\zeta}.\]
    Then $C$ is $(\zeta, \ell, L')$-list-recoverable with high probability.
\end{theorem}

Now we put these together to obtain our list-recoverable base code.
We restate \Cref{thm:basecode} for the reader's convenience.
\basecode*
\begin{proof}
    Given $\zeta, \ell$, fix parameters
    \begin{align*} L' &= \inparen{ \frac{\ell}{1 - \zeta} }^{(1-\zeta)/\zeta} \leq (2\ell)^{1/\zeta},\\
     q & \geq q_0 := (L')^{O(1/\zeta)} = \ell^{O(1/\zeta^2)},\\
      d' &= \inparen{ \frac{L'}{\zeta} }^{O(L'/\zeta)} = \ell^{O\inparen{ \frac{ (2\ell)^{1/\zeta}}{\zeta^2}}},
      \end{align*}
    and
    \[ d = d'\cdot(1 - 3\zeta) =\inparen{ \frac{L'}{\zeta} }^{O(L'/\zeta)} = \ell^{O\inparen{ \frac{ (2\ell)^{1/\zeta}}{\zeta^2}}}. \]
    Choose $n$ to be a multiple of $d'$ (which is the value of $d_0$ in the theorem statement).
    Let $C_{AEL} \subseteq (\F_q^d)^{n'}$ be the code from \Cref{thm:ST25}, with length $n' := n/d'$ and rate $R_{AEL} = 1 - 2\zeta.$  Then \Cref{thm:ST25} implies that $\delta(C_{AEL}) \geq \zeta$, and that $C_{AEL}$ is $(\zeta, L', L_0)$-list-recoverable with
    \[ L_0 = \exp\inparen{ \exp\inparen{\frac{L'}{\zeta} \log\inparen{\frac{L'}{\zeta}}}} = \exp\inparen{\exp\inparen{ \ell^{O\inparen{(2\ell)^{1/\zeta}/\zeta^2}}}}\]

    Let $C_{RLC} \subseteq \F_q^{d'}$ be a random linear code of rate $1 - 3\zeta$.  \Cref{thm:GG26} implies that with high probability, $C_{RLC}$ is $(\zeta, \ell, L')$-list-recoverable with high probability, noting that our choice of $d'$ satisfies $d' \geq (L')^2/\zeta$ for sufficiently small $\zeta$.  Note that in time $q^{O((d')^2)},$ we may do a brute-force search for $C_{RLC}$ that is appropriately list-recoverable.  From now on, suppose that $C_{RLC}$ is indeed $(\zeta, \ell, L')$-list-recoverable.
    
    Let $C \subseteq \F_q^n$ be the concatenation of $C_{AEL}$ an $C_{RLC}$, where $n = n'd'$.  That is, 
    \[ C = \inset{ (E_{RLC}(c_1), E_{RLC}(c_2), \ldots, E_{RLC}(c_{n'}))\,:\, c \in C_{AEL} },\]
    where $E_{RLC}: \F_q^{d} \to \F_q^{d'}$ is an encoding map for $C_{RLC}.$  Now, we claim that $C$ is $(\rho_0, \ell, L_0)$-list-recoverable, for 
    \[ \rho_0 = \Omega(\zeta^2) \]
    and $L_0$ as above.  Indeed, consider the following list-recovery algorithm, which takes as input a collection of input lists $\cS \subseteq { \F_q \choose \leq\ell }^n$.
    \begin{enumerate}
        \item For each $i \in [n']$, list-recover the $i$'th inner code $C_{RLC}$ by brute force to get a list $S'_i$.  This takes time 
        $n' \cdot q^{O(d')}. $ 
        \item List-recover $C_{AEL}$ with the input lists $\cS' = (S'_1, \ldots, S'_{n'})$ to obtain a list $\cL$ of size at most $L_0$.  This takes time $O_{\ell, \zeta}((d'\cdot n')^{3.5}) = O_{\ell, \zeta}((n')^{3.5})$.
        \item Return $\cL$.
    \end{enumerate}
    To see that the algorithm above is correct, assume that there is a codeword $c \in C$ so that $\delta(c,\cS) \leq \zeta^2$.  Write $c = (c^{(1)}, c^{(2)}, \ldots, c^{(n')})$, where each $c^{(i)} \in C_{RLC}$, and break up the input lists $\cS$ correspondingly, so 
    \[ \EE_{i \in [n']} \delta( c^{(i)}, \cS^{(i)} ) \leq \zeta^2.\]
    By Markov's inequality, for at least a $1 - \zeta$ fraction of the $i \in [n']$, we have
    \[ \delta(c^{(i)}, \cS^{(i)} ) \leq \zeta.\]
    For these $i$-s, the first step of the algorithm above is successful, so $c^{(i)} \in S_i'$.  Thus, 
    \[ \EE_{i \in [n']} \mathbf{1}[ c^{(i)} \not\in S_i' ] \leq \zeta,\]
    which implies that running $C_{AEL}$'s list-recovery algorithm in step 2 is successful, namely that $c \in \cL$.  Thus the algorithm is correct, and this implies that $C$ is $(\zeta^2, \ell, L_0)$-list-recoverable in time $$O_{\ell, \zeta}(n^{3.5} +  n q^{O(d')}) = O_{\ell,\zeta}\inparen{n^{3.5} + n q^{O_{\ell,\zeta}(1)}}.$$
    Trivially, the same bound applies to the space.

    Finally, we observe that the rate of $C$ is
    \[ R = R_{AEL} \cdot R_{LRC} = (1 - 2\zeta)(1 - 3\zeta) = 1 - O(\zeta),\]
    and the distance is
    \[ \delta(C) \geq \delta(C_{AEL}) \delta(C_{RLC}) \geq \zeta \cdot \zeta = \zeta^2.\] Above, we have used the fact that a random linear code of rate $1 - 3\zeta$ over an alphabet of size $\exp(\Omega(\zeta^{-2}))$ has distance at least $\zeta$ with high probability (see, e.g., \cite{ECT}, Theorem 4.2.1 and Proposition 3.3.4).
\end{proof}

\end{document}